\documentclass[12pt]{article}
\usepackage{amsmath}
\usepackage{bm}
\usepackage{graphicx,psfrag,epsf}
\usepackage{subfig}
\usepackage{enumerate}
\usepackage{natbib}
\usepackage{cite}
\usepackage{hyperref}
\usepackage{url} % not crucial - just used below for the URL 
\usepackage{amssymb}
\usepackage{amsmath,amssymb,bm}
\usepackage{adjustbox}
\usepackage{multicol}
\usepackage{multirow}
\usepackage{color}
\usepackage{comment}
\newcommand{\blind}{1}

\newtheorem{theorem}{Theorem}[section]

\newtheorem{corollary}[theorem]{Corollary}

\newtheorem{lemma}{Lemma}[section]
\newtheorem{remark}{Remark}[section]

\newenvironment{proof}[1][Proof]{\noindent\textbf{#1.} }{\ \rule{0.5em}{0.5em}}
\allowdisplaybreaks

\begin{document}

\def\spacingset#1{\renewcommand{\baselinestretch}%
{#1}\small\normalsize} \spacingset{1}

%%%%%%%%%%%%%%%%%%%%%%%%%%%%%%%%%%%%%%%%%%%%%%%%%%%%%%%%%%%%%%%%%%%%%%%%%%%%%%

\if1\blind
{
  \title{\bf Comparing Misspecified Models with Big Data: A Variational Bayesian Perspective}
  \author{Yong Li\thanks{
    Li gratefully acknowledges \textit{the financial support of the Chinese Natural
Science fund (No.72273142, 72394392).} Yong Li, School of Economics, Renmin University
of China, Beijing, 1000872, China.}\hspace{.2cm}\\
Renmin University of China\\
gibbsli@ruc.edu.cn\\
\and
Sushanta K. Mallick \thanks{Sushanta K. Mallick, School of Business and Management, Queen Mary University of London}
\hspace{.2cm}\\
Queen Mary University of London\\
s.k.mallick@qmul.ac.uk\\
\and
    Tao Zeng\thanks{Tao gratefully acknowledges \textit{the financial support of the National Natural
Science Foundation of China (No.72073121).} Tao Zeng, School of Economics, Academy of Financial Research and Beijing Research Center, Zhejiang
University, Zhejiang, 310058, China.}\\
    Zhejiang University \\
ztzt6512@gmail.com\\
\and
Junxing Zhang\thanks{Junxing Zhang, School of Economics,
Renmin University of China, Beijing, 1000872, China.}\\
Renmin University of China\\
zjx0316a@ruc.edu.cn}
  \maketitle
} \fi

\if0\blind
{
 \title{\bf  Comparing Misspecified Models with Big Data, A Variational Bayesian Perspective}
\maketitle
%\begin{center}
    %{\LARGE\bf Title}
%\end{center}
  \medskip
} \fi

%\bigskip
\begin{abstract}
Optimal data detection in massive multiple-input multiple-output (MIMO) systems often requires prohibitively high computational complexity. A variety of detection algorithms have been proposed in the literature, offering different trade-offs between complexity and detection performance. In recent years, Variational Bayes (VB) has emerged as a widely used method for addressing statistical inference in the context of massive data. This study focuses on misspecified models and examines the risk functions associated with predictive distributions derived from variational posterior distributions. These risk functions, defined as the expectation of the Kullback-Leibler (KL) divergence between the true data-generating density and the variational predictive distributions, provide a framework for assessing predictive performance. We propose two novel information criteria for predictive model comparison based on these risk functions. Under certain regularity conditions, we demonstrate that the proposed information criteria are asymptotically unbiased estimators of their respective risk functions. Through comprehensive numerical simulations and empirical applications in economics and finance, we demonstrate the effectiveness of these information criteria in comparing misspecified models in the context of massive data.

\end{abstract}

\noindent%

%\newpage
\spacingset{1.9} % DON'T change the spacing!

\section{Introduction}
\label{sec:intro}

In numerous empirical studies, parametric models are commonly employed. However, parametric models inherently carry the risk of model misspecification. As George Box famously stated, ``All models are wrong, but some are useful.'' When a model is misspecified, it can result in inefficient or, in some cases, inconsistent estimation of key parameters. Furthermore, likelihood-based statistical inferences, such as hypothesis testing and goodness-of-fit assessments, are significantly affected. Therefore, developing robust methods to address model misspecification is of critical importance.

Model comparison is one of the most critical issues in statistical inference. For a partial list of studies, see \citet{granger1995comments}, \citet{phillips1994posterior}, \citet{phillips1995bayesian,phillips1996econometric}, \citet{hansen2005test}, and \citet{burnham2008model}. There are essentially two strands of literature on model selection \citep{VehtariandOjanen,anderson2004model}. The first strand aims to answer the question which model best explains the observed data. The Bayes factor (BF, \citealp{kass1995bayes}) and its variations belong to this strand. They compare models by examining “posterior probabilities” given the observed data and search for the “true” model. Bayes Information Criterion (BIC, \citealp{schwarz1978estimating}) is a large sample approximation to BF, although it is based on the maximum likelihood estimator (MLE). The second strand comes from a predictive perspective, answering the question which model gives the best predictions of future observations, which are generated by the same mechanism that gives the observed data. From the predictive perspective, many penalty-based information criteria have been proposed for model comparison. In the frequentist framework, the two most popular information criteria are the Akaike Information Criterion (AIC) proposed by \citet{Akaike1973} and the Takeuchi Information Criterion (TIC) introduced by \citet{takeuchi1976distribution}. Both are asymptotically unbiased estimators of the expected Kullback-Leibler (KL) divergence between the data generating process (DGP) and the plug-in predictive distribution when the MLE is used. The plug-in predictive distribution is obtained by substituting parameter values with their optimal estimates to produce the plug-in estimated sampling distribution. The AIC assumes that all candidate models either nest the true model or are good approximations of the DGP, whereas the TIC allows for model misspecification, with its penalty term involving the inverse of the Hessian matrix. Under the Bayesian framework, Deviance Information Criterion (DIC), proposed by \citet{spiegelhalter2002bayesian}, is one of the most popular penalty-based predictive information criteria. In a recent study, \citet{li2020deviance} developed a variant of DIC for comparing misspecified models, while \citet{li2024deviance} proposed a decision-theoretic interpretation of DIC, demonstrating that DIC is the Bayesian version of AIC.

In recent years, several model selection approaches utilizing the Variational Bayes (VB) method have been introduced. A common strategy in VB-based model selection is to use the evidence lower bound (ELBO) as a proxy for the logarithm of the marginal likelihood function, \(\log p(\mathbf{y})\), to perform Bayes factor (BF) comparisons.  \citet{corduneanu2001variational} investigated VB model selection in the context of mixture models, and used the ELBO as a proxy to determine the optimal number of components. \citet{you2014variational} explored the application of VB to classical Bayesian linear models. They established that, under mild regularity conditions, VB-based estimators possess desirable frequentist properties, such as consistency. Additionally, they proposed two VB-specific information criteria: the Variational AIC (VAIC), which substitutes the VB posterior mean into the DIC, and the Variational Bayesian Information Criterion (VBIC), which uses the ELBO as a proxy for the marginal likelihood. They further showed that VAIC is asymptotically equivalent to the frequentist AIC, while VBIC is first-order equivalent to the BIC in linear regression. \citet{Zhang_2024} proposed using the ELBO as an alternative criterion for model selection and demonstrated its asymptotic equivalence to the BIC. However, in the context of misspecified models and the era of massive data, there has been relatively little research on Bayesian model selection from a predictive perspective. This gap highlights the need for further investigation into model selection methodologies that prioritize predictive performance in such settings.

\begin{comment}
In this paper, for misspecified models with massive data, we propose two new
penalty-based predictive information criterions for model comparison. First,
based on variational posterior distribution, we show that from predictive
viewpoint, two kinds of predictive distributions, that is, variational
plug-in predictive distribution and variational posterior predictive
distribution can be obtained. Second, we investigate the risk functions for these two variational predictive distributions which are the expectation of the K-L divergence between the true density and the predictive distributions.  Thrid, under some regularity conditions, we derive that these proposed two new information criterion are respectively
asymptotically unbiased estimator of the corresponding risk function. 
At last, based on simulation and real studies, we illustrate the proposed information criterion using some popular examples in economics and finance. 
\end{comment}

In this paper, we propose two new penalty-based predictive information criteria for model comparison in the context of misspecified models with massive data. First, based on the variational posterior distribution, we demonstrate that, from a predictive perspective, two types of predictive distributions can be derived: the variational plug-in predictive distribution and the variational posterior predictive distribution. Second, we examine the risk functions associated with these two variational predictive distributions, defined as the expectations of the KL divergence between the DGP and the predictive distributions. Third, under certain regularity conditions, we establish that the proposed information criteria are asymptotically unbiased estimators of their corresponding risk functions. Finally, through simulations and real-world case studies, we illustrate the application of the proposed information criteria.

The paper is organized as follows. Section \ref{sec:review} briefly reviews the literature
on how to make statistical inferences about misspecified models and VB
technique for misspecified models with massive data. Section \ref{sec:risk} investigates
the risk functions of variational predictive distributions. Section \ref{sec:model selection}
introduces the statistical decision theory and proposes the new
penalized-based information criterion to compare misspecified models with
massive data. Section \ref{sec:simu and real} illustrates the new methods using two simulated big
data and two real big data. Section \ref{sec:conclude} concludes the paper. The Appendix
collects the proof of the theoretical results and VB analytical expression
of parametric models used in the paper.
\section{Statistical Inference for Misspecified Models: A Review} \label{sec:review}

\subsection{MLE-based Inference under Model Misspecification}

%White (1982) examined the consequences and detection of model misspecification when ML is used
%for estimation and inference. 

Let the observed data be $\mathbf{y}=\left( y_1,\cdots
,y_{n}\right)$, with an i.i.d. data generating process (DGP) denoted by $g(\mathbf{y})$.
Consider a parametric model, denoted by $p(\mathbf{y}|\mbox{\boldmath${%
\theta}$}) $ used to fit the data, where $\mbox{\boldmath${%
\theta}$}$ is a $P$-dimensional parameter, and $\mbox{\boldmath${%
\theta}$}\in {\mbox{\boldmath${\Theta}$}}\subseteq R^{P}$. The Kullback-Leibler (KL) divergence is used to measure the ``distance'' between $g(\mathbf{y})$ and $p(\mathbf{y}|{\mbox{\boldmath${\theta}$}})$, that is, 
\begin{eqnarray*}
& &KL[g(\mathbf{y}),p(\mathbf{y}|{\mbox{\boldmath${\theta}$}})]=\int g(%
\mathbf{y})\ln \frac{g(\mathbf{y})}{p(\mathbf{y}|{\mbox{\boldmath${\theta}$}}%
)}d\mathbf{y} \\
&=&E_{g(\mathbf{y})}\ln g(\mathbf{y})-E_{g(\mathbf{y})}\ln p(\mathbf{y}|{\mbox{\boldmath${%
\theta}$}}),
\end{eqnarray*}
where $E_{g(\mathbf{y})}$ is with respect to the DGP $g(\mathbf{y})$.
Let ${\mbox{\boldmath${\theta}$}}^{*}$ $\in \Theta \subset R^{p}$ the pseudo
true value that minimizes the KL divergence 
\begin{equation*}
{\mbox{\boldmath${\theta}$}}^{*}=\arg \min_{\theta }KL({\mbox{\boldmath${%
\theta}$}})=\arg \max_{\theta }E_{g(\mathbf{y})}\ln p(\mathbf{y}|{\mbox{\boldmath${%
\theta}$}}),
\end{equation*}
and ${\mbox{\boldmath${\hat{\theta}}$}}$ denoted as the quasi maximum likelihood (QML) estimator
of ${\mbox{\boldmath${\theta}$}}$, which maximizes the log-likelihood
function of the parametric model, 
\begin{equation*}
{\mbox{\boldmath${\hat{\theta}}$}}=\arg \max_{\theta }\ln p(\mathbf{y}|{%
\mbox{\boldmath${\theta}$}}).
\end{equation*}

For simplicity, let 
%$\mathbf{y}=(y_1,y_2,\cdots,y_n)^{\prime}$, and 
$l_t\left(\mathbf{y}_t,\boldsymbol{\theta}\right)=\ln p\left(\mathbf{y}_t | \boldsymbol{\theta}\right)$ represent the conditional log-likelihood for the $t^{t h}$ observation for any $1 \leq t \leq n$. We suppress $l_t\left(\mathbf{y}_t, \boldsymbol{\theta}\right)$ as $l_t(\boldsymbol{\theta})$, so that the log-likelihood function $\ln p(\mathbf{y} | \boldsymbol{\theta})$ is expressed as $\sum_{t=1}^n l_t(\boldsymbol{\theta})$. Define $\nabla^j l_t(\boldsymbol{\theta})$ as the $j^{t h}$ order derivative of $l_t(\boldsymbol{\theta})$ and $\nabla^j l_t(\boldsymbol{\theta})=l_t(\boldsymbol{\theta})$ when $j=0$. Let $\mathbf{\hat{J}}({\mbox{\boldmath${\theta}$}})=\frac{1}{n}\sum_{t=1}^{n}\nabla l_t(\boldsymbol{\theta})\nabla l_t(\boldsymbol{\theta})^{\prime }-\frac{1}{n}\sum_{t=1}^{n}\nabla l_t(\boldsymbol{\theta})\sum_{t=1}^{n}\nabla l_t(\boldsymbol{\theta})^{\prime}$, $\mathbf{\hat{I}}({\mbox{\boldmath${\theta}$}})=-\frac{1}{n}%
\sum_{t=1}^{n}\nabla^2 l_t(\boldsymbol{\theta})$. White (1982) established
the maximum likelihood (ML) theory for misspecified models, that is, 
\begin{equation}
\left(\mathbf{\hat{I}}^{-1}({\mbox{\boldmath${\hat{\theta}}$}})\mathbf{\hat{J}}({\mbox{\boldmath${\hat{\theta}}$}})\mathbf{\hat{I}}^{-1}({\mbox{\boldmath${\hat{\theta}}$}})\right)^{-1/2} \sqrt{n}({\mbox{\boldmath${\hat{\theta}}$}}-{\mbox{\boldmath${\theta}$}}^{*})\overset{d}{\rightarrow }N\left({\mbox{\boldmath${0}$}},\mathbf{I}\right) ,  \label{white}
\end{equation}%
as n goes to infinity where the asymptotic variance takes the sandwich form. If the model is
correctly specified, then 
\begin{equation}
\left(\mathbf{\hat{I}}^{-1}( {\mbox{\boldmath${\hat\theta}$}})\right)^{-1/2} \sqrt{n}({\mbox{\boldmath${\hat{\theta}}$}}-{\mbox{\boldmath${\theta}$}}^*)%
\overset{d}{\rightarrow }N\left(\mathbf{0},\mathbf{I}\right) .  \label{ml}
\end{equation}
as n goes to infinity.

\subsection{Bayesian Inference under Model Misspecification}

Consider a statistical model indexed by a set of $P$ parameters, $\boldsymbol{\theta} \in \boldsymbol{\Theta} \subseteq \mathbb{R}^P$, with a prior distribution $p(\boldsymbol{\theta})$ defined over $\boldsymbol{\theta}$. By applying Bayes' theorem, the posterior distribution can be expressed as:
\begin{equation}
p(\boldsymbol{\theta}|\mathbf{y}) = \frac{p(\mathbf{y}|\boldsymbol{\theta})p(\boldsymbol{\theta})}{p(\mathbf{y})} \propto p(\boldsymbol{\theta})p(\mathbf{y}|\boldsymbol{\theta}), \label{post01}
\end{equation}
where $p(\mathbf{y}) = \int p(\mathbf{y}|\boldsymbol{\theta})p(\boldsymbol{\theta})d\boldsymbol{\theta}$ represents the marginal likelihood.

In most cases, the posterior distribution $p(\boldsymbol{\theta}|\mathbf{y})$ does not have a closed-form solution. Consequently, posterior sampling is typically conducted using Markov Chain Monte Carlo (MCMC) techniques \citep{gelman2003bayesian}. Based on the random samples generated from posterior simulations, Bayesian statistical inference can be performed using the corresponding sample means and covariance matrices. For example, let $\{\boldsymbol{\theta}^{(j)} : j = 1, 2, \cdots, J\}$ denote the effective random samples generated from the posterior distribution after discarding burn-in samples. Bayesian estimates of $\boldsymbol{\theta}$ and the associated standard error can then be calculated as:
$
	\boldsymbol{\bar\theta}=\frac {1}{J}\sum_{j=1}^J\boldsymbol{\theta}^{(j)},
    	\widehat{Var(\boldsymbol{\theta}|\mathbf{y})}
	=\frac{1}{J-1}\sum_{j=1}^J(\boldsymbol{\theta}^{(j)}-\boldsymbol{\bar{\theta}})
	(\boldsymbol{\theta}^{(j)}-\boldsymbol{\bar{\theta}})^{\prime}.
$

These Bayesian estimates are consistent estimators of the posterior mean and covariance matrix. It is well documented in the literature that MCMC techniques are powerful and efficient for posterior simulation. Due to advances in MCMC, Bayesian methods have gained significant popularity for statistical inference and are now widely applied to a variety of complex models.

It is worth noting that the Bayesian large-sample theory exhibits a key difference from the QML large-sample theory, particularly for misspecified models. Unlike QML theory, Bayesian asymptotic results do not differ between correctly specified and misspecified models. In both cases, the Bayesian large-sample theory is given by:
\begin{equation*}
\left(\mathbf{\hat{I}}^{-1}(\hat{\boldsymbol{\theta}})/n\right)^{-1/2}(\boldsymbol{\theta} - \hat{\boldsymbol{\theta}})|\mathbf{y} \overset{d}{\rightarrow} N\left(\mathbf{0}, \mathbf{I}\right),
\end{equation*}
in probability as $n \to \infty$ \citep{kleijn2012bernstein}.

\subsection{Variational Bayes for Misspecified Models with Massive Data}

\begin{comment}
To obtain the $p({\boldsymbol{\theta}}|\mathbf{y})$, the dominant paradigm in Bayes statistics is MCMC, including Metropolis-Hastings algorithm \citep{metropolis1953equation,hastings1970monte}, and the Gibbs sampler \citep{geman1984stochastic} et al. Although MCMC provides an easy and universal method to sample from Markov chain on $\mathbf{\theta}$, it still faces bottleneck, epsecially in massive data sets (the sample size n is very large). An important situation where the log-likelihood is computationally intractable is Massive Data \citep{bardenet2017markov,quiroz2019speeding}, where the log-likelihood function, is a sum of a very large number of terms and thus too expensive to compute. For the reasons of high computational cost of likelihood evaluation when dealing with massive datasets, MCMC may need thousands minutes or days to reach a stationary posterior distribution.  
\end{comment}

To compute \( p(\boldsymbol{\theta}|\mathbf{y}) \), the dominant paradigm in Bayesian statistics is MCMC, including the Metropolis-Hastings algorithm \citep{metropolis1953equation, hastings1970monte} and the Gibbs sampler \citep{geman1984stochastic}, among others. While MCMC provides a flexible and widely applicable method to sample from the posterior distribution of \(\boldsymbol{\theta}\), it faces significant challenges, particularly when applied to massive datasets where the sample size \(n\) is extremely large.

One notable scenario in which the log-likelihood becomes computationally intractable is when dealing with massive data \citep{bardenet2017markov, quiroz2019speeding}. In such cases, the log-likelihood function is represented by the summation of numerous terms, making it prohibitively expensive to evaluate. Due to the high computational cost associated with likelihood evaluations for massive datasets, MCMC methods can require hours or even days to converge to a stationary posterior distribution.

Recently, to address the limitations of Bayesian inference based on MCMC for massive datasets, Variational Bayes (VB) methods \citep{jordan1999introduction}, have garnered significant attention in the research community. VB offers an alternative to MCMC by solving the following optimization problem:
\begin{equation*}
p^{VB}(\boldsymbol{\theta}|\mathbf{y}) = \arg \min_{q(\boldsymbol{\theta}) \in \Gamma} KL[q(\boldsymbol{\theta}), p(\boldsymbol{\theta}|\mathbf{y})],
\end{equation*}
where \(p^{VB}(\boldsymbol{\theta}|\mathbf{y})\) denotes the Variational Bayesian posterior, and the goal is to approximate the posterior \(p(\boldsymbol{\theta}|\mathbf{y})\) using a tractable variational family \(\Gamma\). A commonly used variational family is the mean-field (MF) family, which assumes the factorized form:
%\begin{equation*}
$q(\boldsymbol{\theta}) = \prod_{i=1}^P q_{\theta_i}(\theta_i).$
%\end{equation*}
This simplification facilitates efficient optimization by reducing computational complexity.

Since VB formulates posterior inference as an optimization problem, it provides a computationally efficient alternative to MCMC, particularly in the context of massive datasets under Bayesian modeling \citep{attias2013inferring, bishop2006pattern}. Empirical studies have shown that VB-based algorithms can be orders of magnitude faster than MCMC \citep{blei2017variational, gunawan2017fast}. Beyond the classical mean-field VB, advances such as stochastic variational inference (SVI) \citep{hoffman2013stochastic} have further enabled scalable Bayesian analysis for large-scale datasets.

The asymptotic properties of the VB posterior have been a topic of significant interest in the literature. Define the second-order derivative of the log-likelihood as
$\mathbf{\bar{H}}_n(\boldsymbol{\theta}) := \frac{1}{n} \frac{\partial^2 \ln p(\mathbf{y}, \boldsymbol{\theta})}{\partial \boldsymbol{\theta} \partial \boldsymbol{\theta}^\prime}$, and take the expectation to obtain
$\mathbf{H}_n(\boldsymbol{\theta}) := E[\mathbf{\bar{H}}_n(\boldsymbol{\theta})]
$, then the normal approximation to the VB posterior can be expressed as:
\begin{equation*}
p^{VBN}(\boldsymbol{\theta}|\mathbf{y}) = (2\pi)^{-P/2} \left| -n\mathbf{H}_n^d \right|^{1/2} \exp\left(-\frac{1}{2} (\boldsymbol{\hat{\theta}}_n - \boldsymbol{\theta})^\prime (-n\mathbf{H}_n^d)(\boldsymbol{\hat{\theta}}_n - \boldsymbol{\theta})\right),
\end{equation*}
where \(\mathbf{H}_n^d\) is a diagonal matrix whose diagonal elements match those of \(\mathbf{H}_n\). As established by \citet{han2019statistical} and \citet{Zhang_2024}, the KL divergence between the VB posterior \(p^{VB}(\boldsymbol{\theta}|\mathbf{y})\) and the normal approximation \(p^{VBN}(\boldsymbol{\theta}|\mathbf{y})\) converges to \(0\) in probability as \(n \to \infty\). \citet{wang2019variational} proved that the total variation between the  VB posterior and $p^{VBN}(\boldsymbol{\theta}|\mathbf{y})$ converges to \(0\) in probability as \(n \to \infty\).

\section{Risk of Predictive Distributions on Misspecified Models based on Variational Bayes} \label{sec:risk}

In the literature, assessing the utility of a misspecified statistical model is typically achieved by examining its predictive performance \citep{bernardo1979expected}. Given a set of future observations \(\mathbf{y}_f\), the predictive distribution is denoted by \(p_f(\mathbf{y}_f|\mathbf{y})\). A commonly used approach for quantifying the predictive performance of a misspecified model is to compute the KL divergence between the true data-generating process \(g(\mathbf{y}_f)\) and the predictive distribution \(p_f(\mathbf{y}_f|\mathbf{y})\), scaled by a factor of 2. This measure is expressed as:
\[
2 \times KL\left[ g\left( \mathbf{y}_f \right), p_f\left( \mathbf{y}_f|\mathbf{y} \right) \right] = 2 E_{\mathbf{y}_f}\left[ \ln \frac{g\left( \mathbf{y}_f \right)}{p\left( \mathbf{y}_f|\mathbf{y} \right)} \right],
\]
which can be rewritten as
$2 \int \left[ \ln \frac{g\left( \mathbf{y}_f \right)}{p\left( \mathbf{y}_f|\mathbf{y} \right)} \right] g\left( \mathbf{y}_f \right) \, d\mathbf{y}_f.$
Building on this KL divergence, statistical decision theory allows the specification of a loss function associated with a decision \(d\) as:
\[
\mathcal{L}(\mathbf{y}, d) = 2 \times KL\left[ g\left( \mathbf{y}_f \right), p\left( \mathbf{y}_f|\mathbf{y}, d \right) \right],
\]
where \(p(\mathbf{y}_f|\mathbf{y}, d)\) represents the predictive density based on decision \(d\). The corresponding risk function is then defined as \citep{good1952rational}:
\[
\text{Risk}(d) = E_{\mathbf{y}}\left[ \mathcal{L}(\mathbf{y}, d) \right] = \int \mathcal{L}(\mathbf{y}, d) g(\mathbf{y}) \, d\mathbf{y}.
\]
In the context of VB, two types of predictive distributions can be derived for prediction: the variational plug-in predictive distribution and the variational posterior predictive distribution. These two distributions correspond to different statistical decisions, resulting in two distinct risk functions. In the subsequent subsection, we evaluate these two risk functions and derive estimators for them. To facilitate this analysis, we first establish the necessary notations and outline mild regularity conditions.

Let $%
\mathbf{y}:=(y_{1},\ldots ,y_{n})$ and $l_{t}\left( \mathbf{y}_{t},{%
\mbox{\boldmath${\theta}$}}\right) =\ln p(\mathbf{y}_{t}|{%
\mbox{\boldmath${\theta}$}})$ be the conditional log-likelihood for the $%
t^{th}$ observation for any $1\leq t\leq n$. For simplicity, we suppress $%
l_{t}\left( \mathbf{y}_{t},{\mbox{\boldmath${\theta}$}}\right) $\ as $%
l_{t}\left( {\mbox{\boldmath${\theta}$}}\right) $ so that the log-likelihood
function $\ln p(\mathbf{y|}{\mbox{\boldmath${\theta}$})}$ is $%
\sum_{t=1}^{n}l_{t}\left( {\mbox{\boldmath${\theta}$}}\right) $.%$\footnote{%
%In the definition of log-likelihood, we ignore the initial condition $\ln
%p(y_{0})$. For weakly dependent data, the impact is aymptotically negligible.%}$ 
And define $\bigtriangledown ^{j}l_{t}\left( {\mbox{\boldmath${\theta}$}}%
\right) $\ to be the $j^{th}$\ derivative of $l_{t}\left( {%
\mbox{\boldmath${\theta}$}}\right) $ and $\bigtriangledown ^{j}l_{t}\left( {%
\mbox{\boldmath${\theta}$}}\right) =l_{t}\left( {\mbox{\boldmath${\theta}$}}%
\right) $\ when $j=0$. 
%For latent variable models, let $l_{t}\left( \mathbf{y%
%}_{t},\mathbf{z}_{t},{\mbox{\boldmath${\theta}$}}\right) =\ln p(\mathbf{y}%
%_{t},\mathbf{z}_{t}|{\mbox{\boldmath${\theta}$}})$, $\bigtriangledown
%^{j}l_{t}\left( \mathbf{y}_{t},\mathbf{z}_{t},{\mbox{\boldmath${\theta}$}}%
%\right) $ be the $j^{th}$\ derivative of $l_{t}\left( \mathbf{y}_{t},\mathbf{%
%z}_{t},{\mbox{\boldmath${\theta}$}}\right) $. 
We suppress the superscript
when $j=1$, and 
\begin{eqnarray*}
&&\mathbf{s}(\mathbf{y},{\mbox{\boldmath${\theta}$}}):=\frac{\partial \ln p(%
\mathbf{y}|{\mbox{\boldmath${\theta}$}})}{\partial {\mbox{\boldmath${%
\theta}$}}}=\sum_{t=1}^{n}\bigtriangledown l_{t}\left( {\mbox{\boldmath${%
\theta}$}}\right) ,\;\mathbf{h}(\mathbf{y},{\mbox{\boldmath${\theta}$}}):=%
\frac{\partial ^{2}\ln p(\mathbf{y}|{\mbox{\boldmath${\theta}$}})}{\partial {%
\mbox{\boldmath${\theta}$}}\partial {\mbox{\boldmath${\theta}$}}^{\prime }}%
=\sum_{t=1}^{n}\bigtriangledown ^{2}l_{t}\left( {\mbox{\boldmath${\theta}$}}%
\right) , \\
&&\mathbf{s}_{t}({\mbox{\boldmath${\theta}$}}):=\bigtriangledown l_{t}\left( 
{\mbox{\boldmath${\theta}$}}\right) ,\;\mathbf{h}_{t}({\mbox{\boldmath${%
\theta}$}}):=\bigtriangledown ^{2}l_{t}\left( {\mbox{\boldmath${\theta}$}}%
\right) , \\
&&\mathbf{B}_{n}\left( {\mbox{\boldmath${\theta}$}}\right) :=Var\left[ \frac{%
1}{\sqrt{n}}\sum_{t=1}^{n}\bigtriangledown l_{t}\left( {\mbox{\boldmath${%
\theta}$}}\right) \right] ,\mathbf{\bar{H}}_{n}({\mbox{\boldmath${\theta}$}}%
):=\frac{1}{n}\sum_{t=1}^{n}\mathbf{h}_{t}({\mbox{\boldmath${\theta}$}}), \\
&&\mathbf{\bar{J}}_{n}({\mbox{\boldmath${\theta}$}}):=\frac{1}{n}%
\sum_{t=1}^{n}\left[ \mathbf{s}_{t}({\mbox{\boldmath${\theta}$}})-\mathbf{%
\bar{s}}_{t}({\mbox{\boldmath${\theta}$}})\right] \left[ \mathbf{s}_{t}({%
\mbox{\boldmath${\theta}$}})-\mathbf{\bar{s}}_{t}({\mbox{\boldmath${\theta}$}%
})\right] ^{\prime },\mathbf{\bar{s}}_{t}({\mbox{\boldmath${\theta}$}})=%
\frac{1}{n}\sum_{t=1}^{n}\mathbf{s}_{t}({\mbox{\boldmath${\theta}$}}), \\
&&\mathcal{L}_{n}({\mbox{\boldmath${\theta}$}}):=\ln p({\mbox{%
\boldmath${				\theta}$}}|\mathbf{y}),\mathcal{L}_{n}^{(j)}({%
\mbox{\boldmath${\theta}$}}):=\partial ^{j}\ln p({\mbox{\boldmath${\theta}$}}%
|\mathbf{y})/\partial {\mbox{\boldmath${\theta}$}}^{j}, \\
&&\mathbf{H}_{n}({\mbox{\boldmath${\theta}$}}):=\int \mathbf{\bar{H}}_{n}({%
\mbox{\boldmath${\theta}$}})g\left( \mathbf{y}\right) d\mathbf{y},\;\mathbf{J%
}_{n}({\mbox{\boldmath${\theta}$}})=\int \mathbf{\bar{J}}_{n}({%
\mbox{\boldmath${\theta}$}})g\left( \mathbf{y}\right) d\mathbf{y}.
\end{eqnarray*}
Then, the following regularity conditions can be imposed  

%\textbf{Assumption 5}: The prior $p({\mbox{\boldmath${\theta}$}})$ is $%
%O_{p}(1)$ for all ${\mbox{\boldmath${\theta}$}}$ $\in {\mbox{\boldmath${%
%\Theta}$}}$.

\textbf{Assumption 1}: ${\mbox{\boldmath${\Theta}$}}\subset R^{P}$ is
compact.

\textbf{Assumption 2}: The data $\mathbf{y}=(y_{1},\ldots ,y_{n})$ is
independent and identically distributed. 

\textbf{Assumption 3: }For all $t$, $l_{t}\left( {\mbox{\boldmath		%
\ensuremath{{\theta}}}}\right) $\ is eight-times differentiable on $%
\mbox{\boldmath\ensuremath{{\Theta}}}$ almost surely.

\textbf{Assumption 4}: For $j=0,1,2,3$, for any ${\mbox{\boldmath${\theta}$},%
\mbox{\boldmath${\theta}$}^{\prime }\in \mbox{\boldmath${\Theta}$}}$, $%
\left\Vert \bigtriangledown ^{j}l_{t}\left( {\mbox{\boldmath${\theta}$}}%
\right) -\bigtriangledown ^{j}l_{t}\left( {\mbox{\boldmath${\theta}$}}%
^{\prime }\right) \right\Vert \leq c_{t}^{j}\left( \mathbf{y}_{t}\right)
\left\Vert {\mbox{\boldmath${\theta}$}}-{\mbox{\boldmath${\theta}$}}^{\prime
}\right\Vert $ in probability, where $c_{t}^{j}\left( \mathbf{y}_{t}\right) $
is a positive random variable with $\sup_{t}E\left\Vert c_{t}^{j}\left( 
\mathbf{y}_{t}\right) \right\Vert <\infty $ and $\frac{1}{n}%
\sum_{t=1}^{n}\left( c_{t}^{j}\left( \mathbf{y}_{t}\right) -E\left(
c_{t}^{j}\left( \mathbf{y}_{t}\right) \right) \right) \overset{p}{%
\rightarrow }0$.

\textbf{Assumption 5}: For $j=0,1,\ldots ,4$, there exists a function $M_{t}(%
\mathbf{y}_{t})$ such that for all ${\mbox{\boldmath${\theta}$}}\in {\ %
\mbox{\boldmath${\Theta}$}}$, $\bigtriangledown ^{j}l_{t}\left( {\ %
\mbox{\boldmath${\theta}$}}\right) $ exists, $\sup_{{\mbox{\boldmath${				%
\theta}$}\in \mbox{\boldmath${\Theta}$}}}\left\Vert \bigtriangledown
^{j}l_{t}\left( {\mbox{\boldmath${\theta}$}}\right) \right\Vert \leqslant
M_{t}(\mathbf{y}_{t})$, and $\sup_{t}E\left\Vert M_{t}(\mathbf{y}%
_{t})\right\Vert ^{r+\delta }\leq M<\infty $ for some $\delta >0$ and $r>2$.

\textbf{Assumption 6}: Let $\mbox{\boldmath${\theta}$}_{n}^{p}$ be the
pseudo-true value that minimizes the KL loss between the DGP and the
candidate model 
\begin{equation*}
\mbox{\boldmath${\theta}$}_{n}^{p}=\arg \min_{{\mbox{\boldmath${\theta}$}}%
\in {\mbox{\boldmath${\Theta}$}}}\frac{1}{n}\int \ln \frac{g(\mathbf{y})}{p(%
\mathbf{y}|{\mbox{\boldmath${\theta}$}})}g(\mathbf{y})d\mathbf{y,}
\end{equation*}%
where $\left\{ \mbox{\boldmath${\theta}$}_{n}^{p}\right\} $ is the sequence
of minimizers interior to ${\mbox{\boldmath${\Theta}$}}$ uniformly in $n$.
For all $\varepsilon >0$,%
\begin{equation}
\lim_{n\rightarrow \infty }\sup \sup_{{\Theta \backslash }N\left( %
\mbox{\boldmath${\theta}$}_{n}^{p},\varepsilon \right) }\frac{1}{n}%
\sum_{t=1}^{n}\left\{ E\left[ l_{t}\left( {\mbox{\boldmath${\theta}$}}%
\right) \right] -E\left[ l_{t}\left( \mbox{\boldmath${\theta}$}%
_{n}^{p}\right) \right] \right\} <0,  \label{iden}
\end{equation}%
where $N\left( \mbox{\boldmath${\theta}$}_{n}^{p},\varepsilon \right) $ is
the open ball of radius $\varepsilon $ around $\mbox{\boldmath${\theta}$}%
_{n}^{p}$.

\textbf{Assumption 7}: The sequence $\left\{ \mathbf{H}_{n}\left( %
\mbox{\boldmath${\theta}$}_{n}^{p}\right) \right\} $ is negative definite
and the sequence $\left\{ \mathbf{B}_{n}\left( \mbox{\boldmath${\theta}$}%
_{n}^{p}\right) \right\} $ is positive definite, both uniformly in $n$.

\textbf{Assumption 8}: The prior density $p\left( \mbox{\boldmath	%
\ensuremath{{\theta}}}\right) $ is thrice continuously differentiable
and $0<p\left( {\mbox{\boldmath\ensuremath{{\theta}}}}_{n}^{0}\right)
<\infty $\ uniformly in $n$. Moreover, there exists an $n^{\ast }$ such
that, for any $n>n^{\ast }$, the posterior distribution $p\left( {\ %
\mbox{\boldmath\ensuremath{{\theta}}}}|\mathbf{y}\right) $ is proper and $%
\int \left\Vert {\mbox{\boldmath\ensuremath{{\theta}}}}\right\Vert
^{2}p\left( {\mbox{\boldmath\ensuremath{{\theta}}}}|\mathbf{y}\right) d{\ %
\mbox{\boldmath\ensuremath{{\theta}}}}<\infty $.

Assumptions 1-7 are well-known primitive conditions for developing the QML
theory, namely consistency and asymptotic normality, for independent and
identically distributed data; see, for example, \citet{gallant1988unified} and
\citet{WOOLDRIDGE19942639}. Assumption 8 is the regular condition for prior density,
see, for example, \citet{li2020deviance}. Assumptions 1-8 are sufficient for the assumptions used by
\citet{Zhang_2024} to develope the asymptotic properties of VB posterior distribution without latent variables.  

\subsection{Risk of VB Plug-in Predictive Distribution}

Under VB inference, for a potentially misspecified model, let $\overline{\boldsymbol{\theta}}^{VB}$ denote the VB estimator of the parameter $\boldsymbol{\theta}$ which corresponds to the posterior mean of the variational posterior distribution
$p^{VB}(\boldsymbol{\theta}|\mathbf{y})$. In cases where the posterior mean does not have a closed-form analytical solution, it can generally be approximated consistently using the sample mean
$\boldsymbol{\bar\theta}^{VB}=\frac {1}{J}\sum_{j=1}^J\boldsymbol{\theta_{VB}}^{(j)}
$, where $\boldsymbol{\theta_{VB}}^{(j)}$, $j=1,2,\cdots J$ are generated from $p^{VB}(\boldsymbol{\theta}|\mathbf{y})$.

\begin{comment}
Following the literature of development of the popular information criterion AIC,TIC, DIC,etc, we assume that there are some future replicated data $\mathbf{y}_{rep}$ which has the same DGP with $\mathbf{y}$. More details about $\mathbf{y}_{rep}$, one can refer to the excellent textbook about model selection, that is, \citet{anderson2004model} and its reference therein.  After that, for future data $\mathbf{y}_{rep}$, the VB plug-in predictive distribution can be specified as $p\left(\mathbf{y}_{rep}| \overline{\boldsymbol{\theta}}%
^{VB}\right)$. Hence, for this predictive distribution, the loss function associated with the corresponding statistical decision denoted as $d_1$ can be given by
\end{comment}

Building on the literature regarding the development of popular information criteria such as AIC, TIC, and DIC, we assume the existence of future replicated data \(\mathbf{y}_{rep}\), which shares the same DGP as the observed data \(\mathbf{y}\) and independent of \(\mathbf{y}\). For more details on the concept of \(\mathbf{y}_{rep}\), one may refer to the comprehensive discussion in the seminal textbook on model selection by \citet{anderson2004model} and the references therein. For the future data \(\mathbf{y}_{rep}\), the VB plug-in predictive distribution can be expressed as $p\left(\mathbf{y}_{rep}| \overline{\boldsymbol{\theta}}^{VB}\right)$, where \(\overline{\boldsymbol{\theta}}^{VB}\) represents the VB estimator, typically the posterior mean of the variational posterior distribution. The predictive distribution provides a probabilistic framework for evaluating future observations based on the fitted model. Correspondingly, the loss function associated with the statistical decision, denoted as \(d_1\), can be specified as follows:
\begin{equation*}
\mathcal{L}(\mathbf{y},d_1)=2\times KL\left[ g\left( \mathbf{y}%
_{rep}\right) ,p\left(\mathbf{y}_{rep}| \overline{\boldsymbol{\theta}}%
^{VB}\right) \right].
\end{equation*}%
In this context, the risk function can be expressed as:
\begin{eqnarray*}
&&Risk(d_1)=E_{\mathbf{y}}\left[ \mathcal{L}( \mathbf{y}%
,d_1)\right]=2\times E_{\mathbf{y}}E_{\mathbf{y}_{rep}}\left[ \ln 
\frac{g\left( \mathbf{y}_{rep}\right) }{p\left(\mathbf{y}_{rep}| \overline{%
\boldsymbol{\theta}}^{VB}\right) }\right]\\
&&=E_{\mathbf{y}}E_{\mathbf{y}_{rep}}\left[ 2\ln g\left( 
\mathbf{y}_{rep}\right) \right] +E_{\mathbf{y}}E_{\mathbf{y}_{rep}}\left[
-2\ln p\left(\mathbf{y}_{rep}| \overline{\boldsymbol{\theta}}^{VB}\right) %
\right].
\end{eqnarray*}
Since $E_{\mathbf{y}}E_{\mathbf{y}_{rep}}%
\left[ 2\ln g\left( \mathbf{y}_{rep}\right) \right] $ is the same across all statistical decisions, the risk function can be expressed as:
\begin{equation*}
Risk(d_1)=C+E_{\mathbf{y}}E_{\mathbf{y}_{rep}}\left[ -2\ln
p\left(\mathbf{y}_{rep}| \overline{\boldsymbol{\theta}}^{VB}\right) \right]
\end{equation*}
where $C=E_{\mathbf{y}}E_{\mathbf{y}_{rep}}\left[ 2\ln g\left( 
\mathbf{y}_{rep}\right) \right]$. 
\begin{comment}
From this risk function, we can see that the smaller the $Risk(d_1)$, the better the predictive distribution performs when using $p\left(\mathbf{y}_{rep}| \overline{\boldsymbol{\theta}}^{VB}\right) $ to predict $ \mathbf{y}_{rep}$. Unfortunately, generally, this risk function doesn't have analytical form. Hence, we still need to evaluate the risk function in order to examine the predictive behavior. In the following, we derive an asymptotic expansion of this risk function via the following theorem. 
\end{comment}

It is evident that a smaller value of \( \text{Risk}(d_1) \) indicates better performance of the predictive distribution \( p\left(\mathbf{y}_{rep}| \overline{\boldsymbol{\theta}}^{VB} \right) \) in predicting the replicate data \(\mathbf{y}_{rep}\). However, in general, this risk function does not have a closed-form analytical expression. Therefore, evaluating the risk function is essential for assessing the predictive behavior of the model.

To address this challenge, we derive an asymptotic expansion of the risk function, as presented in the following theorem. This derivation provides a practical approach to approximate the risk function in large-sample scenarios, offering insights into the predictive performance of the VB-based approach.

\begin{theorem}
\label{riskvtic} Under Assumptions 1-8, it can be shown that 
\begin{eqnarray*}
&&E_{\mathbf{y}}E_{\mathbf{y}_{rep}}\left( -2\ln p\left(\mathbf{y}_{rep}| 
\overline{\boldsymbol{\theta}}^{VB}\right) \right)
=E_{\mathbf{y}}\left( -2\ln p\left( \mathbf{y}|\widehat{{\mbox{\boldmath${\theta}$}}}_{n}\left( \mathbf{y}\right) \right) \right) -2%
\mathbf{tr}\left[\mathbf{B}_{n}\mathbf{H}_{n}^{-1}\right] +o\left( 1\right).
\end{eqnarray*}
with 
$
\mathbf{B}_{n} = \mathbf{B}_{n}\left(\boldsymbol{\theta}_n^p\right), \mathbf{H}_{n} = \mathbf{H}_{n}\left(\boldsymbol{\theta}_n^p\right)
$, where $\widehat{{\mbox{\boldmath${\theta}$}}}_{n}\left( \mathbf{y}\right)$ 
is the MLE estimator of $\boldsymbol{\theta}$.
\end{theorem} 

\begin{remark}
	Under Assumptions 1-8, it can be shown that when the model is correctly specified
	\begin{eqnarray*}
		&&E_{\mathbf{y}}E_{\mathbf{y}_{rep}}\left( -2\ln p\left(\mathbf{y}_{rep}| 
		\overline{\boldsymbol{\theta}}^{VB}\right) \right)
		=E_{\mathbf{y}}\left( -2\ln p\left( \mathbf{y}|\widehat{{\mbox{\boldmath${\theta}$}}}_{n}\left( \mathbf{y}\right) \right) \right) -2%
		\mathbf{tr}\left[\mathbf{B}_{n}\mathbf{H}_{n}^{-1}\right] +o\left( 1\right)\\
		&=&E_{\mathbf{y}}\left( -2\ln p\left( \mathbf{y}|\widehat{{\mbox{\boldmath${\theta}$}}}_{n}\left( \mathbf{y}\right) \right) \right) +2%
		\mathbf{tr}\left[\mathbf{H}_{n}\mathbf{H}_{n}^{-1}\right] +o\left( 1\right)\\
		&=&E_{\mathbf{y}}\left( -2\ln p\left( \mathbf{y}|\widehat{{\mbox{\boldmath${\theta}$}}}_{n}\left( \mathbf{y}\right) \right) \right) +2\mathbf{P} +o\left( 1\right).
	\end{eqnarray*}
\end{remark}

\subsection{Risk of VB Posterior Predictive Distribution}

Under the Bayesian framework, the VB posterior predictive distribution for the replicated data \(\mathbf{y}_{rep}\), corresponding to \(p^{VB}(\boldsymbol{\theta}|\mathbf{y})\), is defined as:
\begin{equation}
p^{VB}(\mathbf{y}_{rep}|\mathbf{y})=\int p(\mathbf{y}_{rep}|{\mbox{\boldmath${%
\theta}$}},\mathbf{y})p^{VB}({\mbox{\boldmath${\theta }$}}|\mathbf{y})d{%
\mbox{\boldmath${\theta}$}}.  \label{pred_VB}
\end{equation}%
As described in Section 3.1, the KL divergence between the true data-generating process \(g\left( \mathbf{y}_{rep} \right)\) and the VB posterior predictive distribution \(p^{VB}\left( \mathbf{y}_{rep}|\mathbf{y}\right)\), multiplied by 2, is given by:
\begin{eqnarray*}
&&2\times KL\left[ g\left( \mathbf{y}_{rep}\right) ,p^{VB}\left( \mathbf{y}%
_{rep}|\mathbf{y}\right) \right] = 2E_{\mathbf{y}_{rep}}\left[ \ln \frac{%
g\left( \mathbf{y}_{rep}\right) }{p^{VB}\left( \mathbf{y}_{rep}|\mathbf{y}%
\right) }\right] \\
&=&2\int \left[ \ln \frac{g\left( \mathbf{y}_{rep}\right) }{p^{VB}\left( 
\mathbf{y}_{rep}|\mathbf{y}\right) }\right] g\left( \mathbf{y}%
_{rep}\right) \mbox{d}\mathbf{y}_{rep}
\end{eqnarray*}%
This divergence is used to quantify the predictive performance of the VB posterior predictive distribution. Accordingly, the loss function associated with the statistical decision \(d_2\), which involves using the VB posterior predictive distribution for prediction, is defined as:
\[
\mathcal{L}(\mathbf{y}, d_2) = 2 \times KL\left[ g\left( \mathbf{y}_{rep} \right), p^{VB}\left( \mathbf{y}_{rep}|\mathbf{y} \right) \right].
\]
The corresponding risk function for the decision \(d_2\) can be expressed as:
\begin{eqnarray*}
&&Risk(d_2)=E_{\mathbf{y}}\left[ \mathcal{L}(\mathbf{y}%
,d_2)\right] =2\times E_{\mathbf{y}}E_{\mathbf{y}_{rep}}\left[ \ln 
\frac{g\left( \mathbf{y}_{rep}\right) }{p^{VB}\left( \mathbf{y}_{rep}|%
\mathbf{y}\right) }\right]  \\
&=& E_{\mathbf{y}}E_{\mathbf{y}_{rep}}\left[ 2\ln g\left( 
\mathbf{y}_{rep}\right) \right] +E_{\mathbf{y}}E_{\mathbf{y}_{rep}}\left[
-2\ln p^{VB}\left( \mathbf{y}_{rep}|\mathbf{y}\right) \right], 
\end{eqnarray*}
which can be further rewritten as:
\[
\text{Risk}(d_2) = C + E_{\mathbf{y}}E_{\mathbf{y}_{rep}}\left[ -2 \ln p^{VB}\left( \mathbf{y}_{rep}|\mathbf{y} \right) \right],
\]
where \(C = E_{\mathbf{y}}E_{\mathbf{y}_{rep}}\left[ 2 \ln g\left( \mathbf{y}_{rep} \right) \right]\) is a constant that depends only on the DGP.

From this expression, it is evident that a smaller \(\text{Risk}(d_2)\) indicates better predictive performance of \(p^{VB}\left( \mathbf{y}_{rep}|\mathbf{y} \right)\) in approximating \(g\left( \mathbf{y}_{rep} \right)\). In the following, we derive an asymptotic expansion of this risk function via the following theorem.

\begin{theorem}
\label{riskvpic} Under Assumptions 1-8, it can be shown
\begin{eqnarray*}
&&E_{\mathbf{y}}E_{\mathbf{y}_{rep}}\left( -2\ln p^{VB}\left( \mathbf{y}%
_{rep}|\mathbf{y}\right) \right) \\
&=&E_{\mathbf{y}}\left( -2\ln p\left( \mathbf{y}|\widehat{{%
\mbox{\boldmath${\theta}$}}}_{n}\left( \mathbf{y}\right) \right) \right)
+\ln \left( \left\vert -\mathbf{H}_{n}\left( -\mathbf{H}_{n}^{d}\right)
^{-1}+\mathbf{I}_{n}\right\vert \right) +\mathbf{tr}\left[ \mathbf{B}%
_{n}\left( -\mathbf{H}_{n}\right) ^{-1}\right] \\
&&-\mathbf{tr}\left[ \left( -\mathbf{H}_{n}+\left( -\mathbf{H}%
_{n}^{d}\right) \right) ^{-1}\left( \mathbf{B}_{n}+\left( -\mathbf{H}%
_{n}^{d}\right) \mathbf{C}_{n}\left( -\mathbf{H}_{n}^{d}\right) \right) %
\right] +\mathbf{tr}\left[ \left( -\mathbf{H}_{n}^{d}\right) \mathbf{C}_{n}%
\right] +o\left( 1\right)
\end{eqnarray*}%
where $\mathbf{C}_{n}=\mathbf{H}_{n}^{-1}\mathbf{B}_{n}\mathbf{H}_{n}^{-1}$, 
$\mathbf{H}_{n}^{d}$ is a diagonal matrix with the same diagonal elements as
in $\mathbf{H}_{n}$.
\end{theorem}

\begin{remark}
If $\mathbf{H}_{n}$ is diagonal, that
is $\mathbf{H}_{n}^{d}=\mathbf{H}_{n}$, it can be shown that%
\begin{equation}
\ln \left( \left\vert -\mathbf{H}_{n}\left( -\mathbf{H}_{n}^{d}\right) ^{-1}+%
\mathbf{I}_{n}\right\vert \right) =\ln \left( \left\vert -\mathbf{H}%
_{n}\left( -\mathbf{H}_{n}\right) ^{-1}+\mathbf{I}_{n}\right\vert \right)
=\ln \left( \left\vert 2\mathbf{I}_{n}\right\vert \right) =\mathbf{P}\ln 2,
\end{equation}%
and%
\begin{equation}
\begin{aligned}
&-\mathbf{tr}\left[ \left( -\mathbf{H}_{n}+\left( -\mathbf{H}%
_{n}^{d}\right) \right) ^{-1}\left( \mathbf{B}_{n}+\left( -\mathbf{H}%
_{n}^{d}\right) \mathbf{C}_{n}\left( -\mathbf{H}_{n}^{d}\right) \right) %
\right] +\mathbf{tr}\left[ \left( -\mathbf{H}_{n}^{d}\right) \mathbf{C}_{n}%
\right] \\
&=-\mathbf{tr}\left[ \left( -\mathbf{H}_{n}+\left( -\mathbf{H}_{n}\right)
\right) ^{-1}\left( \mathbf{B}_{n}+\left( -\mathbf{H}_{n}\right) \mathbf{C}%
_{n}\left( -\mathbf{H}_{n}\right) \right) \right] +\mathbf{tr}\left[ \left( -%
\mathbf{H}_{n}\right) \mathbf{C}_{n}\right] \\
&=-\mathbf{tr}\left[ \left( -2\mathbf{H}_{n}\right) ^{-1}\left( 2\mathbf{B}%
_{n}\right) \right] +\mathbf{tr}\left[ \mathbf{B}_{n}\left( -\mathbf{H}%
_{n}\right) ^{-1}\right] =0,
\end{aligned}
\end{equation}
then%
\begin{eqnarray*}
	&&E_{\mathbf{y}}E_{\mathbf{y}_{rep}}\left( -2\ln p^{VB}\left( \mathbf{y}%
	_{rep}|\mathbf{y}\right) \right) \\
	&=&E_{\mathbf{y}}\left( -2\ln p\left( \mathbf{y}|\widehat{{%
			\mbox{\boldmath${\theta}$}}}_{n}\left( \mathbf{y}\right) \right) \right) +%
	\mathbf{P}\ln 2+\mathbf{tr}\left[ \mathbf{B}_{n}\left( -\mathbf{H}%
	_{n}\right) ^{-1}\right] +o\left( 1\right) .
\end{eqnarray*}
\end{remark}
\begin{corollary}
Under Assumptions 1-8, it can be shown that when the model is correctly specified%
\begin{eqnarray*}
&&E_{\mathbf{y}}E_{\mathbf{y}_{rep}}\left( -2\ln p^{VB}\left( \mathbf{y}%
_{rep}|\mathbf{y}\right) \right) \\
&=&E_{\mathbf{y}}\left( -2\ln p\left( \mathbf{y}|\widehat{{%
\mbox{\boldmath${\theta}$}}}_{n}\left( \mathbf{y}\right) \right) \right)
+\ln \left( \left\vert -\mathbf{H}_{n}\left( -\mathbf{H}_{n}^{d}\right)
^{-1}+\mathbf{I}_{n}\right\vert \right) +\mathbf{P} \\
&&-\mathbf{tr}\left[ \left( -\mathbf{H}_{n}+\left( -\mathbf{H}%
_{n}^{d}\right) \right) ^{-1}\left( -\mathbf{H}_{n}+\left( -\mathbf{H}%
_{n}^{d}\right) \left( -\mathbf{H}_{n}\right) ^{-1}\left( -\mathbf{H}%
_{n}^{d}\right) \right) \right] \\
&&+\mathbf{tr}\left[ \left( -\mathbf{H}_{n}^{d}\right) \left( -\mathbf{H}%
_{n}\right) ^{-1}\right] +o\left( 1\right)
\end{eqnarray*}%
where $\mathbf{H}_{n}^{d}$ is a diagonal matrix with the same diagonal
elements as in $\mathbf{H}_{n}$. 
\end{corollary}
\begin{remark}
If $\mathbf{H}_{n}$ is
diagonal, that is $\mathbf{H}_{n}^{d}=\mathbf{H}_{n}$, then%
\begin{eqnarray*}
	&&E_{\mathbf{y}}E_{\mathbf{y}_{rep}}\left( -2\ln p^{VB}\left( \mathbf{y}%
	_{rep}|\mathbf{y}\right) \right) \\
	&=&E_{\mathbf{y}}\left( -2\ln p\left( \mathbf{y}|\widehat{{%
			\mbox{\boldmath${\theta}$}}}_{n}\left( \mathbf{y}\right) \right) \right) +%
	\mathbf{P}\ln 2+\mathbf{P}+o\left( 1\right) .
\end{eqnarray*}
\end{remark}
%\begin{theorem}
%\label{thm3}Under Assumptions 1-9, it can be shown that
%\begin{equation*}
%\lim_{n\rightarrow +\infty }Risk({d_1})\leq \lim_{n\rightarrow +\infty }Risk({d_2})
%\end{equation*}
%\end{theorem}

%\begin{remark}
%Theorem \ref{thm3} shows that when the model is misspecified, under the KL
%loss function, the risk of Bayesian predictive distribution is less
%(weakly) than that of the plug-in predictive distribution.
%\end{remark}
\section{Predictive Information Criteria for Comparing
Misspecified Models with Massive Data based on VB}
\label{sec:model selection}
\begin{comment}
In this section, we show how to develop the newly predictive information criteria for model comparison under misspecification with massive data. Based on the risk functions investigated in Section \ref{sec:risk}, Section \ref{subsec:decision theory} introduces the statistical decision theory for model comparison. In Section \ref{subsec:plugin vbic}, we develop the information criterion named as VTIC on the basis of VB plug-in predictive distribution. In Section \ref{subsec:posterior vbic}, we develop the information criterion named as VPIC on the basis of VB posterior predictive distribution. 
\end{comment}

In this section, we outline the development of new predictive information criteria for model comparison in the context of misspecified models with massive data. Building on the risk functions analyzed in Section \ref{sec:risk}, Section \ref{subsec:decision theory} introduces the framework of statistical decision theory for model comparison. In Section \ref{subsec:plugin vbic}, we propose an information criterion, termed VDIC$_{M}$, based on the VB plug-in predictive distribution. We then present another information criterion, termed VPIC, which is constructed using the VB posterior predictive distribution in Section \ref{subsec:posterior vbic}. At last, in Section \ref{subsec: bf and bic}, we then discuss BFs and BIC in the context of misspecified models. 

\subsection{Statistical Decision Theory based on Risk Function for Model
Selection}
\label{subsec:decision theory}
\begin{comment}
In this section, from the predictive viewpoint, we extend the decisional
approach of Section \ref{sec:risk} to develop information criteria for model comparison.
Now suppose there are $K$ candidate models that are all misspecified and we
have to select a model. Denote these candidate models by $M_{k}$, $%
k=1,2,\cdots ,K$. As argued in the last section, we do so by minimizing the
risk of the statistical decision. Following \citet{Akaike1973} and \citet{takeuchi1976distribution}, we assume that parameter ${\mbox{\boldmath${\theta}$}}$ is only
estimated from the sample $\mathbf{y}$. Unlike Akaike (1973) and \citet{takeuchi1976distribution}, we do not plug the MLE into the KL divergence function.
\end{comment}

In this section, from a predictive perspective, we extend the decisional framework introduced in Section \ref{sec:risk} to develop information criteria for model comparison. Suppose there are \(K\) candidate models, all of which may be misspecified, and the task is to select the most suitable model. These candidate models are denoted by \(M_{k}\), where \(k = 1, 2, \ldots, K\). As discussed in the previous section, this selection is achieved by minimizing the risk associated with the statistical decision.

\begin{comment}
Following the foundational work of \citet{Akaike1973} and \citet{takeuchi1976distribution}, we assume that the parameter \(\boldsymbol{\theta}\) is estimated solely from the observed sample \(\mathbf{y}\). However, unlike \citet{Akaike1973} and \citet{takeuchi1976distribution}, we do not substitute the maximum likelihood estimator (MLE) directly into the Kullback-Leibler (KL) divergence function.
\end{comment}

\begin{comment}
Assuming that the probabilistic behavior of observed data, $\mathbf{y}\in 
\mathbf{Y,}$ is described by a set of probabilistic models such as $\left\{
M_{k}\right\} _{k=1}^{K}:=\left\{ p\left( \mathbf{y}|{\mbox{\boldmath${%
\theta}$}}_{k},M_{k}\right) \right\} _{k=1}^{K}$ where parameter ${%
\mbox{\boldmath${\theta}$}}_{k}$ is the set of parameters in model $M_{k}$.
Formally, the model selection problem can be taken as a decision problem to
select a model among $\left\{ M_{k}\right\} _{k=1}^{K}$ where the action
space has $K$ elements, namely, $\{d_{k}\}_{k=1}^{K}$, where $d_{k}$ means $%
M_{k}$ is selected.
\end{comment}

Assume that the probabilistic behavior of the observed data \(\mathbf{y} \in \mathbf{Y}\) is described by a set of probabilistic models \(\{M_k\}_{k=1}^K := \{p(\mathbf{y}|\boldsymbol{\theta}_k, M_k)\}_{k=1}^K\), where \(\boldsymbol{\theta}_k\) represents the set of parameters associated with model \(M_k\). Formally, the model selection problem can be framed as a decision-making problem, where the goal is to select one model from \(\{M_k\}_{k=1}^K\). In this context, the action space comprises \(K\) elements, denoted by \(\{d_k\}_{k=1}^K\), where \(d_k\) indicates that model \(M_k\) is selected.

\begin{comment}
For the decision problem, similarly to Section \ref{sec:risk}, a loss function, $\mathcal{%
L}(\mathbf{y},d_{k})$, which measures the loss of decision $d_{k}$ as a
function of $\mathbf{y}$, must be specified. Given the loss function, the
risk can be also defined as 
\begin{equation*}
Risk(d_{k})=E_{\mathbf{y}}\left[ \mathcal{L}(\mathbf{y},d_{k})\right] =\int 
\mathcal{L}(\mathbf{y},d_{k})g(\mathbf{y})\mbox{d}\mathbf{y,}
\end{equation*}%
where $g(\mathbf{y})$ is the DGP of $\mathbf{y}$. Hence, the model selection
problem is equivalent to optimizing the statistical decision, 
\begin{equation*}
k^{\ast }=\arg \min_{k}Risk(d_{k}).
\end{equation*}%
Based on the set of candidate models $\left\{ M_{k}\right\} _{k=1}^{K}$, the
model $M_{k^{\ast }}$ with the decision $d_{k^{\ast }}$ is selected.
\end{comment}

For the decision-making process, as in Section \ref{sec:risk}, a loss function \(\mathcal{L}(\mathbf{y}, d_k)\) must be specified. This loss function quantifies the loss incurred by selecting decision \(d_k\). Given the loss function, the corresponding risk can be defined as:
\[
\text{Risk}(d_k) = E_{\mathbf{y}}\left[ \mathcal{L}(\mathbf{y}, d_k) \right] = \int \mathcal{L}(\mathbf{y}, d_k) g(\mathbf{y}) d\mathbf{y},
\]
where \(g(\mathbf{y})\) is the DGP. Consequently, the model selection problem is equivalent to optimizing the statistical decision by minimizing the risk:
\[
k^\ast = \arg \min_k \text{Risk}(d_k).
\]
Based on the set of candidate models \(\{M_k\}_{k=1}^K\), the model \(M_{k^\ast}\), corresponding to the decision \(d_{k^\ast}\), is selected as the optimal model.

The quantity used to assess the predictive ability of a candidate model is the KL divergence between the DGP \(g\left( \mathbf{y}_{rep} \right)\) and a predictive distribution \(p\left( \mathbf{y}_{rep}|\mathbf{y},M_k \right)\), scaled by a factor of 2:
\[
2 \times KL\left[ g\left( \mathbf{y}_{rep} \right), p\left( \mathbf{y}_{rep}|\mathbf{y},M_k \right) \right] = 2 E_{\mathbf{y}_{rep}}\left[ \ln \frac{g\left( \mathbf{y}_{rep} \right)}{p\left( \mathbf{y}_{rep}|\mathbf{y},M_k \right)} \right],
\]
which can also be written as $ 2 \int \left[ \ln \frac{g\left( \mathbf{y}_{rep} \right)}{p\left( \mathbf{y}_{rep}|\mathbf{y},M_k \right)} \right] g\left( \mathbf{y}_{rep} \right) d\mathbf{y}_{rep}.$
Similar to the framework introduced in Section 3, the loss function associated with the decision \(d_k\) is defined as $
\mathcal{L}(\mathbf{y}, d_k) = 2 \times KL\left[ g\left( \mathbf{y}_{rep} \right), p\left( \mathbf{y}_{rep}|\mathbf{y},M_k \right) \right].$
Thus, the model selection problem is formulated as:
\[
k^\ast = \arg \min_k \text{Risk}(d_k) = \arg \min_k E_{\mathbf{y}}\left[ \mathcal{L}(\mathbf{y}, d_k) \right],
\]
which can be further expanded as:
\[
k^\ast = \arg \min_k \left\{ 2 \times E_{\mathbf{y}} E_{\mathbf{y}_{rep}}\left[ \ln \frac{g\left( \mathbf{y}_{rep} \right)}{p\left( \mathbf{y}_{rep}|\mathbf{y},M_k \right)} \right] \right\}.
\]
Rearranging terms gives:
\[
k^\ast = \arg \min_k \left\{ E_{\mathbf{y}} E_{\mathbf{y}_{rep}}\left[ 2 \ln g\left( \mathbf{y}_{rep} \right) \right] + E_{\mathbf{y}} E_{\mathbf{y}_{rep}}\left[ -2 \ln p\left( \mathbf{y}_{rep}|\mathbf{y},M_k \right) \right] \right\}.
\]
Since \(g\left( \mathbf{y}_{rep} \right)\) is the DGP, the term \(E_{\mathbf{y}_{rep}}\left[ 2 \ln g\left( \mathbf{y}_{rep} \right) \right]\) is constant across all candidate models and can therefore be omitted from the equation. Consequently, the model selection problem simplifies to:
\[
k^\ast = \arg \min_k \text{Risk}(d_k) = \arg \min_k E_{\mathbf{y}} E_{\mathbf{y}_{rep}}\left[ -2 \ln p\left( \mathbf{y}_{rep}|\mathbf{y},M_k \right) \right].
\]
The smaller the value of \(\text{Risk}(d_k)\), the better the performance of the candidate model in using the predictive distribution \(p\left( \mathbf{y}_{rep}|\mathbf{y},M_k \right)\) to approximate \(g\left( \mathbf{y}_{rep} \right)\). Evaluating the risk among candidate models is therefore essential for making the optimal decision.

It is important to note that the action space in this context is larger than in previous cases. From a predictive perspective, we not only need to select a model for prediction but also determine which predictive distribution to use.  The action space is denoted byby $\left\{d_{k^1}, d_{k^2}\right\}_{k=1}^K$ where $d_{k^a}(a
\in(1,2)) $ means $M_k$ is selected, and the predictions are generated from $p\left(%
\mathbf{y}_{rep}|\mathbf{y}, M_k, d_a\right)$. If $a=1$, it means that the VB plug-in predictive distribution, $p\left(\mathbf{y}_{\text {rep}}|\mathbf{y}, M_k,
d_1\right)=p\left(\mathbf{y}_{\text {rep}}|\overline{\boldsymbol{\theta}}^{VB},
M_k\right)$ is used; if $a=2$, it means that the VB posterior predictive
distribution, $p\left(\mathbf{y}_{\text {rep }}|\mathbf{y}, M_k, d_2\right)=$
$p^{VB}\left(\mathbf{y}_{rep} \mid \mathbf{y}, M_k\right)$ is used. The KL
divergence for this setup is defined as 
\begin{equation*}
\left.\mathcal{L}\left(\mathbf{y}, d_{k^a}\right)=2 \times K L\left[g\left(%
\mathbf{y}_{\text {rep}}\right), p\left(\mathbf{y}_{r e p}|\mathbf{y},
d_{k^a}\right)\right)\right]
\end{equation*}
where $p\left(\mathbf{y}_{r e p}|\mathbf{y}, d_{k^a}\right):=p\left(\mathbf{y%
}_{r e p} \mid \mathbf{y}, M_k, d_a\right)$. The risk associated with $%
d_{k^a}$ is then given by
\begin{equation*}
Risk\left(d_{k^a}\right)=E_{\mathbf{y}}\left(\mathcal{L}\left(\mathbf{%
y}, d_{k^a}\right)\right)=\int \mathcal{L}\left(\mathbf{y}, d_{k^a}\right) g(%
\mathbf{y}) \mathrm{d} \mathbf{y}.
\end{equation*}
Consequently, the model selection problem is equivalent to solving the following statistical decision problem: 
\begin{equation}
\min_{a \in\{1,2\}} \min_{k \in\{1, \cdots, K\}} Risk%
\left(d_{k^a}\right).
\end{equation}

Since the DGPs \(g\left(\boldsymbol{y}\right)\) and \(g\left(\boldsymbol{y}_{rep}\right)\) are unknown, directly evaluating the risk associated with decision \(d_{k^a}\) is infeasible. However, it is possible to approximate the risk by using an asymptotically unbiased estimator of \(\text{Risk}\left(d_{k^a}\right)\). As noted in the literature, various information criteria proposed for model selection can be interpreted as asymptotically unbiased estimators of the expected loss function, up to a constant, under different statistical decision frameworks \citep{vrieze2012model}.

Traditionally, model selection has been conducted using information criteria that assess the relative quality of statistical models for a given dataset. Under the frequentist framework, criteria such as AIC, TIC, and their variants have been widely applied. Under the Bayesian framework, criteria include DIC and its extensions, such as the deviance information criterion for misspecified models (\(\text{DIC}_M\)) proposed by \citet{li2020deviance}. These information criteria have been shown to follow the principles of statistical decision theory discussed above. Specifically, AIC, TIC, DIC, and \(\text{DIC}_M\) are all constructed by estimating the KL divergence between the DGP and the corresponding predictive distributions. In this study, we develop new approaches that adhere to a similar decision-theoretical framework. To provide context, we first present two remarks that introduce these popular information criteria within this framework. Subsequently, we propose our new information criteria in the following subsections.

\begin{remark}
Under some regularity conditions, under Bayesian framework, for misspecified models, \citet{li2020deviance} proposed the new version of DIC by \citet{spiegelhalter2002bayesian} named as so-called $%
\text{DIC}_{M}^{k}$ for, that is, for model k, 
\begin{equation}
\text{DIC}_{M}^{k}=-2\ln p(\mathbf{y}| \bar{\boldsymbol{\theta}}  _{k},M_{k})+2P_{M}^{k},P_{M}^{k}=\mathbf{tr}\left\{ n\mathbf{\bar{\Omega}}%
_{n}\left( \bar{\boldsymbol{\theta}}_{k} \right) V\left( \bar{\boldsymbol{\theta}}_{k} \right) \right\} ,  \label{dicm1}
\end{equation}%
where $V\left( \bar{\boldsymbol{\theta}}_{k}\right) $ is the posterior
covariance matrix given by $V\left( \bar{\boldsymbol{\theta}}_{k}\right) =E\left[
\left( \boldsymbol{\theta}_{k}-\bar{\boldsymbol{\theta}}_{k}\right) \left( \boldsymbol{\theta}_{k}-\bar{\boldsymbol{\theta}}_{k}\right) ^{\prime }|\mathbf{y,}M_{k}%
\right] $ and 
$
\mathbf{\bar{\Omega}}_{n}\left( \widehat{\boldsymbol{\theta}}_{k}\right) =\frac{1}{n}
\sum_{t=1}^{n}\mathbf{s}_{t}\left( \widehat{\boldsymbol{\theta}}_{k}\right) \mathbf{%
s}_{t}\left( \widehat{\boldsymbol{\theta}}_{k}\right) ^{\prime }$.
For this information criterion, \citet{li2020deviance} showed that the regular plug-in predictive distribution, $p\left( \mathbf{y}_{rep}|\mathbf{y},M_{k},d_{1}\right)
=p\left( \mathbf{y}_{rep}|\bar{\boldsymbol{\theta}}_{k},M_{k}\right) $ can be used for constructing the loss function and the corresponding risk function discussed above.  Hence, from statistical decision viewpoint discussed above, when $a=1$,for misspecified models, it can be shown in \citet{li2020deviance} that 
\begin{equation*}
Risk(d_{k^{1}})=E_{\mathbf{y}}\left( \mathcal{L}(\mathbf{y}%
,d_{k^{1}})\right) =\int \mathcal{L}(\mathbf{y},d_{k^{1}})g(\mathbf{y})%
\mbox{d}\mathbf{y}=E_{\mathbf{y}}\left[ \text{DIC}_{M}^{k}+2C\right] +o(1).
\end{equation*}
If the candidate models are restricted into correctly specified
models or good models which are good approximation to DGP, $\text{DIC}_{M}^{k}$  is reduced as a good approximation of  $\text{DIC}^{k}$ of  \citet{spiegelhalter2002bayesian} given by
\begin{equation}
	\text{DIC}^{k}=-2\ln p(\mathbf{y}|\bar{\boldsymbol{\theta}}
	_{k},M_{k})+2P_{D}^{k},P_{D}^{k}=\int 2\left[ \ln p(\mathbf{y}|\bar{\boldsymbol{\theta}}
	_{k},M_{k})-\ln p(\mathbf{y}|\boldsymbol{\theta}
	_{k},M_{k})\right] %
	\mbox{d}{\mbox{\boldmath${\theta}$}}.  \label{dic2}
\end{equation}
It was shown in \citet{li2024deviance} that 
\begin{equation*}
Risk(d_{k^{1}})=E_{\mathbf{y}}\left( \mathcal{L}(\mathbf{y}%
,d_{k^{1}})\right) =\int \mathcal{L}(\mathbf{y},d_{k^{1}})g(\mathbf{y})%
\mbox{d}\mathbf{y}=E_{\mathbf{y}}\left[ \text{DIC}^{k}+2C\right] +o(1),
\end{equation*}
More details about the theoretical development of $\text{DIC}^{k}$
and $\text{DIC}_{M}^{k}$, one can refer to \citet{spiegelhalter2002bayesian}, \citet{li2020deviance}, \citet{li2024deviance} and reference therein. 
\end{remark}

\begin{remark}
	For some misspecified model k, under frequentist framework, Takeuchi
	information criterion (TIC) of \citet{takeuchi1976distribution} \footnote{TIC is originally developed by \citet{takeuchi1976distribution} for independent data and \citet{li2020deviance} relaxed this limitation to weakly dependent data} generally can be defined as
	\begin{equation}
		\text{TIC}^{k}=-2\ln p\left( \mathbf{y}|\widehat{\boldsymbol{\theta}}_{k}\right)
		+2P_{T}^{k},P_{T}^{k}=-\mathbf{tr}\left\{ \mathbf{\bar{\Omega}}_{n}\left( 
		\widehat{\boldsymbol{\theta}}_{k}\right) \mathbf{\bar{H}}_{n}^{-1}\left(
		\widehat{\boldsymbol{\theta}}_{k}\right) \right\} .  \label{TICeq}
	\end{equation}
	From decision viewpoint, when $a=1$, the MLE, $\widehat{\boldsymbol{\theta}}_{k}$, replaced the
	Bayesian estimator, ${\mbox{\boldmath${\bar{\theta}}$}}_{k}$ to formulate
	the regular plug-in predictive distribution for constricting the risk function. Then, for misspecified models, it can be also shown in \citet{li2020deviance} that 
	\begin{equation*}
		Risk(d_{k^{1}})=E_{\mathbf{y}}\left( \mathcal{L}(\mathbf{y}%
		,d_{k^{1}})\right) =\int \mathcal{L}(\mathbf{y},d_{k^{1}})g(\mathbf{y})%
		\mbox{d}\mathbf{y}=E_{\mathbf{y}}\left[ \text{TIC}^{k}+2C\right] +o(1).
	\end{equation*}
	Furthermore, when the candidate models are restricted into correctly
	specified models or good models which are good approximation to DGP, TIC is reduced as the well-known AIC and it can
	be shown in \citet{li2024deviance} that 
	\begin{equation*}
		Risk(d_{k^{1}})=E_{\mathbf{y}}\left( \mathcal{L}(\mathbf{y}%
		,d_{k^{1}})\right) =\int \mathcal{L}(\mathbf{y},d_{k^{1}})g(\mathbf{y})%
		\mbox{d}\mathbf{y}=E_{\mathbf{y}}\left[ \text{AIC}^{k}+2C\right] +o(1),
	\end{equation*}
	where 
	\begin{equation}
		\text{AIC}^{k}=-2\ln p(\mathbf{y}|\widehat{\boldsymbol{\theta}}_{k},M_{k})+2P^{k}
		\label{aic2}
	\end{equation}
\end{remark}
More details about the theoretical development of $\text{DIC}^{k}$
and $\text{DIC}_{M}^{k}$, one can refer to \citet{takeuchi1976distribution} , \citet{li2020deviance}, \citet{li2024deviance} and reference therein. 

\subsection{Information Criterion for Comparing Misspecified Models based on
Variational Bayes Plug-in Predictive Distributions}
\label{subsec:plugin vbic}

Following the statistical decision theory shown in section \ref{subsec:decision theory}, we utilize $\ln p(\mathbf{y}|{\mbox{\boldmath${\bar{\theta}}$}}
_{k}^{VB},M_{k})$ to construct the loss function and the corresponding risk function. Subsequently, similar to existing information criteria such as AIC, TIC, DIC and $\text{DIC}_M$,we propose a new information criterion for model selection. Let $%
\mathbf{\bar{\Omega}}_{n}\left( {\mbox{\boldmath${		\widehat{\theta}}$}}%
_{k}\right) $ , $\mathbf{\bar{H}}_{n}\left({\ \mbox{\boldmath${\widehat{%
\theta}}$}}_{k}\right)$ be consistent estimators of $\mathbf{B}_{n}({%
\mbox{\boldmath${\theta}$}}_{n}^{p})$ and $\mathbf{H}_{n}({%
\mbox{\boldmath${\theta}$}}_{n}^{p})$ respectively. Based on the results of \citet{han2019statistical} and \citet{Zhang_2024}, we have%
\begin{equation*}
{\mbox{\boldmath${\bar{\theta}}$}}_{k}^{VB}=\boldsymbol{\widehat{\theta}}
_{k}+O_{p}\left(n^{-3/4}\right),
\end{equation*}
where ${\mbox{\boldmath${\bar{\theta}}$}}_{k}^{VB}$ is the mean of
variational posterior density $p^{VB}\left( \mbox{\boldmath${\theta}$}|%
\mathbf{y}\right) $. 
Using this, we derive the consistent estimators of 
$\mathbf{B}_{n}({\mbox{\boldmath${\theta}$}}_{n}^{p})$ and $\mathbf{H}_{n}({%
\mbox{\boldmath${\theta}$}}_{n}^{p})$ as $\mathbf{\bar{\Omega}}_{n}\left( {%
\mbox{\boldmath${\bar{\theta}}$}}_{k}^{VB}\right) $ and $\mathbf{\bar{H}}%
_{n}\left( {\ \mbox{\boldmath${\bar{\theta}}$}}_{k}^{VB}\right) $, respectively. To account for model misspecification, we define a new information criterion, termed the Variational Deviance Information Criterion under Model Misspecification ($\text{VDIC}_M$), using the variational plug-in predictive density:
\begin{equation*}
\text{VDIC}_{M}^k=-2\ln p(\mathbf{y}|{\mbox{\boldmath${\bar{\theta}}$}}
_{k}^{VB},M_{k})+2 P^{k}_{VDIC_M},
\end{equation*}
where the penalty term $%
P^k_{VDIC_M}$ for model $k$ is defined as 
\begin{equation*}
\begin{aligned} P^{k}_{VDIC_M} = -\mathbf{tr}\left[
\mathbf{\bar{\Omega}}_{n}\left(
{\mbox{\boldmath${\bar{					\theta}}$}}_{k}^{VB}\right) \left(
\mathbf{\bar{H}}_{n}\left( {
\mbox{\boldmath${\bar{\theta}}$}}_{k}^{VB}\right) \right) ^{-1}\right].
\end{aligned}
\end{equation*}

\begin{theorem}
\label{theom vtic} Under Assumptions 1-8, we have, 
\begin{equation*}
Risk(d_{k^1})=\int \text{VDIC}_M^{k}\times g(\mathbf{y})d\mathbf{y}+ 2C + o(1), 
\text{ i.e., }E_{\mathbf{y}}(\text{VDIC}_M^{k})=Risk(d_{k^1}) -2C + o(1).
\end{equation*}
\end{theorem}

It can be proved that $\text{VDIC}_M^{k}$ is an asymptotically unbaised estimator of $Risk(d_{k^1})$ up to a constant.
\begin{remark}
For $\text{VDIC}_M$, $-2\ln p(\mathbf{y}|{\mbox{\boldmath${\bar{\theta}}$}}
_{k}^{VB},M_{k})$ can be understood as a Bayesian measure of fit, while $2
P^{k}_{VDIC_M}$ measures the model complexity. This feature of trade-off
between the goodness of fit of the model and the complexity of the model is
shared by other information criteria, such as \text{TIC} and $\text{DIC}_M$.
\end{remark}

\begin{remark}
Similar to \text{TIC} and $\text{DIC}_M$, $\text{VDIC}_M$ works for both correctly specified
and misspecified models.
\end{remark}

\subsection{Information Criterion for Comparing Misspecified Models based on
the VB Posterior Predictive Distribution}
\label{subsec:posterior vbic}

Following the statistical decision theory outlined in Section \ref{subsec:decision theory}, we utilize \(p^{VB}(\mathbf{y}_{rep}|\mathbf{y})\) to construct the loss function and the corresponding risk function. Based on this posterior predictive distribution, a new information criterion can be developed to estimate \(\text{Risk}(d_{k^2})\).

Let \(\mathbf{\bar{\Omega}}_{n}({\boldsymbol{\bar{\theta}}}_k^{VB})\) and \(\mathbf{\bar{H}}_{n}({\boldsymbol{\bar{\theta}}}_k^{VB})\) be consistent estimators of \(\mathbf{B}_{n}({\boldsymbol{\theta}}_n^{p})\) and \(\mathbf{H}_{n}({\boldsymbol{\theta}}_n^{p})\), respectively. The consistent estimator of \(\mathbf{C}_{n}\) is given by:
\[
\mathbf{\hat{C}}_{n}({\boldsymbol{\bar{\theta}}}_k^{VB}) = \left(\mathbf{\bar{H}}_{n}({\boldsymbol{\bar{\theta}}}_k^{VB})\right)^{-1} \mathbf{\bar{\Omega}}_{n}({\boldsymbol{\bar{\theta}}}_k^{VB}) \left(\mathbf{\bar{H}}_{n}({\boldsymbol{\bar{\theta}}}_k^{VB})\right)^{-1},
\]
where \({\boldsymbol{\bar{\theta}}}_k^{VB}\) represents the mean of the variational posterior density \(p^{VB}(\boldsymbol{\theta}|\mathbf{y})\).

When accounting for model misspecification, we define a new information criterion based on the VB posterior predictive density, termed the Variational Predictive Information Criterion (\(\text{VPIC}\)):
\begin{equation*}
\text{VPIC}^k = -2 \ln p(\mathbf{y}|{\boldsymbol{\bar{\theta}}}_k^{VB}, M_k) + 2 P_{VPIC}^k,
\end{equation*}
where the penalty term \(P_{VPIC}^k\) for model \(k\) is defined as:
\[
\begin{aligned}
P_{VPIC}^k = &\frac{1}{2} \mathbf{tr}\left[\mathbf{\bar{\Omega}}_{n}({\boldsymbol{\bar{\theta}}}_k^{VB}) \left(-\mathbf{\bar{H}}_{n}({\boldsymbol{\bar{\theta}}}_k^{VB})\right)^{-1}\right] \\
+ &\frac{1}{2} \ln \left\lvert \left(-\mathbf{\bar{H}}_{n}({\boldsymbol{\bar{\theta}}}_k^{VB})\right) \left(-\mathbf{\bar{H}}_{n}^d({\boldsymbol{\bar{\theta}}}_k^{VB})\right)^{-1} + \mathbf{I}_n \right\rvert \\
- &\frac{1}{2} \mathbf{tr}\Bigg[ \left(-\mathbf{\bar{H}}_{n}({\boldsymbol{\bar{\theta}}}_k^{VB}) + \left(-\mathbf{\bar{H}}_{n}^d({\boldsymbol{\bar{\theta}}}_k^{VB})\right)\right)^{-1} \\
&\hspace{1cm} \times \left( \mathbf{\bar{\Omega}}_{n}({\boldsymbol{\bar{\theta}}}_k^{VB}) + \left(-\mathbf{\bar{H}}_{n}^d({\boldsymbol{\bar{\theta}}}_k^{VB})\right) \mathbf{\hat{C}}_{n}({\boldsymbol{\bar{\theta}}}_k^{VB}) \left(-\mathbf{\bar{H}}_{n}^d({\boldsymbol{\bar{\theta}}}_k^{VB})\right) \right)\Bigg] \\
+ &\frac{1}{2} \mathbf{tr}\left[\left(-\mathbf{\bar{H}}_{n}^d({\boldsymbol{\bar{\theta}}}_k^{VB})\right) \mathbf{\hat{C}}_{n}({\boldsymbol{\bar{\theta}}}_k^{VB})\right].
\end{aligned}
\]

\begin{theorem}
\label{theom2} Under Assumptions 1-8, we have,%
\begin{equation*}
Risk(d_{k^2})=\int \text{VPIC}^{k}\times g(\mathbf{y})d\mathbf{y} +2C +o(1),\text{
i.e., }E_{\mathbf{y}}(\text{VPIC}^{k})=Risk(d_{k^2})-2C +o(1)
\end{equation*}
\end{theorem}
It can be proved that $\text{VPIC}^{k}$ is an asymptotically unbaised estimator of $Risk(d_{k^2})$ up to a constant.

\begin{remark}
For $\text{VPIC}^{k}$, $-2\ln p(\mathbf{y}|{\mbox{\boldmath${\bar{\theta}}$}}%
_{k}^{VB},M_{k})$ can be understood as a Bayesian measure of fit, while $2
P^{k}_{VPIC}$ measures the model complexity. $\text{VPIC}^{k}$ works for both correctly specified
and misspecified models.
\end{remark}

\subsection{BF and BIC type information criteria}
\label{subsec: bf and bic}
The BF and BIC belong to the first strand of model comparison in the section \ref{sec:intro}. They compare competing models by examining model posterior probabilities and search for the``true" model. Both BFs and BIC enjoy the property of consistency, that is, when the true DGP is one of the candidate models, BFs and BIC select it with probability approaching 1 when the sample size goes to infinity.

Suppose there are two candidate models, $M_1$ and $M_2$. The BF of $M_1$ against $M_2$ is defined as $B_{12}=\frac{p\left(\mathbf{y} | M_1\right)}{p\left(\mathbf{y} | M_2\right)}$, 
where $p\left(y | M_k\right)$ is the marginal likelihood of model $M_k$ which is obtained by
$$
p\left(\mathbf{y} | M_k\right)=\int_{\Theta_k} p\left(\mathbf{y} | \boldsymbol{\theta}_k, M_k\right) p\left(\boldsymbol{\theta}_k | M_k\right) \mathrm{d} \boldsymbol{\theta}_k, \quad \boldsymbol{\theta_k} \in \Theta_k, k=1,2,
$$
where $\boldsymbol{\theta}_k$ is the set of parameters in $M_k, p\left(\mathbf{y} | \boldsymbol{\theta}_k, M_k\right)$ the likelihood function of $M_k, p\left(\boldsymbol{\theta}_k \mid M_k\right)$ the prior of $\boldsymbol{\theta}_{\boldsymbol{k}}$ in $M_k$. If $B_{12}>1$, $M_1$ is preferred to $M_2$ and vice versa.

Based on the Laplace approximation, \citet{schwarz1978estimating} showed that the log-marginal likelihood can be approximated by
\begin{equation}
\label{bic_eq}
\begin{aligned}
\ln p\left(\mathbf{y} | M_k\right)=&\ln p\left(\mathbf{y} | \hat{\theta}_k, M_k\right)+\ln p\left(\hat{\boldsymbol{\theta}}_k | M_k\right)\\&+\frac{P_k \pi}{2}-\frac{P_k \ln n}{2}-\frac{\left|-\overline{\mathbf{H}}_n\left(\hat{\boldsymbol{\theta}}_k\right)\right|}{2}+O_p\left(\frac{1}{n}\right),
\end{aligned}
\end{equation}
where $\hat{\boldsymbol{\theta}}_k$ is the MLE of $\boldsymbol{\theta}_k$ and $\overline{\mathbf{H}}_n\left(\hat{\boldsymbol{\theta}}_k\right)$, and $P_k$ is the dimension of $\boldsymbol{\theta}_k$. Ignoring all the $O_p(1)$ terms in (\ref{bic_eq}) and under noninformative priors such as $p\left(\boldsymbol{\theta}_k | \boldsymbol{M}_k\right) \propto 1$, Schwarz defined $\mathrm{BIC}_k$ as
$\mathrm{BIC}_k =-2 \ln p\left(\mathbf{y}|\hat{\boldsymbol{\theta}}_k, M_k\right)+P_k \ln n$,
where, as in AIC and TIC, $-2 \ln p\left(\mathbf{y} |  \hat{\boldsymbol{\theta}}_k, M_k\right)$ is used to measure the model fit, but $P_k \ln n$ is the new penalty term. Obviously, $\mathrm{BIC}_k$ provides an approximation of $-2 \ln \left(\mathbf{y} | M_k\right)$.

Recently, \citet{Zhang_2024} showed that under regular conditions, the difference between the evidence lower bound, which is the by-product of VB algorithm, and $\mathrm{-BIC/2}$, is asymptotically constant as $n$ goes to infinity.

\begin{remark}
From the theoretical viewpoint, different criteria have different theoretical properties. BIC and BFs are consistent if the true model is one of the candidate models while $\mathrm{AIC}$, $\mathrm{TIC}$, $\mathrm{DIC}$, $\mathrm{DIC}_M$, $\mathrm{VDIC}_M$ and $\mathrm{VPIC}$ aim to provide the asymptotically unbiased estimator of the expected KL divergence between the DGP and a predictive distribution. When the true model is not included as a candidate model, which is often the case in practice, it is not clear what the best model selected by BIC and BFs can achieve. In this case, if one is concerned with the KL divergence between the DGP and a predictive distribution, it is expected that $\mathrm{TIC}$, $\mathrm{DIC}_M$ $\mathrm{VDIC}_M$ and $\mathrm{VPIC}$ perform better than BIC and BFs. Moreover, when the sample size is small, even when the true model is a candidate model, BIC and BFs may not select the true model. Again, if one is concerned with the KL divergence between the DGP and a predictive distribution, AIC and $\mathrm{DIC}$ can perform better than BIC and BFs. If one is considering model with massive data, in which MLE or MCMC methods can be intractable or costly, $\mathrm{VDIC}_M$ and $\mathrm{VPIC}$ will perform better. 
\end{remark}

\section{Simulation and Empirical Studies}
\label{sec:simu and real}

\subsection{Simulation Study}

We begin by using two numerical simulation examples to evaluate the performance of our newly proposed criteria in the context of massive data. Both examples involve model misspecification. In the first study, we use polymomial regression to fit a nonlinear model, aming to select the model with the best predicitve among candidate models. Similarly, in the second study, we focus on identifying the "best" model among four candidate probit models. For each scenario, we conduct 1000 replications and apply our two newly developed information criteria to assess their effectiveness with the ELBO criterion proposed by \citet{Zhang_2024}.

In every experiment, we simulate $\mathbf{y}_i$ and calculate $\text{VPIC}^k$, $\text{VDIC}_M^k$, $\text{ELBO}^k$, $\text{AIC}^k$ and $\text{BIC}^k$ of candidate model $M_k, k=1,\dots,K$. Each of the five criteria is used to selected a best model (call it $M_{k_i^*}$), we then record this model and the corresponding optimal criteria $IC(i)$. For $\text{VDIC}_M^k$, we use the VB plug-in predictive distribution $p^{VB}(\mathbf{y}_{rep}|\boldsymbol{\bar{\theta}}^{VB}_{k^{*}},M_{k^{*}})$ under the best model $M_{k*}$ to predict new data. Then we can estimate the risk by
\begin{equation*}
\widehat{Risk\left(d_{k_*^1}\right)} = \frac{1}{1000} \sum_{i=1}^{1000} IC_{k^*}\left(%
\mathbf{y}_i\right), \text{ for } \text{VDIC}_M^k,
\end{equation*} 
where $Risk\left(d_{k_*^1}\right)=E_{\mathbf{y}}\left[\mathcal{L}\left(\mathbf{y}, d_{k^*}\right)\right]=E_{\mathbf{y}}\left[2\times KL \left[g(\mathbf{y}_{rep}),p(\mathbf{y}_{rep}|\mathbf{y},M_{k^*},d_1)\right]\right]$ 

For $\text{VPIC}^k$, we use the VB posterior predictive distribution $p^{VB}(\mathbf{y}_{rep}|\mathbf{y},M_{k*})$ under the best model $M_{k*}$ to predict new data. Then we can estimate the risk by
\begin{equation*}
\widehat{Risk\left(d_{k_*^2}\right)} = \frac{1}{1000} \sum_{i=1}^{1000} IC_{k^*}\left(%
\mathbf{y}_i\right), \text{ for } \text{VPIC}^k,
\end{equation*}
Same risk is calculate to estimate the risk of $\text{AIC}^k$.

For $\text{ELBO}^k$ and $\text{BIC}^k$, we will use two proxies to evaluate its risk. As \citet{Zhang_2024} noted, under some regular conditions, the difference between $-\text{BIC}^k/2$ and $\text{ELBO}^k$ is asymptotically to be constant as 
$n$ goes to infinity. For the reasons that $\text{BIC}$ is constructed as an approximation of the marginal likelihood $p(\mathbf{y})$, not from predictive perspective, averaging for $-2\times\text{ELBOs}$ and $\text{BIC}$ in all replication as the risk of both $\text{ELBO}$ and $\text{BIC}$ is not a proper way. We will use two proxies to see the relative risk of $\text{ELBO}$. In each experiment, when choosing the best model $M_{k_i^*}$ under $\text{ELBO}$ or $\text{BIC}$, we will use both $\text{VDIC}_M^{k_i^*}$ and $\text{VPIC}^{k_i^*}$ whose expectation is the KL loss as proxy. Then we can estimate the risk of $\text{ELBO}$ and $\text{BIC}$ by
\begin{equation*}
    \begin{aligned}
        \widehat{Risk(d_{k^*})}_1 &= \frac{1}{1000}\sum_{i=1}^{1000} IC_{k^*}(\mathbf{y}_i), \text{ IC is } \text{VDIC}_M^k, \text{ and }\\
        \widehat{Risk(d_{k^*})}_2 &= \frac{1}{1000}\sum_{i=1}^{1000} IC_{k^*}(\mathbf{y}_i), \text{ IC is } \text{VPIC}^k,
    \end{aligned}
\end{equation*}
named as $\text{ELBO1}$, $\text{ELBO2}$, $\text{BIC1}$, and $\text{BIC2}$.

\subsubsection{Polymomial Regression}
We begin with a simple experiment to compare alternative model selection criteria when the true DGP is not included in the set of candidate models. In other words, all candidate models are misspecified. Following Ding et al. (2019), we generate data from the following model
$$
y_i=\ln \left(1+46 x_i\right)+e_i, e_i \sim N(0,1), i=1, \ldots, N,
$$
where $x_i=0.7(i-1) / n$ which is fixed under repeated sampling by design. In practice, researchers do not know the functional form. Suppose the following set of polynomial regressions is considered,
$$
M_k: y_i=\sum_{j=0}^{k-1} \beta_{k, j+1} x_i^j+u_i
$$
where $k=1, \ldots,\left\lfloor \ln\left(N\right)\right\rfloor$ and $u_i$ is assumed to be i.i.d. $N\left(0, \sigma^2\right)$. When $k \rightarrow \infty$ as $N \rightarrow \infty$, the polynomial regression is related to the sieve estimator which uses progressively more complex models to estimate an unknown function as more data becomes available. In our experiment, we estimate and compare all the candidate models $\left\{M_k, k=1, \ldots,\left\lfloor \ln\left(n^{3/4}\right) \right\rfloor\right\}$. Let $\mathrm{x}^j=\left(x_1^j, x_2^j, \ldots, x_N^j\right)^{\prime}, \mathbf{X}_k=\left(\mathrm{x}^0, \mathrm{x}^1, \ldots, \mathrm{x}^{k-1}\right)$, and $\mathbf{X}=\left(\mathrm{x}^0, \mathrm{x}^1, \ldots, \mathrm{x}^{\left[\ln(N)\right]-1}\right)$. In $M_k$, function $f\left(\boldsymbol{\beta}_k,\mathbf{X}_k\right) = \sum_{j=0}^{k-1} \beta_{k, j+1} x_i^j$ is used to approximate $\ln \left(1+46 x_i\right)$. Let $\boldsymbol{\beta}_k=\left(\beta_1, \ldots, \beta_k\right)^{\prime}$ so that $\boldsymbol{\theta}_k=\left(\boldsymbol{\beta}_k^{\prime}, \sigma^2\right)$, and the number of parameters is $k+1$.  

For Bayesian analysis, we assign priors to $\boldsymbol{\beta}_k$ and $\sigma^2 = h^{-1}$ as follows:
\[
\boldsymbol{\beta}_k \sim N(\tilde{\mu}, h^{-1} \tilde{V}), \quad h \sim \text{Gamma}(a, b),
\]
the hyperparameters of the priors are set as \(a = 1\), \(b = 1\), \(\tilde{\mu} = \mathbf{0}\), and \(\tilde{V} = 10^5 \times \mathbf{I}_k\). The optimal VB posterior of $\beta$ and 
$h$, which is $q(\beta, h)=q(\beta)q(h)$, approximats the true posterior $p\left(\beta,h|y\right)$, see more details in appendix B.1.

In the simulation study, the sample size varies from \(N = 500\) to \(N = 1,000,000\). For each sample size, we simulate the DGP 1000 times. In the \(i\)-th replication, a dataset of size \(N\) is simulated, and the values of \(\text{VPIC}^k\), \(\text{VDIC}_M^k\), \(\text{ELBO}^k\), \(\text{AIC}^k\) and \(\text{BIC}^k\) are computed for the candidate models \(M_k, k = 1, \dots, \left\lfloor \ln\left(N\right) \right\rfloor\).

\begin{figure}
    \centering
    \includegraphics[scale=0.4]{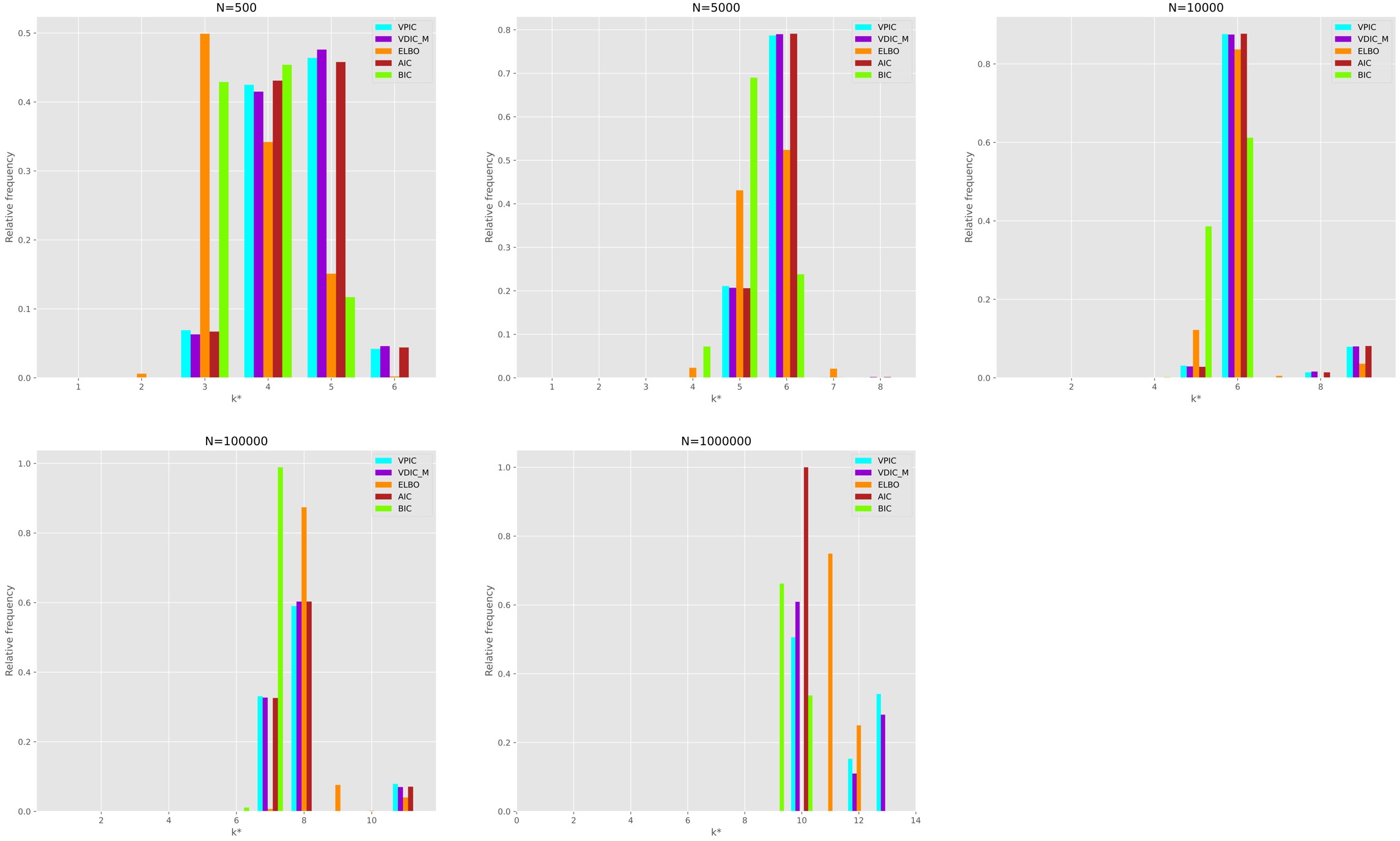}
    \caption{The figure plots relative frequencies of the polynomial orders selected by different
criteria.}
    \label{fig:poly_relative_k}
\end{figure}

\begin{table}[ht!]
    \small
    \centering
    \caption{Averge $k^*$ selected under different criteria}
    \begin{tabular}{cccccc}
    \hline N & VPIC & $\text{VDIC}_M$ & ELBO & AIC & BIC \\
    \hline 500 & 4.479 & 4.505 & 3.644 & 4.479 & 3.688 \\
    5,000 & 5.792 & 5.798 & 5.546 & 5.799 & 5.166 \\
    10,000 & 6.234 & 6.243 & 5.991 & 6.243 & 5.610 \\
    100,000 & 7.906 & 7.883 & 8.191 & 7.887 & 6.989 \\
    1,000,000 & 11.329 & 11.063 & 11.252 & 10.000 & 9.335 \\
    \hline
    \end{tabular}
    \label{optimal k poly}
\end{table}

The relative frequencies of the selected models by each of three criteria (namely VPIC, $\text{VDIC}_M$, ELBO, AIC and BIC) are reported in Figure \ref{fig:poly_relative_k}. And the average values of $k^*$ is listed in the table \ref{optimal k poly}, all across 1,000 replications. Several interesting results can be found in Figure \ref{fig:poly_relative_k}. First, the models selected by the ELBO and $\text{BIC}$ tend to be parsimonious than those selected by $\text{VPIC}$, $\text{VDIC}_M$ and $\text{AIC}$, this result is not surprising as BIC has a larger penalty term than AIC. Second, as $N$ increases, the average $k^*$ s selected by VPIC and $\text{VDIC}_M$ tends to be similar, suggested that they tend to select the same model. Though under regular conditions, the difference between $\text{BIC}$ and $\text{ELBO}$ are constant as $N$ goes infinity, in our simulation, we find that the averge $k^*$ selected by of $\text{BIC}$ and $\text{ELBO}$ tends to be different as $N$ increases. Third, as the sample size increases, the average $k^*$ s selected by all criteria tend to increase. This is not surprising as the true DGP is not a candidate model.

\begin{table}[ht!]
\small
\centering
\caption{Average risk of different criteria using polymomial regression (Scaled)}
\label{vbic for polymomial models}
\par
    \begin{tabular}{cccccccc}
    \hline & VPIC & $\text{VDIC}_M$ & ELBO1 & ELBO2 & AIC & BIC1 & BIC2 \\
    \hline 500 & \textcolor{red}{1.42124} & 1.42177 & 1.42297 & 1.42355 & 1.42190 & 1.42268 & \textcolor{blue}{1.42326} \\
    5000 & \textcolor{red}{14.19515} & 14.19565 & 14.19553 & 14.19605 & 14.19566 & 14.19710 & \textcolor{blue}{14.19764} \\
    10000 & \textcolor{red}{28.38818} & 28.38866 & 28.38837 & 28.38886 & 28.38867 & 28.38987 & \textcolor{blue}{28.39038} \\
    100000 & \textcolor{red}{283.82618} & 283.82660 & 283.82679 & 283.82715 & 283.82659 & 283.82735 & \textcolor{blue}{283.82779} \\
    1000000 & \textcolor{red}{2837.99446} & 2838.00302 & 2838.01325 & \textcolor{blue}{2838.01344} & 2838.00639 & 2838.01196 & 2838.01227 \\
    \hline
    \end{tabular}
\end{table}

Table \ref{vbic for polymomial models} reports the results of risks. We report $(\widehat{\operatorname{Risk}}-1-\ln (2 \pi))$ scaled by $10^3$ instead of $\widehat{\operatorname{Risk}}$ to better highlight differences in the risks under different criteria. We focus on the risk of $\mathrm{VPIC}$, $\mathrm{VDIC}_M$, $\mathrm{AIC}$ and two proxies of risk $\mathrm{ELBO}$ and $\mathrm{BIC}$. In our simulation experiment, $\mathrm{\text{VDIC}_M}$ and $\mathrm{\text{VPIC}}$ have smaller risks than $\mathrm{ELBO}$, $\mathrm{AIC}$ and $\mathrm{BIC}$. The most important result from Table \ref{vbic for polymomial models} is that $\text{VPIC}$ leads to a much smaller value of the expected KL divergence than the other criteria. Results obtained from this Monte Carlo study indicate that if one's objective is to get a best prediction for the future data, we should not only consider how to choose a "best" model and estimator the parametric in this model. We should take predictive distribution into consideration, that means we should use $\mathrm{\text{VDIC}_M}$ and $\mathrm{VPIC}$, not only one criterion, compare these criteria and get a minimum. Then we choose this "optimal" model and use the corresponding predictive distribution to predict the future data.

\subsubsection{Probit Regression}

In this subsection, we report a generalized linear model (GLM) example, using probit regression. We have linear predictor $Z_i=X_i^{\prime} \beta$ based on vector $X_i$, and we choose the probit link $g\left(E\left[Y_i|X_i\right]\right)=g\left(p_i\right)=Z_i$ as link function, 
the inverse of the link function $g^{-1}\left(\cdot\right)=\Phi\left(\cdot\right)$ is the cumulative distribution function (cdf) of standard normal distribution, it is shown that
\begin{equation}
Y_i \mid X_i \overset{i . i . d .}{\sim} Bernoulli\left(\Phi%
\left(X_i^{\prime} \beta \right)\right),  \label{llk of probit}
\end{equation}
where $\beta$ is $p\times1$ vector. For the Bayesian analysis, we assume a normal prior \(\beta \sim N(\tilde{\mu}, \tilde{V})\), where $\tilde{\mu}  = \mathbf{0}$ and $\tilde{V} = 10^5 \times \mathbf{I}_p$, then employ the mean-field VB method to derive the optimal VB posterior distribution \(q(\beta)\). For further details, refer to Appendix B.2
%\ref{mean feild probit}.

In this simulation study, we define the DGP as $p=4$, $\beta =
\left(\beta_0, \beta_1, \beta_2, \beta_3\right)^{\prime}$ with $\beta_0=-0.2$, $\beta_1=0.3$, $%
\beta_2=0$, $\beta_3 = 0.7$, $X_i = \left(1, x_{i1},x_{i2}, x_{i3}\right)^{\prime}$, and $%
N$ ranging from $N=500$ to $N=1,000,000$. We define such a model to simulate the scenario of under-fitting and over-fitting. Similarly to the first simulation study,  
we consider seven candidate models, as detailed below, 
\begin{table}[ht!]
\small
\centering
\caption{Candidate models of probit simulated data}
\label{candidate model of simulation2}
\begin{tabular}{cccc}
\hline Model & Numbers of variable & Model & Model specification\\
\hline $M_1$ & 1 & $Z_i=\left(1,x_{i1}\right)^{\prime}$ & Underfitting\\
$M_2$ & 1 & $Z_i=\left(1,x_{i2}\right)^{\prime}$ & Underfitting\\
$M_3$ & 1 & $Z_i=\left(1,x_{i3}\right)^{\prime}$ & Underfitting\\
$M_4$ & 2 & $Z_i=\left(1,x_{i1},x_{i2}\right)^{\prime}$ & Underfitting\\
$M_5$ & 2 & $Z_i=\left(1,x_{i1},x_{i3}\right)^{\prime}$ & Correctly specified\\
$M_6$ & 2 & $Z_i=\left(1,x_{i2},x_{i3}\right)^{\prime}$ & Underfitting\\
$M_7$ & 3 & $Z_i=\left(1,x_{i1},x_{i2},x_{i3}\right)^{\prime}$ & Overfitting\\
\hline
\end{tabular}
\end{table}

We replicated DGP for 1000 times, in the $i^{th}$
replication, we generate the data with sample size $N$, and calculate $\text{VPIC}^k$, $\text{VDIC}_M^k$, $\text{ELBO}^k$, $\text{AIC}^k$ and $\text{BIC}^k$ with $M_k=M_1,\dots,M_7$. Then we compare the performance of these criteria.

\begin{table}[ht!]
\small
\centering
\caption{Average risk of different criteria using probit regression (Scaled)}
\label{vbic for probit simu}
\par
\begin{tabular}{cccccccc}
    \hline & VPIC & $\text{VDIC}_m$ & ELBO1 & ELBO2 & AIC & BIC1 & BIC2 \\
    \hline 500 & \textcolor{red}{0.59530} & 0.59624 & 0.59824 & \textcolor{blue}{0.59907} & 0.59626 & 0.59554 & 0.59643 \\
    5,000 & \textcolor{red}{5.96466} & 5.96539 & 5.96474 & \textcolor{blue}{5.96562} & 5.96539 & 5.96473 & 5.96561 \\
    10,000 & \textcolor{red}{11.92945} & 11.93039 & 11.92982 & \textcolor{blue}{11.93070} & 11.93040 & 11.92977 & 11.93065 \\
    100,000 & \textcolor{red}{119.30140} & 119.30233 & 119.30170 & \textcolor{blue}{119.30258} & 119.30233 & 119.30169 & 119.30257 \\
    1,000,000 & \textcolor{red}{1193.04493} & 1193.04586 & 1193.04522 & \textcolor{blue}{1193.04610} & 1193.04586 & 1193.04522 & \textcolor{blue}{1193.04610} \\
    \hline
\end{tabular}
\end{table}

Table \ref{vbic for probit simu} presents the average risk associated with two different information criteria for seven candidate models under \(N\) ranging from 500 to 1,000,000. Each column in the table reports the risk when choosing the optimal candidate model \(M_{k^*}\). The risk of \(\text{VPIC}^k\) is consistently lower than that of \(\text{VDIC}_M^k\). Also, same like the results in the first simulation, our Monte Carlo experiment has shown that predictive risk under choosing the optimal candidate model from Variational predictive distribution is lower than that from variational lower bound. 

\subsection{Empirical Studies}
\begin{comment}
In this subsection, we first use the linear model of different covariates to
investigate which model best predict the number of passengers transported in
the flight. In the second real study, we then investigate the credit risk
model, which is usually modelled by binary classification. In this study, we will use three new critera to choose  which is the best model to explain the reason why the borrowers get loan from platform among six candidate probit models.
\end{comment}
In this subsection, we first analyze a linear model with different covariates to identify the model that best predicts the number of passengers transported by flight. In the second study, we examine a credit risk model, typically formulated as a binary classification problem. These real data studies aim to show the performance of our two proposed new criteria, and to present that these VB based information criteria can well behave under big data analysis.  

\subsubsection{US Domestic Flights Predictive Model}
\begin{comment}
In this section we use linear model with different covariates to investigate
which model best predict the number of passengers transported in the flight.
This dataset is related to the US domestic flights from 1990 to 2009, which
contains about $N=3.61$ Million data. Dataset we use is avalibale on Kaggle.
\citet{Chasiotis_Karlis_2024} use this data to fit 
the linear regression model where  $p=5$ measurements are selected. In this study, we can  use the linear regression to model the relationships between
the dependent variable and covariates.
\end{comment}
In this section, we analyze a linear model with different covariates to identify the model that best predicts the number of passengers transported by flight. The data set used in this analysis pertains to US domestic flights from 1990 to 2009 and contains approximately \(N = 3.61\) million observations. This data set is publicly available on Kaggle.  \citet{Chasiotis_Karlis_2024} employed this dataset to fit a linear regression model, selecting \(p = 5\) measurements as covariates. In this study, we utilize linear regression to explore the relationships between the dependent variable PASSENGERS (number of passengers, $y$) and the selected covariates, including SEATES (number of seats available on flight, $x_1$), FLIGHTS (number of flights between two locations, $x_2$), DISTANCE (distance flown between origin and destination, $x_3$), ORIGIN POP (origination city's population, $x_4$), DESTINATION POP (destination city's population, $x_5$), ORIGIN LONG (origination airport longitude, $x_6$), DESTINATION LONG (destination airport longitude, $x_7$), ORIGIN LAT (origination airport latitude, $x_8$), DESTINATION LAT (origination airport latitude, $x_9$). To conduct the model selection problem of this dataset, we consider four candidate models, and we list the candidate models and related considerations.  

\begin{table}[ht!]
\small
\centering
\caption{Candidate model set for US domestic flights data}
\label{model descrpition in linear empirical}
\begin{tabular}{ccc}
\hline
Model & Description & Number of covariate\\ \hline
$M_1$ & $Y_i = \beta_0+\beta_1 x_1 + \beta_2 x_2 + \epsilon_i$ & 2\\
$M_2$ & $Y_i = \beta_0+\beta_1 x_1 + \beta_2 x_2 + \beta_3 x_3 +\epsilon_i$ & 3\\ 
$M_3$ & $Y_i = \beta_0+\beta_1 x_1 + \beta_2 x_2 + \beta_3 x_3 + \beta_4 x_4+ \beta_5 x_5 + \epsilon_i$ & 5\\ 
$M_4$ & $Y_i = \beta_0+\beta_1 x_1 + \beta_2 x_2 + \beta_3 x_3 + \dots + \beta_9 x_9 + \epsilon_i$ & 9\\  \hline
\end{tabular}%
\end{table}

Table \ref{model descrpition in linear empirical} lists the variables we use in the linear regression model. For model comparison, we use mean-field VB to obtain the variational posterior estimators, and then compute the two new proposed information criteria $\text{VPIC}^{k}$ and $\text{VDIC}_M^{k}$ for all candidate models. In choosing the optimal model, to compare the performance of our new proposed method, with other commonly used criteria, we also report $\text{ELBO}^k$, $\text{AIC}^{k}$, $\text{BIC}^k$, $\text{DIC}^k$ and $\text{DIC}^k_M$.

\begin{comment}
Table \ref{Ic for air flight model} reports $\text{VPIC}^k$, $\text{VDIC}_M^k$
%, $VAIC_k$,
with corresponding penalty terms $P^k_{VPIC}$, $P^k_{VDIC_M}$, %$P_k^{VAIC}$, 
and logrithm likelihood $\mathcal{L}(\boldsymbol{\theta},k)$ of candidate models 
$\left\{M_k\right\}_{k=1}^4$. The difference of criteria is mainly caused by
the penalty term. First and most, both $\text{VPIC}^k$ and $\text{VDIC}_M^k$ choose the model $M_4$,
which means under the same criteria, we choose model $M_4$ as the better
model than other candidate models. What's more, among VPIC and $\text{VDIC}_M$, the $%
\text{VPIC}^k$ of $M_4$ is the minimum. Thus we suggest choose the model $M_4$ and
use the VB posterior predictive distribution to make decision to reach the minimum predictive risk.
\end{comment}

Table \ref{Ic for air flight model} presents the values of \(\text{VPIC}^k\) and \(\text{VDIC}_M^k\), along with $\text{ELBO}^k$ and conventional (or benchmark) information criteria like $\text{AIC}^k$, $\text{BIC}^k$, $\text{DIC}$ and $\text{DIC}_M$. For the candidate models \(\{M_k\}_{k=1}^4\). Importantly, both \(\text{VPIC}^k\) and \(\text{VDIC}_M^k\) select model \(M_4\), same as the benchmark information criteria, indicating that \(M_4\) is preferred over the other candidate models under the same criteria. Based on these results, we recommend selecting model \(M_4\) and using the VB posterior predictive distribution for decision-making to achieve the minimum predictive risk.
\begin{table}[ht!]
    \small
    \centering
    \caption{Model selection results of 4 candidate models in US air flight data}
    \label{Ic for air flight model}
    \begin{tabular}{ccccc}
    \hline & $\text{M}_1$ & $\text{M}_2$ & $\text{M}_3$ & $\text{M}_4$ \\
    \hline $\text{VPIC}$ & 60255994.7307 & 60127455.0509 & 60123759.8081 & \textcolor{red}{60113021.1802} \\
    $\text{VDIC}_M$ & 60256044.6951 & 60127505.9314 & 60123811.8674 & \textcolor{red}{60113073.3737} \\
    $\text{ELBO}$ & -54214029.0097 & -54149777.3111 & -54147985.5373 & -54142673.4084 \\
    $\text{AIC}$ & 60255940.7704 & 60127400.1087 & 60123703.9091 & 60112963.8209 \\
    $\text{BIC}$ & 60255993.1500 & 60127465.5832 & 60123795.5734 & 60113107. 8649 \\
    $\text{DIC}$ & 60255940.8067 & 60127400.1023 & 60123703.8001 & 60112963.8361 \\
    $\text{DIC}_M$ & 60256042.8225 & 60127506.9269 & 60123810.6099 & 60113072.8852 \\
    \hline
    \end{tabular}
\end{table}

To show the difference among seven information criteria, we report a more detailed summary, shown in table \ref{detailed info of us air flight}. As is known, common information criterion are constructed with two terms: one is the fit term $D(\boldsymbol{\theta})$ equals $-2*\ell\left(\boldsymbol{\theta}\right)$, where $\ell\left(\boldsymbol{\theta}\right)$ is the logarithm likelihood function, and penalty term $P_{IC}$ depend on different measures. If one conducts model selection under Bayes framework, one aims to use the true posterior mean of $\boldsymbol{\theta}$, which is $\bar{\boldsymbol{\theta}}$, or turn to use VB posterior mean $\bar{\boldsymbol{\theta}}^{VB}$ in fit term. Some results can be found in this table. First, as we report in the table, the difference between $\ell\left(\bar{\boldsymbol{\theta}}\right)$ and $\ell\left(\bar{\boldsymbol{\theta}}^{VB}\right)$ is very small, showing that the true posterior mean and VB posterior mean tend to converge to the same value as the size of observed data $N$ goes to infinity. It should also be noted that the inference time between $\bar{\boldsymbol{\theta}}$ and $\bar{\boldsymbol{\theta}}^{VB}$ differs in application, to obtain $\bar{\boldsymbol{\theta}}$ in this around 3 million data, we expend 7159.35 seconds using MCMC, however, 2.54 seconds is used to obtain $\bar{\boldsymbol{\theta}}^{VB}$ under VB as we recorded. As N becomes larger or a more complicated model incoming, one may have to turn to used VB based information criteria rather than using other criteria. Second, both the penalty term $P_{\text{VDIC}_M}$ and $P_{\text{DIC}_M}$ are similar, indicating that $\text{VDIC}_M$ behaves like $\text{DIC}_M$. In addition, $P_{DIC}$ is similar to $P_{AIC}$, as \citet{li2024deviance} showed that $\text{DIC}$ is a Bayesian version of $\text{AIC}$.

\begin{table}[ht!]
    \caption{Difference among fit term and penalty term}
    \label{detailed info of us air flight}
    \small
    \centering
    \begin{tabular}{ccccc}
    \hline & $\text{M}_1$ & $\text{M}_2$ & $\text{M}_3$ & $\text{M}_4$ \\
    \hline $\ell\left(\bar{\boldsymbol{\theta}}^{VB}\right)$ & -30127966.3852 & -30063695.0543 & -30061844.9545 & -30056470.9105 \\
    $\ell\left(\bar{\boldsymbol{\theta}}\right)$ & -30127966.3856 & -30063695.0544 & -30061844.9552 & -30056470.9115 \\
    $||\ell\left(\bar{\boldsymbol{\theta}}^{VB}\right)-\ell\left(\bar{\boldsymbol{\theta}}\right)||$ & 0.0003 & 0.0001 & 0.0007 & 0.0010 \\
    $P^k_{\text{VPIC}}$ & 30.98 & 32.47 & 34.95 & 39.68 \\
    $P^k_{\text{VDIC}_M}$ & 55.96 & 57.91 & 60.98 & 65.78 \\
    $P^k_{\text{AIC}}$ & 4 & 5 & 7 & 11 \\
    $P^k_{\text{BIC}}$ & 30.19 & 37.74 & 52.83 & 83.02 \\
    $P^k_{\text{DIC}}$ & 4.02 & 5.00 & 6.94 & 11.01 \\
    $P^k_{\text{DIC}_M}$ & 55.03 & 58.41 & 60.35 & 65.53 \\
    \hline
    \end{tabular}
\end{table}

\subsubsection{Credit Risk Analysis}
\begin{comment}
The credit risk model is the application of the Probit regression. It is
determined whether the loan should be given according to the borrower's
various information. When using binary classification, we will design $Y_i=1$
if borrower is allowed to give loans, and if $Y_i=0$, it means no loan. We
will use LendingClub dateset from Kaggle, ranging from 2007 to 2020 Q3.
\end{comment}
The credit risk analysis is an application of binary classification model, including probit regression and logistic regression, used to determine whether a loan should be granted based on various borrower-specific information. In the context of binary classification, we define \(Y_i = 1\) if a loan is approved for the borrower, and \(Y_i = 0\) if it is not.  For this study, we utilize the LendingClub dataset, which is publicly available on Kaggle. This data downloaded from Kaggle has about 3 million, and covers the period from 2007 to the third quarter of 2020. By referring filtering process in \href{https://www.kaggle.com/code/kevinkollcaku/loan-classification}{Loan Classification}, we finally got 1.74 million data points. 
%The dataset we use contains about 1.74 million data points 
%\textbf{The consolidated data sample size is about 3 million, through the feature engineering of the consolidated data, we finally got about 1.74 million valid data about information of loan borrowers.} (\textbf{Note:I don't understant these words}) 
Tabel \ref{varaible descrpition} lists independent varaible and dependent variables that we are interested in. 

\begin{table}[ht!]
\small
\centering
\caption{Variable description for credit risk model}
\label{varaible descrpition}
\begin{tabular}{ccc}
\hline
Variable & Symbol & Description \\ \hline
loan status & $Y_i$ & Current status of the loan, if loaned $Y_i=1$, else $%
Y_i=0$ \\ 
annual inc & $AnuI_i$ & Annual income provided by the borrower during
registration \\ 
emp length & $Emp_i$ & Employment length in years. \\ 
dti & $DTI_i$ & Debt-to-Income Ratio, excluding mortgage and the requested
LC loan \\ 
loan amount & $Loanam_i$ & The amount of the loan applied for by the borrower
\\ 
term & $Term_i$ & The number of payments on the loan. \\ \hline
\end{tabular}%
\end{table}

We use probit regression and logistic regression to model the factors that affect personal
loans, linear combination is $Z_i = \beta_0+\beta_1 \log AnuI_i + \beta_2
Emp_i + \beta_3 DTI_i + \beta_4 \log Loanam_i + \beta_5 Term_i$. Candidate models are $M_k, k=1,2$, which can be listed as 
$$Y_i|Z_i \overset{i.i.d.}{\sim} Bernoulli\left(\mu%
\left(Z_i\right)\right),$$
where $\mu\left(Z_i\right)$ differs in 
$M_1: \mu\left(Z_i\right) =\Phi\left(Z_i\right)$, and $M_2: \mu\left(Z_i\right) =logit\left(Z_i\right)$,
where $\Phi\left(\cdot\right)$ is the cumulative density function (CDF) of standard normal distribution, and $logit\left(\cdot\right)$ is the logit link function.  For choosing the best model, we use mean-field VB to obtain the variational posterior mean estimator and compute the $\text{VPIC}^{k}$, $\text{VDIC}_M^{k}$, $\text{ELBO}^k$. Benchmark criteria, including $\text{AIC}^k$ and $\text{BIC}^k$ are also calculated 
for all candidate models. The estimator of two models are reported in table \ref{cred_risk_est}
%, the information criteria are listed in table \ref{Ic for credit risk model}.
\begin{table}[ht!]
    \caption{Variational posterior mean and standard error of $\boldsymbol{\beta}$ in M1 and M2}
    \label{cred_risk_est}
    \centering
   \begin{tabular}{cccccccc}
    \hline & & $\beta_0$ & $\beta_1$ & $\beta_2$ & $\beta_3$ & $\beta_4$ & $\beta_5$ \\
    \hline \multirow{2}{*}{M1} & $\mu_{VB}$ &  1.33 & 0.18 & 0.01 & -0.14 & -0.01 & -0.06 \\
     & $\sigma^2_{VB}$ & 1.67E-02 & 1.62E-03 & 2.06E-04 & 1.29E-03 & 8.52E-05 & 1.74E-04 \\
    \multirow{2}{*}{M2} & $\mu_{VB}$ & -1.18 & 0.49 & 0.01 & -0.06 & -0.02 & -0.11 \\
    & $\sigma^2_{VB}$ & 7.69E-03 & 1.45E-03 & 1.92E-04 & 1.79E-03 & 2.24E-04 & 2.58E-04 \\
    \hline
    \end{tabular}
\end{table}

\begin{comment}
Table \ref{Ic for credit risk model} reports $\text{VPIC}^k$ and $\text{VDIC}_M^k$
%, $VAIC_k$,
with corresponding penalty terms $P^k_{VPIC}$, $P^k_{VDIC_M}$, %P_k^{VAIC}$, 
and logrithm likelihood $\mathcal{L}(\boldsymbol{\theta},k)$ of candidate models 
$\left\{M_k\right\}_{k=1}^4$. The difference of criteria is mainly caused by
the penalty term. First and most, both $\text{VPIC}^k$ and $\text{VDIC}_M^k$ choose the model $M_4$,
which means under the same criteria, we choose model $M_4$ as the better
model than other candidate models. What's more, among VPIC and $\text{VDIC}_M$, the $%
\text{VPIC}^k$ of $M_4$ is the minimum. Thus we suggest choose the model $M_4$ with the VB posterior predictive distribution to make decision to reach the minimum predictive risk.
\end{comment}

Table \ref{Ic for credit risk model} presents the values of \(\text{VPIC}^k\) and \(\text{VDIC}_M^k\), along with $\text{AIC}^k$ and $\text{BIC}^k$ for models \(\{M_k\}_{k=1}^2\). The primary differences between the criteria is mainly due to the logarithm likelihood function (or fit term), which is no surprising as the prior of $\boldsymbol{\beta}$ is vauge. Importantly, both \(\text{VPIC}^k\) and \(\text{VDIC}_M^k\) identify model \(M_1\) as the best among the candidate models, indicating its superiority under these criteria. These VB based criteria suggest that the probit model is better than the logit model. Based on these findings, we recommend selecting model \(M_1\) and employing the VB posterior predictive distribution for decision-making to minimize predictive risk.

\begin{table}[ht!]
\small
\centering
\caption{Model selection results for the probit model and the logit model}
\label{Ic for credit risk model}
\begin{tabular}{cccccc}
\hline & $\text{VPIC}$ & $\text{VDIC}_M$ & $\text{ELBO}$ & $\text{AIC}$ & $\text{BIC}$ \\
\hline $\text{M}_1$ & \textcolor{red}{1572336.0802} & \textcolor{red}{1570922.7133} & -785536.9832 & 1570922.7620 & 1570996.9902 \\
$\text{M}_2$ & 1583798.1590 & 1582025.9001 & -825708.4686 & 1582024.7883 & 1582099.0165 \\
\hline
\end{tabular}
\end{table}
\section{Conclusion} \label{sec:conclude}
\begin{comment}
In the literature on machine learning research, VB technique has appealed a lot of attention for statistical inference of misspecified model with massive data. In this paper, firstly, for two possible predictive distributions based on variational posterior distribution, we investigate the risk function of these two predictive distributions which are
the expectation of the K-L divergence between the true density and the predictive distributions. Under some regularity conditions, we derive the corresponding asymptotically unbiased estimators for these two risk functions. After that, under the statistical decision theory, on the basis of risk functions, we propose two new information criterion for comparing misspecified models with massive data. 
%\textbf{In addition,  we established the relationship between the propose information criterion and the existing information criterion such as the popular AIC,TIC,DIC.} 
At last, we illustrate the proposed new information
criterions using some real studies in economics and finance under massive data.
\end{comment}

%In the machine learning literature, VB techniques have garnered significant attention for statistical inference in misspecified models with massive datasets. 
In this paper, we propose two novel penalty-based predictive information criteria for model comparison in the context of misspecified models with massive data. First, leveraging the VB 
%variational 
posterior distribution, we demonstrate that two types of predictive distributions can be derived from a predictive perspective: the variational plug-in predictive distribution and the variational posterior predictive distribution. Second, we investigate the risk functions associated with these two variational predictive distributions, which are defined as the expectations of the KL divergence between the DGP and the predictive distributions. Third, under specific regularity conditions, we prove that the proposed information criteria are asymptotically unbiased estimators of their respective risk functions. Finally, through comprehensive numerical simulations and empirical applications in the fields of economics and finance, we demonstrate the performance of the proposed information criteria for model comparison of misspecified models in the context of massive data.

%comprehensive numerical simulations and empirical applications in the fields of economics and finance, we demonstrate the efficacy of the proposed information criteria for model comparison of misspecified models in the context of massive data.

%Finally, through simulation studies and real-world applications, we illustrate the utility and effectiveness of the proposed information criteria.

\bibliographystyle{apalike}
\bibliography{ref2024}

\pagebreak
\appendix
%\newpage
\setcounter{page}{1}{}\setlength{\baselineskip}{12pt} %
\setcounter{equation}{0}

\begin{center}
	{\Large Online Supplement for }
	
	{\Large \textquotedblleft Comparing Misspecified Models with Big Data: A Variational Bayesian Perspective\textquotedblright\ }
	
	\quad
	
	\ Yong Li$^{a}$, Sushanta K. Mallick$^{b}$, Tao Zeng$^{c}$, Junxing Zhang$^{a}$\\
	\quad
	
		$^{a}$ School of Economics, Renmin University of China, China \\
        $^{b}$ School of Business and Management, Queen Mary University of London
        \\
	$^{c}$ School of Economics, Zhejiang University, China

	{\small \ \ \ \ \ \ \ \ \ }
\end{center}
  This Online Supplement consists of two sections. Section A contains the proofs of Theorem \ref{riskvtic} and Theorem \ref{riskvpic}, with related lemmas used in these proofs. Section B contains VB analytical expression of parametric models used in the paper. 

\section{Proofs for Theorems and related lemmas}	

\subsection{Notations}

\begin{tabular}{llll}
$:=$ & definitional equality & $\overleftrightarrow{{\mbox{\boldmath${%
\theta}$}}}_{n}$ & posterior mode \\ 
$o(1)$ & tend to zero & $\widehat{{\mbox{\boldmath${\theta}$}}}_{n}$ & QML
estimate \\ 
$o_{p}(1)$ & tend to zero in probability & $\mbox{\boldmath${\theta}$}%
_{n}^{p}$ & pseudo true parameter \\ 
$\overset{p}{\rightarrow }$ & converge in probability & ${%
\mbox{\boldmath${\hat{\theta}}$}}_{AT}$ & $\arg \max $ of $2\ln p(\mathbf{y}|%
{\mbox{\boldmath${\theta}$}})+\ln p({\mbox{\boldmath${\theta}$}})$ \\ 
$\overline{\mbox{\boldmath${\theta}$}}_{n}$ & posterior mean & $\widetilde{%
\mbox{\boldmath${\theta}$}}_{n}$ & $\arg \max $ of $\ln p\left( \mathbf{y}%
_{rep}|\mbox{\boldmath${\theta}$}\right) +\ln p\left( \mathbf{y}|%
\mbox{\boldmath${\theta}$}\right) +\ln p\left( \mbox{\boldmath${\theta}$}%
\right) $%
\end{tabular}

\subsection{Proof of Theorems in the main paper}

Denote%
\begin{equation*}
\widetilde{\mbox{\boldmath${\theta}$}}_{n}^{s}:=\arg \max_{%
\mbox{\boldmath${\theta}$}}\ln p\left( \mathbf{y}_{rep}|\mbox{\boldmath${%
\theta}$}\right) -\frac{n}{2}\left( \widehat{{\mbox{\boldmath${\theta}$}}}%
_{n}\left( \mathbf{y}\right) -\mbox{\boldmath${\theta}$}\right) ^{\prime
}\left( -\mathbf{H}_{n}^{d}\right) \left( \widehat{{\mbox{\boldmath${%
\theta}$}}}_{n}\left( \mathbf{y}\right) -\mbox{\boldmath${\theta}$}\right)
\end{equation*}%
where $\mathbf{H}_{n}^{d}$ is diagonal and has the same diagonal terms as $%
\mathbf{H}_{n}$. Then we have the following three lemmas under the
condition that $\mathbf{y}$ and $\mathbf{y}_{rep}$ are independent. These
three lemma are useful to prove Theorem 3.2. %\ref{riskvpic}

\begin{lemma}
\label{lemmaconsistency}Under Assumptions 1-8, $\widetilde{%
\mbox{\boldmath${\theta}$}}_{n}\overset{p}{\rightarrow }\mbox{\boldmath${%
\theta}$}_{n}^{p}{.}$
\end{lemma}

\begin{proof}
The proof follows the argument in Theorem 4.2 in Wooldridge (1994) and
Bester and Hansen (2006). Let $Q_{n}\left( \mbox{\boldmath${\theta}$}\right)
=n^{-1}\sum_{t=1}^{n}\left[ l_{t}\left( \mathbf{y}_{rep}^{t},%
\mbox{\boldmath${\theta}$}\right) -\frac{1}{2}\left( \widehat{{%
\mbox{\boldmath${\theta}$}}}_{n}\left( \mathbf{y}\right) -%
\mbox{\boldmath${\theta}$}\right) ^{\prime }\left( -\mathbf{H}%
_{n}^{d}\right) \left( \widehat{{\mbox{\boldmath${\theta}$}}}_{n}\left( 
\mathbf{y}\right) -\mbox{\boldmath${\theta}$}\right) \right] $ and\textbf{\ }%
$\bar{Q}_{n}\left( \mbox{\boldmath${\theta}$}\right) =n^{-1}E\left[ \ln
p\left( \mathbf{y}_{rep}|\mbox{\boldmath${\theta}$}\right) -\frac{n}{2}%
\left( \widehat{{\mbox{\boldmath${\theta}$}}}_{n}\left( \mathbf{y}\right) -%
\mbox{\boldmath${\theta}$}\right) ^{\prime }\left( -\mathbf{H}%
_{n}^{d}\right) \left( \widehat{{\mbox{\boldmath${\theta}$}}}_{n}\left( 
\mathbf{y}\right) -\mbox{\boldmath${\theta}$}\right) \right] $. \ For
simplicity, let%
\begin{equation*}
l_{t}^{\prime }\left( \mathbf{y},\mbox{\boldmath${\theta}$}\right) =-\frac{1%
}{2}\left( \widehat{{\mbox{\boldmath${\theta}$}}}_{n}\left( \mathbf{y}%
\right) -\mbox{\boldmath${\theta}$}\right) ^{\prime }\left( -\mathbf{H}%
_{n}^{d}\right) \left( \widehat{{\mbox{\boldmath${\theta}$}}}_{n}\left( 
\mathbf{y}\right) -\mbox{\boldmath${\theta}$}\right) ,
\end{equation*}%
then%
\begin{equation*}
Q_{n}\left( \mbox{\boldmath${\theta}$}\right) =n^{-1}\sum_{t=1}^{n}\left[
l_{t}\left( \mathbf{y}_{rep}^{t},\mbox{\boldmath${\theta}$}\right)
+l_{t}^{\prime }\left( \mathbf{y},\mbox{\boldmath${\theta}$}\right) \right] ,
\end{equation*}%
\begin{equation*}
\bar{Q}_{n}\left( \mbox{\boldmath${\theta}$}\right) =n^{-1}E\left[
\sum_{t=1}^{n}\left[ l_{t}\left( \mathbf{y}_{rep}^{t},\mbox{\boldmath${%
\theta}$}\right) +l_{t}^{\prime }\left( \mathbf{y},\mbox{\boldmath${\theta}$}%
\right) \right] \right] .
\end{equation*}%
Then we need to show that, for each $\varepsilon >0$, 
\begin{equation*}
P\left[ \sup_{\mbox{\boldmath${\theta}$}\mathbf{\in {\mbox{\boldmath${%
\Theta}$}}}}\left\vert Q_{n}\left( \mbox{\boldmath${\theta}$}\right) -\bar{Q}%
_{n}\left( \mbox{\boldmath${\theta}$}\right) \right\vert >\varepsilon \right]
\rightarrow 0\text{.}
\end{equation*}%
Let $\delta >0$ be a number to be set later. Because ${\mbox{\boldmath${%
\Theta}$}}$\ is compact, there exists a finite number of spheres of radius $%
\delta $\ about $\mbox{\boldmath${\theta}$}_{j}$, say $\zeta _{\delta
}\left( \mbox{\boldmath${\theta}$}_{j}\right) =\left\{ \mbox{\boldmath${%
\theta}$}\mathbf{\in {\mbox{\boldmath${\Theta}$}}:}\left\Vert %
\mbox{\boldmath${\theta}$}\mathbf{-}\mbox{\boldmath${\theta}$}%
_{j}\right\Vert \leq \delta \right\} $, $j=1,\ldots ,K\left( \delta \right) $%
, which covers ${\mbox{\boldmath${\Theta}$}}$ (Gallant and White, 1988). Set%
\textbf{\ }$\zeta _{j}=\zeta _{\delta }\left( \mbox{\boldmath${\theta}$}%
_{j}\right) $\textbf{, }$K=K\left( \delta \right) $\textbf{.} It follows that%
\begin{eqnarray*}
P\left[ \sup_{\mbox{\boldmath${\theta}$}\mathbf{\in {\mbox{\boldmath${%
\Theta}$}}}}\left\vert Q_{n}\left( \mbox{\boldmath${\theta}$}\right) -\bar{Q}%
_{n}\left( \mbox{\boldmath${\theta}$}\right) \right\vert >\varepsilon \right]
&\leq &P\left[ \max_{1\leq j\leq K}\sup_{\mbox{\boldmath${\theta}$}\mathbf{%
\in }\zeta _{j}}\left\vert Q_{n}\left( \mbox{\boldmath${\theta}$}\right) -%
\bar{Q}_{n}\left( \mbox{\boldmath${\theta}$}\right) \right\vert >\varepsilon %
\right] \\
&\leq &\sum_{j=1}^{K}P\left[ \sup_{\mbox{\boldmath${\theta}$}\mathbf{\in }%
\zeta _{j}}\left\vert Q_{n}\left( \mbox{\boldmath${\theta}$}\right) -\bar{Q}%
_{n}\left( \mbox{\boldmath${\theta}$}\right) \right\vert >\varepsilon \right]
.
\end{eqnarray*}%
For all $\mbox{\boldmath${\theta}$}\in \zeta _{j},$ 
\begin{eqnarray*}
&&\left\vert Q_{n}\left( \mbox{\boldmath${\theta}$}\right) -\bar{Q}%
_{n}\left( \mbox{\boldmath${\theta}$}\right) \right\vert \\
&\leq &\left\vert Q_{n}\left( \mbox{\boldmath${\theta}$}\right) -Q_{n}\left( %
\mbox{\boldmath${\theta}$}_{j}\right) \right\vert +\left\vert Q_{n}\left( %
\mbox{\boldmath${\theta}$}_{j}\right) -\bar{Q}_{n}\left( \mbox{\boldmath${%
\theta}$}_{j}\right) \right\vert +\left\vert \bar{Q}_{n}\left( %
\mbox{\boldmath${\theta}$}_{j}\right) -\bar{Q}_{n}\left( \mbox{\boldmath${%
\theta}$}\right) \right\vert \\
&\leq &\frac{1}{n}\sum_{t=1}^{n}\left\vert l_{t}^{\prime }\left( \mathbf{y},%
\mbox{\boldmath${\theta}$}\right) -l_{t}^{\prime }\left( \mathbf{y}^{t},%
\mbox{\boldmath${\theta}$}_{j}\right) \right\vert +\frac{1}{n}%
\sum_{t=1}^{n}\left\vert l_{t}\left( \mathbf{y}_{rep}^{t},%
\mbox{\boldmath${\theta}$}\right) -l_{t}\left( \mathbf{y}_{rep}^{t},%
\mbox{\boldmath${\theta}$}_{j}\right) \right\vert \\
&&+\left\vert \frac{1}{n}\sum_{t=1}^{n}\left( l_{t}^{\prime }\left( \mathbf{y%
},\mbox{\boldmath${\theta}$}_{j}\right) -E\left[ l_{t}^{\prime }\left( 
\mathbf{y},\mbox{\boldmath${\theta}$}_{j}\right) \right] \right) \right\vert
+\left\vert \frac{1}{n}\sum_{t=1}^{n}\left( l_{t}\left( \mathbf{y}_{rep}^{t},%
\mbox{\boldmath${\theta}$}_{j}\right) -E\left[ l_{t}\left( \mathbf{y}%
_{rep}^{t},\mbox{\boldmath${\theta}$}_{j}\right) \right] \right) \right\vert
\\
&&+\frac{1}{n}\sum_{t=1}^{n}\left\vert E\left[ l_{t}^{\prime }\left( \mathbf{%
y},\mbox{\boldmath${\theta}$}\right) \right] -E\left[ l_{t}^{\prime }\left( 
\mathbf{y},\mbox{\boldmath${\theta}$}_{j}\right) \right] \right\vert +\frac{1%
}{n}\sum_{t=1}^{n}\left\vert E\left[ l_{t}\left( \mbox{\boldmath${\theta}$}%
\right) \right] -E\left[ l_{t}\left( \mbox{\boldmath${\theta}$}_{j}\right) %
\right] \right\vert ,
\end{eqnarray*}%
where $E\left[ l_{t}\left( \mbox{\boldmath${\theta}$}\right) \right] :=E%
\left[ l_{t}\left( \mathbf{y}^{t},\mbox{\boldmath${\theta}$}\right) \right]
=E\left[ l_{t}\left( \mathbf{y}_{rep}^{t},\mbox{\boldmath${\theta}$}\right) %
\right] $. By Assumption 4, for all $\mbox{\boldmath${\theta}$}\in \zeta
_{j},$%
\begin{equation*}
\left\vert l_{t}\left( \mathbf{y}_{rep}^{t},\mbox{\boldmath${\theta}$}%
\right) -l_{t}\left( \mathbf{y}_{rep}^{t},\mbox{\boldmath${\theta}$}%
_{j}\right) \right\vert \leq c_{t}\left( \mathbf{y}_{rep}^{t}\right)
\left\Vert \mbox{\boldmath${\theta}$}\mathbf{-}\mbox{\boldmath${\theta}$}%
_{j}\right\Vert \leq \delta c_{t}\left( \mathbf{y}_{rep}^{t}\right) .
\end{equation*}%
and%
\begin{equation*}
\left\vert E\left[ l_{t}\left( \mathbf{y}_{rep}^{t},\mbox{\boldmath${%
\theta}$}\right) \right] -E\left[ l_{t}\left( \mathbf{y}_{rep}^{t},%
\mbox{\boldmath${\theta}$}_{j}\right) \right] \right\vert \leq \delta \bar{c}%
_{t},
\end{equation*}%
where $\bar{c}_{t}=E\left[ c_{t}\left( \mathbf{y}_{rep}^{t}\right) \right] $%
. Note that%
\begin{eqnarray*}
&&l_{t}^{\prime }\left( \mathbf{y},\mbox{\boldmath${\theta}$}\right) \\
&=&-\frac{1}{2}\left( \widehat{{\mbox{\boldmath${\theta}$}}}_{n}\left( 
\mathbf{y}\right) -\mbox{\boldmath${\theta}$}\right) ^{\prime }\left( -%
\mathbf{H}_{n}^{d}\right) \left( \widehat{{\mbox{\boldmath${\theta}$}}}%
_{n}\left( \mathbf{y}\right) -\mbox{\boldmath${\theta}$}\right) \\
&=&-\frac{1}{2}\left( \widehat{{\mbox{\boldmath${\theta}$}}}_{n}\left( 
\mathbf{y}\right) -\mbox{\boldmath${\theta}$}_{j}+\mbox{\boldmath${\theta}$}%
_{j}-\mbox{\boldmath${\theta}$}\right) ^{\prime }\left( -\mathbf{H}%
_{n}^{d}\right) \left( \widehat{{\mbox{\boldmath${\theta}$}}}_{n}\left( 
\mathbf{y}\right) -\mbox{\boldmath${\theta}$}_{j}+\mbox{\boldmath${\theta}$}%
_{j}-\mbox{\boldmath${\theta}$}\right) \\
&=&-\frac{1}{2}\left( \widehat{{\mbox{\boldmath${\theta}$}}}_{n}\left( 
\mathbf{y}\right) -\mbox{\boldmath${\theta}$}_{j}\right) ^{\prime }\left( -%
\mathbf{H}_{n}^{d}\right) \left( \widehat{{\mbox{\boldmath${\theta}$}}}%
_{n}\left( \mathbf{y}\right) -\mbox{\boldmath${\theta}$}_{j}\right) -\frac{1%
}{2}\left( \widehat{{\mbox{\boldmath${\theta}$}}}_{n}\left( \mathbf{y}%
\right) -\mbox{\boldmath${\theta}$}_{j}\right) ^{\prime }\left( -\mathbf{H}%
_{n}^{d}\right) \left( \mbox{\boldmath${\theta}$}_{j}-\mbox{\boldmath${%
\theta}$}\right) \\
&&-\frac{1}{2}\left( \mbox{\boldmath${\theta}$}_{j}-\mbox{\boldmath${%
\theta}$}\right) ^{\prime }\left( -\mathbf{H}_{n}^{d}\right) \left( 
\widehat{{\mbox{\boldmath${\theta}$}}}_{n}\left( \mathbf{y}\right) -%
\mbox{\boldmath${\theta}$}_{j}\right) -\frac{1}{2}\left( \mbox{\boldmath${%
\theta}$}_{j}-\mbox{\boldmath${\theta}$}\right) ^{\prime }\left( -\mathbf{H}%
_{n}^{d}\right) \left( \mbox{\boldmath${\theta}$}_{j}-\mbox{\boldmath${%
\theta}$}\right) ,
\end{eqnarray*}%
then we have%
\begin{eqnarray*}
&&\frac{1}{n}\sum_{t=1}^{n}\left\vert l_{t}^{\prime }\left( \mathbf{y},%
\mbox{\boldmath${\theta}$}\right) -l_{t}^{\prime }\left( \mathbf{y}^{t},%
\mbox{\boldmath${\theta}$}_{j}\right) \right\vert \\
&=&\left\vert l_{t}^{\prime }\left( \mathbf{y},\mbox{\boldmath${\theta}$}%
\right) -l_{t}^{\prime }\left( \mathbf{y},\mbox{\boldmath${\theta}$}%
_{j}\right) \right\vert \\
&\leq &\left\vert \left( \widehat{{\mbox{\boldmath${\theta}$}}}_{n}\left( 
\mathbf{y}\right) -\mbox{\boldmath${\theta}$}_{j}\right) ^{\prime }\left( -%
\mathbf{H}_{n}^{d}\right) \left( \mbox{\boldmath${\theta}$}_{j}-%
\mbox{\boldmath${\theta}$}\right) -\frac{1}{2}\left( \mbox{\boldmath${%
\theta}$}_{j}-\mbox{\boldmath${\theta}$}\right) ^{\prime }\left( -\mathbf{H}%
_{n}^{d}\right) \left( \mbox{\boldmath${\theta}$}_{j}-\mbox{\boldmath${%
\theta}$}\right) \right\vert \\
&\leq &\left\vert \left( \widehat{{\mbox{\boldmath${\theta}$}}}_{n}\left( 
\mathbf{y}\right) -\mbox{\boldmath${\theta}$}_{j}\right) ^{\prime }\left( -%
\mathbf{H}_{n}^{d}\right) \left( \mbox{\boldmath${\theta}$}_{j}-%
\mbox{\boldmath${\theta}$}\right) \right\vert +\frac{1}{2}\left\vert \left( %
\mbox{\boldmath${\theta}$}_{j}-\mbox{\boldmath${\theta}$}\right) ^{\prime
}\left( -\mathbf{H}_{n}^{d}\right) \left( \mbox{\boldmath${\theta}$}_{j}-%
\mbox{\boldmath${\theta}$}\right) \right\vert \\
&\leq &\left\vert \left\vert \widehat{{\mbox{\boldmath${\theta}$}}}%
_{n}\left( \mathbf{y}\right) -\mbox{\boldmath${\theta}$}_{j}\right\vert
\right\vert \left\vert \left\vert -\mathbf{H}_{n}^{d}\right\vert
\right\vert \left\vert \left\vert \mbox{\boldmath${\theta}$}_{j}-%
\mbox{\boldmath${\theta}$}\right\vert \right\vert +\frac{1}{2}\left\vert
\left\vert -\mathbf{H}_{n}^{d}\right\vert \right\vert \left\vert
\left\vert \mbox{\boldmath${\theta}$}_{j}-\mbox{\boldmath${\theta}$}%
\right\vert \right\vert ^{2} \\
&\leq &\left\vert \left\vert \widehat{{\mbox{\boldmath${\theta}$}}}%
_{n}\left( \mathbf{y}\right) -\mbox{\boldmath${\theta}$}_{j}\right\vert
\right\vert \left\vert \left\vert -\mathbf{H}_{n}^{d}\right\vert
\right\vert \delta +\frac{1}{2}\left\vert \left\vert -\mathbf{H}%
_{n}^{d}\right\vert \right\vert \delta ^{2} \\
&\leq &\left\vert \left\vert \widehat{{\mbox{\boldmath${\theta}$}}}%
_{n}\left( \mathbf{y}\right) -\mbox{\boldmath${\theta}$}_{n}^{p}\right\vert
\right\vert \left\vert \left\vert -\mathbf{H}_{n}^{d}\right\vert
\right\vert \delta +\left\vert \left\vert \mbox{\boldmath${\theta}$}_{n}^{p}-%
\mbox{\boldmath${\theta}$}_{j}\right\vert \right\vert \left\vert \left\vert -%
\mathbf{H}_{n}^{d}\right\vert \right\vert \delta +\frac{1}{2}\left\vert
\left\vert -\mathbf{H}_{n}^{d}\right\vert \right\vert \delta ^{2}
\end{eqnarray*}%
and%
\begin{eqnarray*}
&&\frac{1}{n}\sum_{t=1}^{n}\left\vert E\left[ l_{t}^{\prime }\left( \mathbf{y%
},\mbox{\boldmath${\theta}$}\right) \right] -E\left[ l_{t}^{\prime }\left( 
\mathbf{y},\mbox{\boldmath${\theta}$}_{j}\right) \right] \right\vert \\
&\leq &E\left\vert l_{t}^{\prime }\left( \mathbf{y},\mbox{\boldmath${%
\theta}$}\right) -l_{t}^{\prime }\left( \mathbf{y},\mbox{\boldmath${\theta}$}%
_{j}\right) \right\vert \\
&\leq &E\left( \left\vert \left\vert \widehat{{\mbox{\boldmath${\theta}$}}}%
_{n}\left( \mathbf{y}\right) -\mbox{\boldmath${\theta}$}_{j}\right\vert
\right\vert \right) \left\vert \left\vert -\mathbf{H}_{n}^{d}\right\vert
\right\vert \delta +\frac{1}{2}\left\vert \left\vert -\mathbf{H}%
_{n}^{d}\right\vert \right\vert \delta ^{2}.
\end{eqnarray*}%
It can be shown that%
\begin{eqnarray*}
&&-2l_{t}^{\prime }\left( \mathbf{y},\mbox{\boldmath${\theta}$}_{j}\right) \\
&=&\left( \widehat{{\mbox{\boldmath${\theta}$}}}_{n}\left( \mathbf{y}\right)
-\mbox{\boldmath${\theta}$}_{j}\right) ^{\prime }\left( -\mathbf{H}%
_{n}^{d}\right) \left( \widehat{{\mbox{\boldmath${\theta}$}}}_{n}\left( 
\mathbf{y}\right) -\mbox{\boldmath${\theta}$}_{j}\right) \\
&=&\left( \widehat{{\mbox{\boldmath${\theta}$}}}_{n}\left( \mathbf{y}\right)
-\mbox{\boldmath${\theta}$}_{n}^{p}+\mbox{\boldmath${\theta}$}_{n}^{p}-%
\mbox{\boldmath${\theta}$}_{j}\right) ^{\prime }\left( -\mathbf{H}%
_{n}^{d}\right) \left( \widehat{{\mbox{\boldmath${\theta}$}}}_{n}\left( 
\mathbf{y}\right) -\mbox{\boldmath${\theta}$}_{n}^{p}+\mbox{\boldmath${%
\theta}$}_{n}^{p}-\mbox{\boldmath${\theta}$}_{j}\right) \\
&=&\left( \widehat{{\mbox{\boldmath${\theta}$}}}_{n}\left( \mathbf{y}\right)
-\mbox{\boldmath${\theta}$}_{n}^{p}\right) ^{\prime }\left( -\mathbf{H}%
_{n}^{d}\right) \left( \widehat{{\mbox{\boldmath${\theta}$}}}_{n}\left( 
\mathbf{y}\right) -\mbox{\boldmath${\theta}$}_{n}^{p}\right) +\left( 
\widehat{{\mbox{\boldmath${\theta}$}}}_{n}\left( \mathbf{y}\right) -%
\mbox{\boldmath${\theta}$}_{n}^{p}\right) ^{\prime }\left( -\mathbf{H}%
_{n}^{d}\right) \left( \mbox{\boldmath${\theta}$}_{n}^{p}-%
\mbox{\boldmath${\theta}$}_{j}\right) \\
&&+\left( \mbox{\boldmath${\theta}$}_{n}^{p}-\mbox{\boldmath${\theta}$}%
_{j}\right) ^{\prime }\left( -\mathbf{H}_{n}^{d}\right) \left( \widehat{{%
\mbox{\boldmath${\theta}$}}}_{n}\left( \mathbf{y}\right) -%
\mbox{\boldmath${\theta}$}_{n}^{p}\right) +\left( \mbox{\boldmath${\theta}$}%
_{n}^{p}-\mbox{\boldmath${\theta}$}_{j}\right) ^{\prime }\left( -\mathbf{H}%
_{n}^{d}\right) \left( \mbox{\boldmath${\theta}$}_{n}^{p}-%
\mbox{\boldmath${\theta}$}_{j}\right) ,
\end{eqnarray*}%
then%
\begin{eqnarray*}
&&\left\vert \frac{1}{n}\sum_{t=1}^{n}\left( l_{t}^{\prime }\left( \mathbf{y}%
,\mbox{\boldmath${\theta}$}_{j}\right) -E\left[ l_{t}^{\prime }\left( 
\mathbf{y},\mbox{\boldmath${\theta}$}_{j}\right) \right] \right) \right\vert
\\
&=&\frac{1}{2}\left\vert \left( \widehat{{\mbox{\boldmath${\theta}$}}}%
_{n}\left( \mathbf{y}\right) -\mbox{\boldmath${\theta}$}_{j}\right) ^{\prime
}\left( -\mathbf{H}_{n}^{d}\right) \left( \widehat{{\mbox{\boldmath${%
\theta}$}}}_{n}\left( \mathbf{y}\right) -\mbox{\boldmath${\theta}$}%
_{j}\right) -E\left[ \left( \widehat{{\mbox{\boldmath${\theta}$}}}_{n}\left( 
\mathbf{y}\right) -\mbox{\boldmath${\theta}$}_{j}\right) ^{\prime }\left( -%
\mathbf{H}_{n}^{d}\right) \left( \widehat{{\mbox{\boldmath${\theta}$}}}%
_{n}\left( \mathbf{y}\right) -\mbox{\boldmath${\theta}$}_{j}\right) \right]
\right\vert \\
&\leq &\frac{1}{2}\left\vert 
\begin{array}{c}
\left( \widehat{{\mbox{\boldmath${\theta}$}}}_{n}\left( \mathbf{y}\right) -%
\mbox{\boldmath${\theta}$}_{n}^{p}\right) ^{\prime }\left( -\mathbf{H}%
_{n}^{d}\right) \left( \widehat{{\mbox{\boldmath${\theta}$}}}_{n}\left( 
\mathbf{y}\right) -\mbox{\boldmath${\theta}$}_{n}^{p}\right) \\ 
-E\left[ \left( \widehat{{\mbox{\boldmath${\theta}$}}}_{n}\left( \mathbf{y}%
\right) -\mbox{\boldmath${\theta}$}_{n}^{p}\right) ^{\prime }\left( -\mathbf{%
H}_{n}^{d}\right) \left( \widehat{{\mbox{\boldmath${\theta}$}}}_{n}\left( 
\mathbf{y}\right) -\mbox{\boldmath${\theta}$}_{n}^{p}\right) \right]%
\end{array}%
\right\vert \\
&&+\left\vert \left( \left( \widehat{{\mbox{\boldmath${\theta}$}}}_{n}\left( 
\mathbf{y}\right) -\mbox{\boldmath${\theta}$}_{n}^{p}\right) ^{\prime }-E%
\left[ \left( \widehat{{\mbox{\boldmath${\theta}$}}}_{n}\left( \mathbf{y}%
\right) -\mbox{\boldmath${\theta}$}_{n}^{p}\right) ^{\prime }\right] \right)
\left( -\mathbf{H}_{n}^{d}\right) \left( \mbox{\boldmath${\theta}$}%
_{n}^{p}-\mbox{\boldmath${\theta}$}_{j}\right) \right\vert \\
&\leq &\frac{1}{2}\left\vert 
\begin{array}{c}
\left( \widehat{{\mbox{\boldmath${\theta}$}}}_{n}\left( \mathbf{y}\right) -%
\mbox{\boldmath${\theta}$}_{n}^{p}\right) ^{\prime }\left( -\mathbf{H}%
_{n}^{d}\right) \left( \widehat{{\mbox{\boldmath${\theta}$}}}_{n}\left( 
\mathbf{y}\right) -\mbox{\boldmath${\theta}$}_{n}^{p}\right) \\ 
-E\left[ \left( \widehat{{\mbox{\boldmath${\theta}$}}}_{n}\left( \mathbf{y}%
\right) -\mbox{\boldmath${\theta}$}_{n}^{p}\right) ^{\prime }\left( -\mathbf{%
H}_{n}^{d}\right) \left( \widehat{{\mbox{\boldmath${\theta}$}}}_{n}\left( 
\mathbf{y}\right) -\mbox{\boldmath${\theta}$}_{n}^{p}\right) \right]%
\end{array}%
\right\vert \\
&&+\left\vert \left\vert \left( \widehat{{\mbox{\boldmath${\theta}$}}}%
_{n}\left( \mathbf{y}\right) -\mbox{\boldmath${\theta}$}_{n}^{p}\right)
^{\prime }-E\left[ \left( \widehat{{\mbox{\boldmath${\theta}$}}}_{n}\left( 
\mathbf{y}\right) -\mbox{\boldmath${\theta}$}_{n}^{p}\right) ^{\prime }%
\right] \right\vert \right\vert \left\vert \left\vert \left( -\mathbf{H}%
_{n}^{d}\right) \left( \mbox{\boldmath${\theta}$}_{n}^{p}-%
\mbox{\boldmath${\theta}$}_{j}\right) \right\vert \right\vert .
\end{eqnarray*}%
Let%
\begin{equation*}
l_{t}^{\prime }\left( \mathbf{y},\mbox{\boldmath${\theta}$}_{n}^{p}\right)
=\left( \widehat{{\mbox{\boldmath${\theta}$}}}_{n}\left( \mathbf{y}\right) -%
\mbox{\boldmath${\theta}$}_{n}^{p}\right) ^{\prime }\left( -\mathbf{H}%
_{n}^{d}\right) \left( \widehat{{\mbox{\boldmath${\theta}$}}}_{n}\left( 
\mathbf{y}\right) -\mbox{\boldmath${\theta}$}_{n}^{p}\right) ,
\end{equation*}%
we have%
\begin{equation*}
l_{t}^{\prime }\left( \mathbf{y},\mbox{\boldmath${\theta}$}_{n}^{p}\right)
-E\left( l_{t}^{\prime }\left( \mathbf{y},\mbox{\boldmath${\theta}$}%
_{n}^{p}\right) \right) =o_{p}\left( 1\right) ,
\end{equation*}%
\begin{equation*}
\left( \widehat{{\mbox{\boldmath${\theta}$}}}_{n}\left( \mathbf{y}\right) -%
\mbox{\boldmath${\theta}$}_{n}^{p}\right) ^{\prime }-E\left[ \left( \widehat{%
{\mbox{\boldmath${\theta}$}}}_{n}\left( \mathbf{y}\right) -%
\mbox{\boldmath${\theta}$}_{n}^{p}\right) ^{\prime }\right] =o_{p}\left(
1\right)
\end{equation*}%
by Assumptions 1-8.

Thus, we have%
\begin{eqnarray*}
&&\sup_{\mbox{\boldmath${\theta}$}\mathbf{\in }\zeta _{j}}\left\vert
Q_{n}\left( \mbox{\boldmath${\theta}$}\right) -\bar{Q}_{n}\left( %
\mbox{\boldmath${\theta}$}\right) \right\vert \\
&\leq &\left\vert \left\vert -\mathbf{H}_{n}^{d}\right\vert \right\vert
\delta ^{2}+E\left( \left\vert \left\vert \widehat{{\mbox{\boldmath${%
\theta}$}}}_{n}\left( \mathbf{y}\right) -\mbox{\boldmath${\theta}$}%
_{j}\right\vert \right\vert \right) \left\vert \left\vert -\mathbf{H}%
_{n}^{d}\right\vert \right\vert \delta +\left\vert \left\vert %
\mbox{\boldmath${\theta}$}_{n}^{p}-\mbox{\boldmath${\theta}$}_{j}\right\vert
\right\vert \left\vert \left\vert -\mathbf{H}_{n}^{d}\right\vert
\right\vert \delta +\frac{2\delta }{n}\sum_{t=1}^{n}\bar{c}_{t} \\
&&+\left\vert \left\vert \widehat{{\mbox{\boldmath${\theta}$}}}_{n}\left( 
\mathbf{y}\right) -\mbox{\boldmath${\theta}$}_{n}^{p}\right\vert \right\vert
\left\vert \left\vert -\mathbf{H}_{n}^{d}\right\vert \right\vert \delta
+\left\vert l_{t}^{\prime }\left( \mathbf{y},\mbox{\boldmath${\theta}$}%
_{n}^{p}\right) -E\left( l_{t}^{\prime }\left( \mathbf{y},%
\mbox{\boldmath${\theta}$}_{n}^{p}\right) \right) \right\vert \\
&&+\left\vert \left\vert \left( \widehat{{\mbox{\boldmath${\theta}$}}}%
_{n}\left( \mathbf{y}\right) -\mbox{\boldmath${\theta}$}_{n}^{p}\right)
^{\prime }-E\left[ \left( \widehat{{\mbox{\boldmath${\theta}$}}}_{n}\left( 
\mathbf{y}\right) -\mbox{\boldmath${\theta}$}_{n}^{p}\right) ^{\prime }%
\right] \right\vert \right\vert \left\vert \left\vert \left( -\mathbf{H}%
_{n}^{d}\right) \left( \mbox{\boldmath${\theta}$}_{n}^{p}-%
\mbox{\boldmath${\theta}$}_{j}\right) \right\vert \right\vert \\
&&+\frac{\delta }{n}\sum_{t=1}^{n}\left[ c_{t}\left( \mathbf{y}%
_{rep}^{t}\right) -\bar{c}_{t}\right] +\left\vert \frac{1}{n}%
\sum_{t=1}^{n}\left( l_{t}\left( \mathbf{y}_{rep}^{t},\mbox{\boldmath${%
\theta}$}_{j}\right) -E\left[ l_{t}\left( \mathbf{y}_{rep}^{t},%
\mbox{\boldmath${\theta}$}_{j}\right) \right] \right) \right\vert .
\end{eqnarray*}%
By Assumptions 1-8, there exists some $C^{\ast }\left( \delta \right)
<\infty $ such that 
\begin{equation*}
C^{\ast }\left( \delta \right) \geq \left\vert \left\vert -\mathbf{H}%
_{n}^{d}\right\vert \right\vert \delta ^{2}+E\left( \left\vert \left\vert 
\widehat{{\mbox{\boldmath${\theta}$}}}_{n}\left( \mathbf{y}\right) -%
\mbox{\boldmath${\theta}$}_{j}\right\vert \right\vert \right) \left\vert
\left\vert -\mathbf{H}_{n}^{d}\right\vert \right\vert \delta +\left\vert
\left\vert \mbox{\boldmath${\theta}$}_{n}^{p}-\mbox{\boldmath${\theta}$}%
_{j}\right\vert \right\vert \left\vert \left\vert -\mathbf{H}%
_{n}^{d}\right\vert \right\vert \delta +\frac{2\delta }{n}\sum_{t=1}^{n}%
\bar{c}_{t}.
\end{equation*}%
And if we define%
\begin{eqnarray*}
Z_{n,j}^{\ast } &=&\left\vert \left\vert \widehat{{\mbox{\boldmath${\theta}$}%
}}_{n}\left( \mathbf{y}\right) -\mbox{\boldmath${\theta}$}%
_{n}^{p}\right\vert \right\vert \left\vert \left\vert -\mathbf{H}%
_{n}^{d}\right\vert \right\vert \delta +\left\vert l_{t}^{\prime }\left( 
\mathbf{y},\mbox{\boldmath${\theta}$}_{n}^{p}\right) -E\left( l_{t}^{\prime
}\left( \mathbf{y},\mbox{\boldmath${\theta}$}_{n}^{p}\right) \right)
\right\vert \\
&&+\left\vert \left\vert \left( \widehat{{\mbox{\boldmath${\theta}$}}}%
_{n}\left( \mathbf{y}\right) -\mbox{\boldmath${\theta}$}_{n}^{p}\right)
^{\prime }-E\left[ \left( \widehat{{\mbox{\boldmath${\theta}$}}}_{n}\left( 
\mathbf{y}\right) -\mbox{\boldmath${\theta}$}_{n}^{p}\right) ^{\prime }%
\right] \right\vert \right\vert \left\vert \left\vert \left( -\mathbf{H}%
_{n}^{d}\right) \left( \mbox{\boldmath${\theta}$}_{n}^{p}-%
\mbox{\boldmath${\theta}$}_{j}\right) \right\vert \right\vert \\
&&+\frac{\delta }{n}\sum_{t=1}^{n}\left[ c_{t}\left( \mathbf{y}%
_{rep}^{t}\right) -\bar{c}_{t}\right] +\left\vert \frac{1}{n}%
\sum_{t=1}^{n}\left( l_{t}\left( \mathbf{y}_{rep}^{t},\mbox{\boldmath${%
\theta}$}_{j}\right) -E\left[ l_{t}\left( \mathbf{y}_{rep}^{t},%
\mbox{\boldmath${\theta}$}_{j}\right) \right] \right) \right\vert ,
\end{eqnarray*}%
we have $Z_{n,j}^{\ast }=o_{p}\left( 1\right) $ by Assumptions 1-8.

It follows that%
\begin{equation*}
P\left[ \max_{\mbox{\boldmath${\theta}$}\mathbf{\in }\zeta _{j}}\left\vert
Q_{n}\left( \mbox{\boldmath${\theta}$}\right) -\bar{Q}_{n}\left( %
\mbox{\boldmath${\theta}$}\right) \right\vert >\varepsilon \right] \leq P%
\left[ Z_{n,j}^{\ast }>\varepsilon -C^{\ast }\left( \delta \right) \right] .
\end{equation*}%
Now choose $\delta \leq 1$ such that $\varepsilon -C^{\ast }\left( \delta
\right) <\varepsilon /2$. Then%
\begin{equation*}
P\left[ \sup_{\mbox{\boldmath${\theta}$}\mathbf{\in }\zeta _{j}}\left\vert
Q_{n}\left( \mbox{\boldmath${\theta}$}\right) -\bar{Q}_{n}\left( %
\mbox{\boldmath${\theta}$}\right) \right\vert >\varepsilon \right] \leq P%
\left[ Z_{n,j}^{\ast }>\varepsilon /2\right] .
\end{equation*}%
Next, choose $n_{0}$ so that%
\begin{equation*}
P\left[ Z_{n,j}^{\ast }>\varepsilon /2\right] \leq \frac{\varepsilon }{K}
\end{equation*}%
for all $n\geq n_{0}$ and all $j=1,\ldots ,K$ by Assumptions 1-8 since $K$
is finite. Hence,%
\begin{equation*}
P\left[ \sup_{\mbox{\boldmath${\theta}$}\mathbf{\in {\mbox{\boldmath${%
\Theta}$}}}}\left\vert Q_{n}\left( \mbox{\boldmath${\theta}$}\right) -\bar{Q}%
_{n}\left( \mbox{\boldmath${\theta}$}\right) \right\vert >\varepsilon \right]
\rightarrow 0\text{.}
\end{equation*}%
It then follows that $Q_{n}\left( \mbox{\boldmath${\theta}$}\right) $
satisfies a uniform law of large numbers and the consistency of $\widetilde{%
\mbox{\boldmath${\theta}$}}_{n}$ followed by the usual argument.
\end{proof}

\begin{lemma}
\label{lemmaclts}Under Assumptions 1-8, $\mathbf{D}_{n}^{-1/2}\sqrt{n}\left( 
\widetilde{\mbox{\boldmath${\theta}$}}_{n}^{s}-\mbox{\boldmath${\theta}$}%
_{n}^{p}\right) \overset{d}{\rightarrow }N\left( 0,\mathbf{I}_{P}\right) $
where 
\begin{equation*}
\mathbf{D}_{n}=\left( -\mathbf{H}_{n}+\left( -\mathbf{H}_{n}^{d}\right)
\right) ^{-1}\left( \mathbf{B}_{n}+\left( -\mathbf{H}_{n}^{d}\right) 
\mathbf{C}_{n}\left( -\mathbf{H}_{n}^{d}\right) \right) \left( -\mathbf{H}%
_{n}+\left( -\mathbf{H}_{n}^{d}\right) \right) ^{-1}.
\end{equation*}
\end{lemma}

\begin{proof}
The proof follows from Bester and Hansen (2006). By Lemma \ref%
{lemmaconsistency}, we have,%
\begin{eqnarray*}
0 &=&\frac{1}{n}\sum_{t=1}^{n}\bigtriangledown l_{t}\left( \mathbf{y}%
_{rep}^{t},\widetilde{\mbox{\boldmath${\theta}$}}_{n}^{s}\right) +\left( -%
\mathbf{H}_{n}^{d}\right) \left( \widehat{{\mbox{\boldmath${\theta}$}}}%
_{n}\left( \mathbf{y}\right) -\widetilde{\mbox{\boldmath${\theta}$}}%
_{n}^{s}\right) \\
&=&\frac{1}{n}\sum_{t=1}^{n}\bigtriangledown l_{t}\left( \mathbf{y}%
_{rep}^{t},\mbox{\boldmath${\theta}$}_{n}^{p}\right) +\left( -\mathbf{H}%
_{n}^{d}\right) \left( \widehat{{\mbox{\boldmath${\theta}$}}}_{n}\left( 
\mathbf{y}\right) -\mbox{\boldmath${\theta}$}_{n}^{p}\right) +\frac{1}{n}%
\sum_{t=1}^{n}\bigtriangledown ^{2}l_{t}\left( \mathbf{y}_{rep}^{t},%
\widetilde{\mbox{\boldmath${\theta}$}}_{n3}\right) \left( \widetilde{%
\mbox{\boldmath${\theta}$}}_{n}^{s}-\mbox{\boldmath${\theta}$}_{n}^{p}\right)
\\
&&-\left( -\mathbf{H}_{n}^{d}\right) \left( \widetilde{\mbox{\boldmath${%
\theta}$}}_{n}^{s}-\mbox{\boldmath${\theta}$}_{n}^{p}\right)
\end{eqnarray*}%
where $\widetilde{\mbox{\boldmath${\theta}$}}_{n3}$ is an intermediate value
between $\widetilde{\mbox{\boldmath${\theta}$}}_{n}^{s}$ and $%
\mbox{\boldmath${\theta}$}_{n}^{p}$. It follows that%
\begin{eqnarray*}
\sqrt{n}\left( \widetilde{\mbox{\boldmath${\theta}$}}_{n}^{s}-%
\mbox{\boldmath${\theta}$}_{n}^{p}\right) &=&\left(
-n^{-1}\sum_{t=1}^{n}\bigtriangledown ^{2}l_{t}\left( \mathbf{y}_{rep}^{t},%
\widetilde{\mbox{\boldmath${\theta}$}}_{n3}\right) +\left( -\mathbf{H}%
_{n}^{d}\right) \right) ^{-1}\times \\
&&\left( n^{-1/2}\sum_{t=1}^{n}\bigtriangledown l_{t}\left( \mathbf{y}%
_{rep}^{t},\mbox{\boldmath${\theta}$}_{n}^{p}\right) +\left( -\mathbf{H}%
_{n}^{d}\right) \sqrt{n}\left( \widehat{{\mbox{\boldmath${\theta}$}}}%
_{n}\left( \mathbf{y}\right) -\mbox{\boldmath${\theta}$}_{n}^{p}\right)
\right) .
\end{eqnarray*}%
Under the assumptions, we have%
\begin{equation*}
-n^{-1}\sum_{t=1}^{n}\bigtriangledown ^{2}l_{t}\left( \mathbf{y}_{rep}^{t},%
\widetilde{\mbox{\boldmath${\theta}$}}_{n3}\right) \overset{p}{\rightarrow }-%
\mathbf{H}_{n},\text{ }
\end{equation*}%
\begin{equation*}
\mathbf{B}_{n}^{-1/2}n^{-1/2}\sum_{t=1}^{n}\bigtriangledown l_{t}\left( 
\mathbf{y}_{rep}^{t},\mbox{\boldmath${\theta}$}_{n}^{p}\right) \overset{d}{%
\rightarrow }N\left( 0,\mathbf{I}_{P}\right) ,\text{ }\mathbf{C}_{n}^{-1/2}%
\sqrt{n}\left( \widehat{{\mbox{\boldmath${\theta}$}}}_{n}\left( \mathbf{y}%
\right) -\mbox{\boldmath${\theta}$}_{n}^{p}\right) \overset{d}{\rightarrow }%
N\left( 0,\mathbf{I}_{P}\right) .
\end{equation*}%
Note that $Var\left( n^{-1/2}\sum_{t=1}^{n}\bigtriangledown l_{t}\left( 
\mathbf{y}^{t},\mbox{\boldmath${\theta}$}_{n}^{p}\right) \right) \rightarrow 
\mathbf{B}_{n}$ as $n\rightarrow \infty $. By the central limit theorem and
the Cramer-Wold device, we get%
\begin{equation*}
\mathbf{D}_{n}^{-\frac{1}{2}}\sqrt{n}\left( \widetilde{\mbox{\boldmath${%
\theta}$}}_{n}^{s}-\mbox{\boldmath${\theta}$}_{n}^{p}\right) \overset{d}{%
\rightarrow }N\left( 0,\mathbf{I}_{P}\right) \text{ }
\end{equation*}%
where $\mathbf{D}_{n}=\left( -\mathbf{H}_{n}+\left( -\mathbf{H}%
_{n}^{d}\right) \right) ^{-1}\left( \mathbf{B}_{n}+\left( -\mathbf{H}%
_{n}^{d}\right) \mathbf{C}_{n}\left( -\mathbf{H}_{n}^{d}\right) \right)
\left( -\mathbf{H}_{n}+\left( -\mathbf{H}_{n}^{d}\right) \right) ^{-1}$.
\end{proof}

\begin{lemma}
\label{lemmajoints}Under Assumption 1-8, the asymptotic joint distribution
of $\sqrt{n}\left( \widetilde{\mbox{\boldmath${\theta}$}}_{n}^{s}-%
\mbox{\boldmath${\theta}$}_{n}^{p}\right) $ , $\sqrt{n}\left( \widehat{{%
\mbox{\boldmath${\theta}$}}}_{n}\left( \mathbf{y}\right) -%
\mbox{\boldmath${\theta}$}_{n}^{p}\right) $ and $\sqrt{n}\left( \widehat{{%
\mbox{\boldmath${\theta}$}}}_{n}\left( \mathbf{y}_{rep}\right) -%
\mbox{\boldmath${\theta}$}_{n}^{p}\right) $ is%
\begin{equation*}
\left[ 
\begin{array}{ccc}
\mathbf{D}_{n} & \mathbf{F}_{n} & \mathbf{G}_{n} \\ 
\mathbf{F}_{n} & \mathbf{C}_{n} & \mathbf{0} \\ 
\mathbf{G}_{n} & \mathbf{0} & \mathbf{C}_{n}%
\end{array}%
\right] ^{-1/2}\left( 
\begin{array}{c}
\sqrt{n}\left( \widetilde{\mbox{\boldmath${\theta}$}}_{n}^{s}-%
\mbox{\boldmath${\theta}$}_{n}^{p}\right) \\ 
\sqrt{n}\left( \widehat{{\mbox{\boldmath${\theta}$}}}_{n}\left( \mathbf{y}%
\right) -\mbox{\boldmath${\theta}$}_{n}^{p}\right) \\ 
\sqrt{n}\left( \widehat{{\mbox{\boldmath${\theta}$}}}_{n}\left( \mathbf{y}%
_{rep}\right) -\mbox{\boldmath${\theta}$}_{n}^{p}\right)%
\end{array}%
\right) \overset{d}{\rightarrow }N\left( 0,\mathbf{I}_{3P}\right) ,
\end{equation*}%
where %$\mathbf{D}_{n}=\left( -\mathbf{H}_{n}+\left( -\mathbf{H}%
%_{n}^{d}\right) \right) ^{-1}\left( \mathbf{B}_{n}+\left( -\mathbf{H}%
%_{n}^{d}\right) \mathbf{C}_{n}\left( -\mathbf{H}_{n}^{d}\right) \right)
%\left( -\mathbf{H}_{n}+\left( -\mathbf{H}_{n}^{d}\right) \right) ^{-1}$, 
$\mathbf{F}_{n}=\left( -\mathbf{H}_{n}+\left( -\mathbf{H}_{n}^{d}\right)
\right) ^{-1}\left( -\mathbf{H}_{n}^{d}\right) \mathbf{C}_{n}$ and $%
\mathbf{G}_{n}=\left( -\mathbf{H}_{n}+\left( -\mathbf{H}_{n}^{d}\right)
\right) ^{-1}\mathbf{B}_{n}\left( -\mathbf{H}_{n}\right) ^{-1}$.
\end{lemma}

\begin{proof}
By Lemma \ref{lemmaclts}, we have%
\begin{eqnarray*}
\sqrt{n}\left( \widetilde{\mbox{\boldmath${\theta}$}}_{n}^{s}-%
\mbox{\boldmath${\theta}$}_{n}^{p}\right) &=&\left(
-n^{-1}\sum_{t=1}^{n}\bigtriangledown ^{2}l_{t}\left( \mathbf{y}_{rep}^{t},%
\widetilde{\mbox{\boldmath${\theta}$}}_{n3}\right) +\left( -\mathbf{H}%
_{n}^{d}\right) \right) ^{-1}\times \\
&&\left( n^{-1/2}\sum_{t=1}^{n}\bigtriangledown l_{t}\left( \mathbf{y}%
_{rep}^{t},\mbox{\boldmath${\theta}$}_{n}^{p}\right) +\left( -\mathbf{H}%
_{n}^{d}\right) \sqrt{n}\left( \widehat{{\mbox{\boldmath${\theta}$}}}%
_{n}\left( \mathbf{y}\right) -\mbox{\boldmath${\theta}$}_{n}^{p}\right)
\right) .
\end{eqnarray*}%
\begin{equation*}
\sqrt{n}\left( \widehat{{\mbox{\boldmath${\theta}$}}}_{n}\left( \mathbf{y}%
_{rep}\right) -\mbox{\boldmath${\theta}$}_{n}^{p}\right) =\left(
-n^{-1}\sum_{t=1}^{n}\bigtriangledown ^{2}l_{t}\left( \mathbf{y}_{rep}^{t},%
\widetilde{\mbox{\boldmath${\theta}$}}_{n4}\right) \right)
^{-1}n^{-1/2}\sum_{t=1}^{n}\bigtriangledown l_{t}\left( \mathbf{y}_{rep}^{t},%
\mbox{\boldmath${\theta}$}_{n}^{p}\right) ,
\end{equation*}%
where $\widetilde{\mbox{\boldmath${\theta}$}}_{n4}$ is an intermediate value
between $\widehat{{\mbox{\boldmath${\theta}$}}}_{n}\left( \mathbf{y}%
_{rep}\right) $ and $\mbox{\boldmath${\theta}$}_{n}^{p}$. Hence, we have%
\begin{eqnarray*}
&&Cov\left( \sqrt{n}\left( \widetilde{\mbox{\boldmath${\theta}$}}_{n}^{s}-%
\mbox{\boldmath${\theta}$}_{n}^{p}\right) ,\sqrt{n}\left( \widehat{{%
\mbox{\boldmath${\theta}$}}}_{n}\left( \mathbf{y}\right) -%
\mbox{\boldmath${\theta}$}_{n}^{p}\right) \right) \\
&=&E\left( \sqrt{n}\left( \widetilde{\mbox{\boldmath${\theta}$}}_{n}^{s}-%
\mbox{\boldmath${\theta}$}_{n}^{p}\right) \sqrt{n}\left( \widehat{{%
\mbox{\boldmath${\theta}$}}}_{n}\left( \mathbf{y}\right) -%
\mbox{\boldmath${\theta}$}_{n}^{p}\right) ^{\prime }\right) +o\left( 1\right)
\\
&=&E\left[ 
\begin{array}{c}
\left\{ -n^{-1}\sum_{t=1}^{n}\bigtriangledown ^{2}l_{t}\left( \mathbf{y}%
_{rep}^{t},\widetilde{\mbox{\boldmath${\theta}$}}_{n3}\right) +\left( -%
\mathbf{H}_{n}^{d}\right) \right\} ^{-1}\left( -\mathbf{H}_{n}^{d}\right)
\\ 
\times \sqrt{n}\left( \widehat{{\mbox{\boldmath${\theta}$}}}_{n}\left( 
\mathbf{y}\right) -\mbox{\boldmath${\theta}$}_{n}^{p}\right) \sqrt{n}\left( 
\widehat{{\mbox{\boldmath${\theta}$}}}_{n}\left( \mathbf{y}\right) -%
\mbox{\boldmath${\theta}$}_{n}^{p}\right) ^{\prime }%
\end{array}%
\right] +o\left( 1\right) \\
&=&\left( -\mathbf{H}_{n}+\left( -\mathbf{H}_{n}^{d}\right) \right)
^{-1}\left( -\mathbf{H}_{n}^{d}\right) \mathbf{C}_{n}+o\left( 1\right)
\end{eqnarray*}%
and%
\begin{eqnarray*}
&&Cov\left( \sqrt{n}\left( \widetilde{\mbox{\boldmath${\theta}$}}_{n}^{s}-%
\mbox{\boldmath${\theta}$}_{n}^{p}\right) ,\sqrt{n}\left( \widehat{{%
\mbox{\boldmath${\theta}$}}}_{n}\left( \mathbf{y}_{rep}\right) -%
\mbox{\boldmath${\theta}$}_{n}^{p}\right) \right) \\
&=&E\left( \sqrt{n}\left( \widetilde{\mbox{\boldmath${\theta}$}}_{n}^{s}-%
\mbox{\boldmath${\theta}$}_{n}^{p}\right) \sqrt{n}\left( \widehat{{%
\mbox{\boldmath${\theta}$}}}_{n}\left( \mathbf{y}_{rep}\right) -%
\mbox{\boldmath${\theta}$}_{n}^{p}\right) ^{\prime }\right) +o\left( 1\right)
\\
&=&E\left[ 
\begin{array}{c}
\left\{ -n^{-1}\sum_{t=1}^{n}\bigtriangledown ^{2}l_{t}\left( \mathbf{y}%
_{rep}^{t},\widetilde{\mbox{\boldmath${\theta}$}}_{n3}\right) +\left( -%
\mathbf{H}_{n}^{d}\right) \right\}
^{-1}n^{-1/2}\sum_{t=1}^{n}\bigtriangledown l_{t}\left( \mathbf{y}_{rep}^{t},%
\mbox{\boldmath${\theta}$}_{n}^{p}\right) \\ 
\times n^{-1/2}\sum_{t=1}^{n}\bigtriangledown l_{t}\left( \mathbf{y}%
_{rep}^{t},\mbox{\boldmath${\theta}$}_{n}^{p}\right) \left(
-n^{-1}\sum_{t=1}^{n}\bigtriangledown ^{2}l_{t}\left( \mathbf{y}_{rep}^{t},%
\widetilde{\mbox{\boldmath${\theta}$}}_{n4}\right) \right) ^{-1}%
\end{array}%
\right] +o\left( 1\right) \\
&=&\left( -\mathbf{H}_{n}+\left( -\mathbf{H}_{n}^{d}\right) \right) ^{-1}%
\mathbf{B}_{n}\left( -\mathbf{H}_{n}\right) ^{-1}+o\left( 1\right)
\end{eqnarray*}%
Then we have%
\begin{equation*}
\left[ 
\begin{array}{ccc}
\mathbf{D}_{n} & \mathbf{F}_{n} & \mathbf{G}_{n} \\ 
\mathbf{F}_{n} & \mathbf{C}_{n} & \mathbf{0} \\ 
\mathbf{G}_{n} & \mathbf{0} & \mathbf{C}_{n}%
\end{array}%
\right] ^{-1/2}\left( 
\begin{array}{c}
\sqrt{n}\left( \widetilde{\mbox{\boldmath${\theta}$}}_{n}^{s}-%
\mbox{\boldmath${\theta}$}_{n}^{p}\right) \\ 
\sqrt{n}\left( \widehat{{\mbox{\boldmath${\theta}$}}}_{n}\left( \mathbf{y}%
\right) -\mbox{\boldmath${\theta}$}_{n}^{p}\right) \\ 
\sqrt{n}\left( \widehat{{\mbox{\boldmath${\theta}$}}}_{n}\left( \mathbf{y}%
_{rep}\right) -\mbox{\boldmath${\theta}$}_{n}^{p}\right)%
\end{array}%
\right) \overset{d}{\rightarrow }N\left( 0,\mathbf{I}_{3P}\right) ,
\end{equation*}%
where 
\begin{equation*}
\mathbf{D}_{n}=\left( -\mathbf{H}_{n}+\left( -\mathbf{H}_{n}^{d}\right)
\right) ^{-1}\left( \mathbf{B}_{n}+\left( -\mathbf{H}_{n}^{d}\right) 
\mathbf{C}_{n}\left( -\mathbf{H}_{n}^{d}\right) \right) \left( -\mathbf{H}%
_{n}+\left( -\mathbf{H}_{n}^{d}\right) \right) ^{-1},
\end{equation*}
$\mathbf{F}_{n}=\left( -\mathbf{H}_{n}+\left( -\mathbf{H}_{n}^{d}\right)
\right) ^{-1}\left( -\mathbf{H}_{n}^{d}\right) \mathbf{C}_{n}$ and $%
\mathbf{G}_{n}=\left( -\mathbf{H}_{n}+\left( -\mathbf{H}_{n}^{d}\right)
\right) ^{-1}\mathbf{B}_{n}\left( -\mathbf{H}_{n}\right) ^{-1}$.
\end{proof}

%\begin{theorem}
%\label{thm61}Under Assumptions 1-15, assume the prior $p({%
%\mbox{\boldmath${\theta}$}})=O_{p}(1)$, it can be shown that,
%\begin{eqnarray*}
%E_{\mathbf{y}}E_{\mathbf{y}_{rep}}\left( -2\ln p\left( \mathbf{y}_{rep}|%
%\mathbf{y}\right) \right) &=&E_{\mathbf{y}}\left[ -2\ln p\left( \mathbf{y}|%
%\mathbf{\hat{\theta}}_{m}\right) +\left( 1+\ln 2\right) P\right] +o\left(
%1\right) \\
%&=&E_{\mathbf{y}}\left[ -2\ln p\left( \mathbf{y}|\mathbf{y}\right) +P\right]
%+o\left( 1\right) .
%\end{eqnarray*}
%\end{theorem}

\subsubsection{Proof of Theorem 3.1}
%\protect\ref{riskvtic}}
We write $\mathbf{H}_n\left(\boldsymbol{\theta}_n^p\right)$ as $\mathbf{H}_n, \mathbf{B}_n\left(\boldsymbol{\theta}_n^p\right)$ as $\mathbf{B}_n$, and let $\mathbf{C}_n=\mathbf{H}_n^{-1} \mathbf{B}_n \mathbf{H}_n^{-1}$. Note that
\begin{equation}
	\label{vb consistency mle}
	\overline{\boldsymbol{\theta}}^{VB}(\mathbf{y})=\widehat{\boldsymbol{\theta}}_n(\mathbf{y})+O_p\left(n^{-3/4}\right),
\end{equation}
in \citet{Zhang_2024}. Then, we have
\begin{equation}
\label{vb consistency psedo true}
\overline{\boldsymbol{\theta}}^{VB}(\mathbf{y})=\boldsymbol{\theta}_n^p+O_p\left(n^{-1/2}\right),
\end{equation}
\begin{equation}
	\label{var score f}
	\frac{1}{\sqrt{n}} \mathbf{B}_n^{-1/2} \frac{\partial \ln p\left(\mathbf{y}_{\text {rep}}| \boldsymbol{\theta}_n^p\right)}{\partial \boldsymbol{\theta}} \xrightarrow{d} N\left(0, \mathbf{I}_P\right),
\end{equation}
and
\begin{equation}
	\label{vb to pesedo true}
\mathbf{C}_n^{-1/2} \sqrt{n}\left(\widehat{\boldsymbol{\theta}}_n(\mathbf{y})-\boldsymbol{\theta}_n^p\right) \xrightarrow{d} N\left(0, \mathbf{I}_P\right).
\end{equation}
We are now in the position to prove Theorem 3.1. Note that
%\ref{riskvtic}
$$
\begin{aligned}
	& E_{\mathbf{y}} E_{\mathbf{y}_{\text {rep}}}\left(-2 \ln p\left(\mathbf{y}_{\text {rep}} |\overline{\boldsymbol{\theta}}^{VB}\left(\mathbf{y}\right)\right)\right) \\
	= & E_{\mathbf{y}} E_{\mathbf{y}_{\text {rep}}}\left(-2 \ln p\left(\mathbf{y}_{\text {rep}} |\overline{\boldsymbol{\theta}}^{VB}\left(\mathbf{y}_{\text{rep}}\right)\right)\right) \\
	+ & E_{\mathbf{y}} E_{\mathbf{y}_{\text {rep}}}\left(-2 \ln p\left(\mathbf{y}_{\text {rep}} | \boldsymbol{\theta}_n^p\right)\right)-E_{\mathbf{y}} E_{\mathbf{y}_{\text {rep}}}\left(-2 \ln p\left(\mathbf{y}_{\text {rep}}| \overline{\boldsymbol{\theta}}^{VB}\left(\mathbf{y}_{\text {rep}}\right)\right)\right) \\
	+ & E_{\mathbf{y}} E_{\mathbf{y}_{\text {rep }}}\left(-2 \ln p\left(\mathbf{y}_{\text {rep }}| \overline{\boldsymbol{\theta}}^{VB}(\mathbf{y})\right)\right)-E_{\mathbf{y}} E_{\mathbf{y}_{\text {rep }}}\left(-2 \ln p\left(\mathbf{y}_{\text {rep }}| \boldsymbol{\theta}_n^p\right)\right)\\
	= & T_1 + T_2 + T_3
\end{aligned}
$$
where 
$$
T_1 =E_{\mathbf{y}} E_{\mathbf{y}_{\text {rep}}}\left(-2 \ln p\left(\mathbf{y}_{\text {rep}} |\overline{\boldsymbol{\theta}}^{VB}\left(\mathbf{y}_{\text{rep}}\right)\right)\right),
$$
$$
T_2 = E_{\mathbf{y}} E_{\mathbf{y}_{\text {rep}}}\left(-2 \ln p\left(\mathbf{y}_{\text {rep}} | \boldsymbol{\theta}_n^p\right)\right)-E_{\mathbf{y}} E_{\mathbf{y}_{\text {rep}}}\left(-2 \ln p\left(\mathbf{y}_{\text {rep}}| \overline{\boldsymbol{\theta}}^{VB}\left(\mathbf{y}_{\text {rep}}\right)\right)\right),
$$
and
$$
T_3 = E_{\mathbf{y}} E_{\mathbf{y}_{\text {rep }}}\left(-2 \ln p\left(\mathbf{y}_{\text {rep }}| \overline{\boldsymbol{\theta}}^{VB}(\mathbf{y})\right)\right)-E_{\mathbf{y}} E_{\mathbf{y}_{\text {rep }}}\left(-2 \ln p\left(\mathbf{y}_{\text {rep }}| \boldsymbol{\theta}_n^p\right)\right).
$$
Now let us analyze $T_2$ and $T_3$. First, expanding $\ln p\left(\mathbf{y}_{\text{rep}} |\boldsymbol{\theta}_n^p\right)$ at $\overline{\boldsymbol{\theta}}^{VB} \left(\mathbf{y}_{\text{rep}}\right)$
\begin{equation}
	\label{theorem31_t21}
	\begin{aligned}
		& \ln p\left(\mathbf{y}_{\text {rep}}| \boldsymbol{\theta}_n^p\right) \\
		= & \ln p\left(\mathbf{y}_{\text {rep}}| \overline{\boldsymbol{\theta}}^{VB}\left(\mathbf{y}_{\text {rep}}\right)\right)+\frac{\partial \ln p\left(\mathbf{y}_{\text {rep}}| \overline{\boldsymbol{\theta}}^{VB}\left(\mathbf{y}_{\text {rep}}\right)\right)}{\partial \boldsymbol{\theta}^{\prime}}\left(\boldsymbol{\theta}_n^p-\overline{\boldsymbol{\theta}}^{VB}\left(\mathbf{y}_{\text {rep}}\right)\right) \\
		& +\frac{1}{2}\left(\boldsymbol{\theta}_n^p-\overline{\boldsymbol{\theta}}^{VB}\left(\mathbf{y}_{\text {rep}}\right)\right)^{\prime} \frac{\partial^2 \ln p\left(\mathbf{y}_{\text {rep}}| \overline{\boldsymbol{\theta}}^{VB}\left(\mathbf{y}_{\text {rep}}\right)\right)}{\partial \boldsymbol{\theta} \partial \boldsymbol{\theta}^{\prime}}\left(\boldsymbol{\theta}_n^p-\overline{\boldsymbol{\theta}}^{VB}\left(\mathbf{y}_{\text {rep}}\right)\right)\\
		&+\frac{1}{6}\left[\left(\boldsymbol{\theta}_n^p-\overline{\boldsymbol{\theta}}^{VB}\left(\mathbf{y}_{\text {rep}}\right)\right) \otimes\left(\boldsymbol{\theta}_n^p-\overline{\boldsymbol{\theta}}^{VB}\left(\mathbf{y}_{\text {rep}}\right)\right)\right]^{\prime} \frac{\partial^3 \ln p\left(\mathbf{y}_{\text {rep}}| \overline{\boldsymbol{\theta}}^{*VB}\left(\mathbf{y}_{\text {rep}}\right)\right)}{\partial \boldsymbol{\theta} \partial \boldsymbol{\theta}^{\prime} \partial \boldsymbol{\theta}}\left(\boldsymbol{\theta}_n^p-\overline{\boldsymbol{\theta}}^{VB}\left(\mathbf{y}_{\text {rep}}\right)\right)		
	\end{aligned}
\end{equation}
where $\overline{\boldsymbol{\theta}}^{*VB}\left(\mathbf{y}_{\text {rep}}\right)$ lies between $\boldsymbol{\theta}_n^p$ and $\overline{\boldsymbol{\theta}}^{VB}\left(\mathbf{y}_{\text {rep}}\right)$. Note that the last term can be written as
\begin{equation}
	\begin{aligned}
		R T_{1, n} & =\frac{1}{6} \frac{1}{\sqrt{n}}\left[\sqrt{n}\left(\boldsymbol{\theta}_n^p-\overline{\boldsymbol{\theta}}^{VB}\left(\mathbf{y}_{\text {rep}}\right)\right) \otimes \sqrt{n}\left(\boldsymbol{\theta}_n^p-\overline{\boldsymbol{\theta}}^{VB}\left(\mathbf{y}_{\text{rep}}\right)\right)\right]^{\prime} \\
		& \times \frac{1}{n} \sum_{t=1}^n \nabla^3 l_t\left(\overline{\boldsymbol{\theta}}^{*VB}\left(\mathbf{y}_{rep}\right)\right) \sqrt{n}\left(\boldsymbol{\theta}_n^p-\overline{\boldsymbol{\theta}}^{VB}\left(\mathbf{y}_{rep}\right)\right)
	\end{aligned}
\end{equation}
where $\sqrt{n}\left(\boldsymbol{\theta}_n^p-\overline{\boldsymbol{\theta}}^{VB}\left(\mathbf{y}_{\text {rep}}\right)\right)=O_p(1)$ by Assumptions 1-8 and

$$
\begin{aligned}
	\left\|\frac{1}{n} \sum_{t=1}^n \nabla^3 l_t\left(\overline{\boldsymbol{\theta}}^{*VB}\left(\mathbf{y}_{rep}\right)\right)\right\| & \leq \frac{1}{n} \sum_{t=1}^n\left\|\nabla^3 l_t\left(\overline{\boldsymbol{\theta}}^{*VB}\left(\mathbf{y}_{r e p}\right)\right)\right\| \leq \frac{1}{n} \sum_{t=1}^n \sup _{\boldsymbol{\theta} \in \boldsymbol{\Theta}}\left\|\nabla^j l_t(\boldsymbol{\theta})\right\| \\
	& \leq \frac{1}{n} \sum_{t=1}^n M_t\left(\mathbf{y}_t\right)
\end{aligned}
$$
by Assumption 5. It can be shown that
$$
P\left(\frac{1}{n} \sum_{t=1}^n M_t\left(\mathbf{y}_t\right)>C\right) \leq \frac{\frac{1}{n} \sum_{t=1}^n E\left(M_t\left(\mathbf{y}_t\right)\right)}{C} \leq \frac{\sup_t E\left(M_t\left(\mathbf{y}_t\right)\right)}{C} \leq \frac{M}{C}
$$
by the Markov inequality. Let $\varepsilon=M / C$, for any $\varepsilon$, there exists a constant $C=M / \varepsilon$ such that
$$
P\left(\frac{1}{n} \sum_{t=1}^n M_t\left(\mathbf{y}_t\right)>C\right) \leq \varepsilon .
$$
Thus, $\frac{1}{n} \sum_{t=1}^n M_t\left(\mathbf{y}_t\right)=O_p(1)$ and $\left\|\frac{1}{n} \sum_{t=1}^n \nabla^3 l_t\left(\overline{\boldsymbol{\theta}}^{*VB}\left(\mathbf{y}_{\text {rep }}\right)\right)\right\|=O_p(1)$. Hence, we have $R T_{1, n}=O_p\left(n^{-1 / 2}\right)$.

We can rewrite (\ref{theorem31_t21}) as 
$$
\begin{aligned}
	& \ln p\left(\mathbf{y}_{\text {rep}}| \boldsymbol{\theta}_n^p\right) \\
	= & \ln p\left(\mathbf{y}_{\text {rep}}| \overline{\boldsymbol{\theta}}^{VB}\left(\mathbf{y}_{\text {rep}}\right)\right)+\frac{\partial \ln p\left(\mathbf{y}_{\text {rep}}| \overline{\boldsymbol{\theta}}^{VB}\left(\mathbf{y}_{\text {rep}}\right)\right)}{\partial \boldsymbol{\theta}^{\prime}}\left(\boldsymbol{\theta}_n^p-\overline{\boldsymbol{\theta}}^{VB}\left(\mathbf{y}_{\text {rep}}\right)\right) \\
	& +\frac{1}{2}\left(\boldsymbol{\theta}_n^p-\overline{\boldsymbol{\theta}}^{VB}\left(\mathbf{y}_{\text {rep}}\right)\right)^{\prime} \frac{\partial^2 \ln p\left(\mathbf{y}_{\text {rep}}| \overline{\boldsymbol{\theta}}^{VB}\left(\mathbf{y}_{\text {rep}}\right)\right)}{\partial \boldsymbol{\theta} \partial \boldsymbol{\theta}^{\prime}}\left(\boldsymbol{\theta}_n^p-\overline{\boldsymbol{\theta}}^{VB}\left(\mathbf{y}_{\text {rep}}\right)\right) + RT_{1,n}\\
	= & \ln p\left(\mathbf{y}_{\text {rep}}| \overline{\boldsymbol{\theta}}^{VB}\left(\mathbf{y}_{\text {rep}}\right)\right)+\frac{\partial \ln p\left(\mathbf{y}_{\text {rep}}| \widehat{\boldsymbol{\theta}}\left(\mathbf{y}_{\text {rep}}\right)\right)}{\partial \boldsymbol{\theta}^{\prime}}\left(\boldsymbol{\theta}_n^p-\overline{\boldsymbol{\theta}}^{VB}\left(\mathbf{y}_{\text {rep}}\right)\right) \\
	& +  \frac{1}{2}\left(\boldsymbol{\theta}_n^p-\overline{\boldsymbol{\theta}}^{VB}\left(\mathbf{y}_{\text {rep}}\right)\right)^{\prime} \frac{\partial^2 \ln p\left(\mathbf{y}_{\text {rep}}| \overline{\boldsymbol{\theta}}^{VB}\left(\mathbf{y}_{\text {rep}}\right)\right)}{\partial \boldsymbol{\theta} \partial \boldsymbol{\theta}^{\prime}}\left(\boldsymbol{\theta}_n^p-\overline{\boldsymbol{\theta}}^{VB}\left(\mathbf{y}_{\text {rep}}\right)\right)\\
	& +\left(\frac{\partial \ln p\left(\mathbf{y}_{rep}|\overline{\boldsymbol{\theta}}^{VB}\left(\mathbf{y}_{\text {rep}}\right)\right)}{\partial \boldsymbol{\theta}^{\prime}}-\frac{\partial \ln p\left(\mathbf{y}_{rep}|\widehat{\boldsymbol{\theta}}_n\left(\mathbf{y}_{\text {rep}}\right)\right)}{\partial \boldsymbol{\theta}^{\prime}}\right)\left(\boldsymbol{\theta}_n^p-\overline{\boldsymbol{\theta}}^{VB}\left(\mathbf{y}_{\text {rep}}\right)\right)+R T_{1,n}\\
	= & \ln p\left(\mathbf{y}_{\text {rep}}| \overline{\boldsymbol{\theta}}^{VB}\left(\mathbf{y}_{\text {rep}}\right)\right)+\frac{\partial \ln p\left(\mathbf{y}_{\text {rep}}| \widehat{\boldsymbol{\theta}}\left(\mathbf{y}_{\text {rep}}\right)\right)}{\partial \boldsymbol{\theta}^{\prime}}\left(\boldsymbol{\theta}_n^p-\overline{\boldsymbol{\theta}}^{VB}\left(\mathbf{y}_{\text {rep}}\right)\right) \\
	& +  \frac{1}{2}\left(\boldsymbol{\theta}_n^p-\overline{\boldsymbol{\theta}}^{VB}\left(\mathbf{y}_{\text {rep}}\right)\right)^{\prime} \frac{\partial^2 \ln p\left(\mathbf{y}_{\text {rep}}| \overline{\boldsymbol{\theta}}^{VB}\left(\mathbf{y}_{\text {rep}}\right)\right)}{\partial \boldsymbol{\theta} \partial \boldsymbol{\theta}^{\prime}}\left(\boldsymbol{\theta}_n^p-\overline{\boldsymbol{\theta}}^{VB}\left(\mathbf{y}_{\text {rep}}\right)\right) + RT_n
\end{aligned}
$$
from (\ref{vb consistency mle}) where $RT_n = RT_{1,n} + RT_{2,n}$ with
\begin{equation}
	\label{rt 2n}
	RT_{2,n} = \left(\frac{\partial \ln p\left(\mathbf{y}_{rep}|\overline{\boldsymbol{\theta}}^{VB}\left(\mathbf{y}_{\text {rep}}\right)\right)}{\partial \boldsymbol{\theta}^{\prime}}-\frac{\partial \ln p\left(\mathbf{y}_{rep}|\widehat{\boldsymbol{\theta}}_n\left(\mathbf{y}_{\text {rep}}\right)\right)}{\partial \boldsymbol{\theta}^{\prime}}\right)\left(\boldsymbol{\theta}_n^p-\overline{\boldsymbol{\theta}}^{VB}\left(\mathbf{y}_{\text {rep}}\right)\right)
\end{equation}
We can rewrite the first term on the right-hand side of (\ref{rt 2n}) as
$$
\begin{aligned}
	& \left(\frac{\partial \ln p\left(\mathbf{y}_{rep}|\overline{\boldsymbol{\theta}}^{VB}\left(\mathbf{y}_{\text {rep}}\right)\right)}{\partial \boldsymbol{\theta}}-\frac{\partial \ln p\left(\mathbf{y}_{rep}|\widehat{\boldsymbol{\theta}}_n\left(\mathbf{y}_{\text {rep}}\right)\right)}{\partial \boldsymbol{\theta}}\right) \\
	= & \frac{1}{n} \frac{\partial^2 \ln p\left(\mathbf{y}_{\text {rep}} \mid \widehat{\boldsymbol{\theta}}_n^{\#}\left(\mathbf{y}_{\text {rep}}\right)\right)}{\partial \boldsymbol{\theta} \partial \boldsymbol{\theta}^{\prime}} n\left(\overline{\boldsymbol{\theta}}^{VB}\left(\mathbf{y}_{\text {rep}}\right)-\widehat{\boldsymbol{\theta}}_n\left(\mathbf{y}_{\text {rep}}\right)\right)=O_p\left(1\right)
\end{aligned}
$$
where $\widehat{\boldsymbol{\theta}}_n^{\#}\left(\mathbf{y}_{\text {rep}}\right)$ lies between $\overline{\boldsymbol{\theta}}^{VB}\left(\mathbf{y}_{\text {rep}}\right)$ and $\widehat{\boldsymbol{\theta}}_n\left(\mathbf{y}_{\text {rep}}\right)$. Thus,
$$
R T_{2, n}=O_p(1) O_p\left(n^{-1 / 2}\right)=O_p\left(n^{-1 / 2}\right)
$$
Hence, we have
\begin{equation}
	\label{Op rt_n}
	R T_n=R T_{1, n}+R T_{2, n}=O_p\left(n^{-1 / 2}\right)
\end{equation}

Now we will consider the expectation of the norm of $R T_{1, n}$ and $R T_{2, n}$. For $R T_{1, n}$, we first consider the term
\begin{equation}
	\label{rt_1n subterm}
	\begin{aligned}
	 & \left[\sqrt{n}\left(\boldsymbol{\theta}_n^p-\overline{\boldsymbol{\theta}}^{VB}\left(\mathbf{y}_{\text {rep}}\right)\right) \otimes \sqrt{n}\left(\boldsymbol{\theta}_n^p-\overline{\boldsymbol{\theta}}^{VB}\left(\mathbf{y}_{\text{rep}}\right)\right)\right]^{\prime} \\
		& \times \frac{1}{n} \sum_{t=1}^n \nabla^3 l_t\left(\overline{\boldsymbol{\theta}}^{*VB}\left(\mathbf{y}_{rep}\right)\right) \sqrt{n}\left(\boldsymbol{\theta}_n^p-\overline{\boldsymbol{\theta}}^{VB}\left(\mathbf{y}_{rep}\right)\right)
	\end{aligned}
\end{equation}
and try to prove that the expectation of (\ref{rt_1n subterm}) is bounded. It can be shown that
\begin{equation}
	\label{Cauchy Schwarz inequ of expectation of tr1n}
	\begin{aligned} E&\left[\left\|\left[\sqrt{n}\left(\boldsymbol{\theta}_n^p-\overline{\boldsymbol{\theta}}^{VB}\left(\mathbf{y}_{r e p}\right)\right) \otimes \sqrt{n}\left(\boldsymbol{\theta}_n^p-\overline{\boldsymbol{\theta}}^{VB}\left(\mathbf{y}_{\text {rep}}\right)\right)\right]^{\prime} \frac{1}{n} \sum_{t=1}^n \nabla^3 l_t\left(\overline{\boldsymbol{\theta}}^{*VB}\left(\mathbf{y}_{r e p}\right)\right) \sqrt{n}\left(\boldsymbol{\theta}_n^p-\overline{\boldsymbol{\theta}}^{VB}\left(\mathbf{y}_{rep}\right)\right)\right\|\right] \\
		\leq & \left(E\left[\left\|\left[\sqrt{n}\left(\boldsymbol{\theta}_n^p-\overline{\boldsymbol{\theta}}^{VB}\left(\mathbf{y}_{rep}\right)\right) \otimes \sqrt{n}\left(\boldsymbol{\theta}_n^p-\overline{\boldsymbol{\theta}}^{VB}\left(\mathbf{y}_{\text {rep}}\right)\right)\right]^{\prime}\right\|^2\right]\right)^{1 / 2} \\
		& \times\left(E\left[\left\|\frac{1}{n} \sum_{t=1}^n \nabla^3 l_t\left(\overline{\boldsymbol{\theta}}^{*VB}\left(\mathbf{y}_{rep}\right)\right) \sqrt{n}\left(\boldsymbol{\theta}_n^p-\overline{\boldsymbol{\theta}}^{VB}\left(\mathbf{y}_{rep}\right)\right)\right\|^2\right]\right)^{1 / 2} \\
		= & \left(E\left[\left\|\sqrt{n}\left(\boldsymbol{\theta}_n^p-\overline{\boldsymbol{\theta}}^{VB}\left(\mathbf{y}_{\text {rep}}\right)\right)\right\|^4\right]\right)^{1 / 2}\left(E\left[\left\|\frac{1}{n} \sum_{t=1}^n \nabla^3 l_t\left(\overline{\boldsymbol{\theta}}^{*VB}\left(\mathbf{y}_{r e p}\right)\right) \sqrt{n}\left(\boldsymbol{\theta}_n^p-\overline{\boldsymbol{\theta}}^{VB}\left(\mathbf{y}_{r e p}\right)\right)\right\|^2\right]\right)^{1 / 2}
	\end{aligned}
\end{equation}
by the Cauchy-Schwarz Inequality and the fact that
$$
\left\|\left[\sqrt{n}\left(\boldsymbol{\theta}_n^p-\overline{\boldsymbol{\theta}}^{VB}\left(\mathbf{y}_{re p}\right)\right) \otimes \sqrt{n}\left(\boldsymbol{\theta}_n^p-\overline{\boldsymbol{\theta}}^{VB}\left(\mathbf{y}_{r e p}\right)\right)\right]^{\prime}\right\|=\left\|\sqrt{n}\left(\boldsymbol{\theta}_n^p-\overline{\boldsymbol{\theta}}^{VB}\left(\mathbf{y}_{r e p}\right)\right)\right\|^2.
$$
To prove that (\ref{Cauchy Schwarz inequ of expectation of tr1n}) is bounded, we need to prove that
\begin{equation}
	\label{expectation of term1 tr1n}
	E\left[\left\|\sqrt{n}\left(\boldsymbol{\theta}_n^p-\overline{\boldsymbol{\theta}}^{VB}\left(\mathbf{y}_{r e p}\right)\right)\right\|^4\right]
\end{equation}
and 
\begin{equation}
	\label{expectation of term2 tr1n}
	E\left[\left\|\frac{1}{n} \sum_{t=1}^n \nabla^3 l_t\left(\overline{\boldsymbol{\theta}}^{*VB}\left(\mathbf{y}_{\text {rep}}\right)\right) \sqrt{n}\left(\boldsymbol{\theta}_n^p-\overline{\boldsymbol{\theta}}^{VB}\left(\mathbf{y}_{\text {rep}}\right)\right)\right\|^2\right]
\end{equation}
are both bounded.

For (\ref{expectation of term1 tr1n}), we have
$$
\begin{aligned}
	& \left(E\left[\left\|\sqrt{n}\left(\boldsymbol{\theta}_n^p-\overline{\boldsymbol{\theta}}^{VB}\left(\mathbf{y}_{\text {rep}}\right)\right)\right\|^4\right]\right)^{1 / 4} \\
	= & \left(E\left[\left\|\sqrt{n}\left(\boldsymbol{\theta}_n^p-\widehat{\boldsymbol{\theta}}_n\left(\mathbf{y}_{\text {rep}}\right)+\widehat{\boldsymbol{\theta}}_n\left(\mathbf{y}_{\text {rep}}\right)-\overline{\boldsymbol{\theta}}^{VB}\left(\mathbf{y}_{\text {rep}}\right)\right)\right\|^4\right]\right)^{1 / 4} \\
	\leq & \left(E\left[\left(\left\|\sqrt{n}\left(\boldsymbol{\theta}_n^p-\widehat{\boldsymbol{\theta}}_n\left(\mathbf{y}_{\text {rep}}\right)\right)\right\|+\left\|\sqrt{n}\left(\widehat{\boldsymbol{\theta}}_n\left(\mathbf{y}_{\text {rep}}\right)-\overline{\boldsymbol{\theta}}^{VB}\left(\mathbf{y}_{\text {rep}}\right)\right)\right\|\right)^4\right]\right)^{1 / 4} \\
	\leq & \left(E\left[\left\|\sqrt{n}\left(\boldsymbol{\theta}_n^p-\widehat{\boldsymbol{\theta}}_n\left(\mathbf{y}_{\text {rep}}\right)\right)\right\|^4\right]\right)^{1 / 4}+\left(E\left[\left\|\sqrt{n}\left(\widehat{\boldsymbol{\theta}}_n\left(\mathbf{y}_{\text {rep}}\right)-\overline{\boldsymbol{\theta}}^{VB}\left(\mathbf{y}_{\text {rep}}\right)\right)\right\|^4\right]\right)^{1 / 4}
\end{aligned}
$$
by the triangular inequality and the Minkowski inequality. To prove that (\ref{expectation of term1 tr1n}) is bounded, it is suffice to show
\begin{equation}
	\label{1 term in Minkowski inequality of CS term1}
	E\left[\left\|\sqrt{n}\left(\boldsymbol{\theta}_n^p-\widehat{\boldsymbol{\theta}}_n\left(\mathbf{y}_{\text {rep}}\right)\right)\right\|^4\right]
\end{equation}
and
\begin{equation}
	\label{2 term in Minkowski inequality of CS term1}
	E\left[\left\|\sqrt{n}\left(\widehat{\boldsymbol{\theta}}_n\left(\mathbf{y}_{rep}\right)-\overline{\boldsymbol{\theta}}^{VB}\left(\mathbf{y}_{\text {rep}}\right)\right)\right\|^4\right]
\end{equation}
are both bounded. \citet{li2024deviance} have proved that
\begin{equation}
	\label{1 term in Minkowski inequality of CS term1 is not infinety}
	E\left[\left\|\sqrt{n}\left(\boldsymbol{\theta}_n^p-\widehat{\boldsymbol{\theta}}_n\left(\mathbf{y}_{\text {rep}}\right)\right)\right\|^4\right] <\infty
\end{equation}
under Assumption 1-8.

For (\ref{2 term in Minkowski inequality of CS term1}), following Theorem1 and Corollary 1 of \citet{han2019statistical}, if we use $\overline{\boldsymbol{\theta}}^{VB}\left(\mathbf{y}_{\text {rep}}\right)$ to approximate $\widehat{\boldsymbol{\theta}}_n\left(\mathbf{y}_{rep}\right)$, the bound of the approximate error is
\begin{equation}
	\left\|\sqrt{n}\left(\widehat{\boldsymbol{\theta}}_n\left(\mathbf{y}_{rep}\right)-\overline{\boldsymbol{\theta}}^{VB}\left(\mathbf{y}_{rep}\right)\right)\right\| \leq \frac{C M^{3 / 2}(\log n)^{d / 2+3 / 2}}{n^{1 / 4}} .
\end{equation}
with a exist constant $C$ and for any $M \geq 1$. Therefore (\ref{2 term in Minkowski inequality of CS term1}) is bounded by
\begin{equation}
	\label{2 term in Minkowski inequality of CS term1 is not infinety}
	E\left[\left\|\sqrt{n}\left(\widehat{\boldsymbol{\theta}}_n\left(\mathbf{y}_{rep}\right)-\overline{\boldsymbol{\theta}}^{VB}\left(\mathbf{y}_{rep}\right)\right)\right\|^4\right] \leq \frac{C^4 M^{6}(\log n)^{2d+6}}{n}=O\left(n^{-1}\right)<\infty.
\end{equation}
Thus, from (\ref{1 term in Minkowski inequality of CS term1 is not infinety}) and (\ref{2 term in Minkowski inequality of CS term1 is not infinety}), we have
\begin{equation}
	\label{expectation of term1 tr1n is not infinety}
	\begin{aligned}
		& \left(E\left[\left\|\sqrt{n}\left(\boldsymbol{\theta}_n^p-\overline{\boldsymbol{\theta}}^{VB}\left(\mathbf{y}_{\text {rep}}\right)\right)\right\|^4\right]\right)^{1 / 4} \\
		\leq & \left(E\left[\left\|\sqrt{n}\left(\boldsymbol{\theta}_n^p-\widehat{\boldsymbol{\theta}}_n\left(\mathbf{y}_{\text {rep}}\right)\right)\right\|^4\right]\right)^{1 / 4}+\left(E\left[\left\|\sqrt{n}\left(\widehat{\boldsymbol{\theta}}_n\left(\mathbf{y}_{\text {rep}}\right)-\overline{\boldsymbol{\theta}}^{VB}\left(\mathbf{y}_{\text {rep}}\right)\right)\right\|^4\right]\right)^{1 / 4}\\
		< & \infty.
	\end{aligned}
\end{equation}

For (\ref{expectation of term2 tr1n}), we have 
\begin{equation}
	\label{expectation of term2 tr1n is not infinety}
	\begin{aligned}
		& E\left[\left\|\frac{1}{n} \sum_{t=1}^n \nabla^3 l_t\left(\overline{\boldsymbol{\theta}}^{*VB}\left(\mathbf{y}_{\text {rep}}\right)\right) \sqrt{n}\left(\boldsymbol{\theta}_n^p-\overline{\boldsymbol{\theta}}^{VB}\left(\mathbf{y}_{\text {rep}}\right)\right)\right\|^2\right] \\
		\leq & E\left[\left\|\frac{1}{n} \sum_{t=1}^n \nabla^3 l_t\left(\overline{\boldsymbol{\theta}}^{*VB}\left(\mathbf{y}_{\text{rep}}\right)\right)\right\|^2\left\|\sqrt{n}\left(\boldsymbol{\theta}_n^p-\overline{\boldsymbol{\theta}}^{VB}\left(\mathbf{y}_{\text {rep}}\right)\right)\right\|^2\right] \\
		\leq & \left(E\left[\left\|\frac{1}{n} \sum_{t=1}^n \nabla^3 l_t\left(\overline{\boldsymbol{\theta}}^{*VB}\left(\mathbf{y}_{\text {rep}}\right)\right)\right\|^4\right]\right)^{1 / 2}\left(E\left[\left\|\sqrt{n}\left(\boldsymbol{\theta}_n^p-\overline{\boldsymbol{\theta}}^{VB}\left(\mathbf{y}_{\text {rep}}\right)\right)\right\|^4\right]\right)^{1 / 2} \\
		< & \infty
	\end{aligned}
\end{equation}
by Assumption 5 and (\ref{expectation of term1 tr1n is not infinety}). Thus, from (\ref{rt_1n subterm}), (\ref{Cauchy Schwarz inequ of expectation of tr1n}), (\ref{expectation of term1 tr1n is not infinety}) and (\ref{expectation of term2 tr1n is not infinety}), we have
\begin{equation}
	\label{expectation of tr1n is small op1}
	\begin{aligned}
		& E\left\|R T_{1, n}\right\| \\
		\leq & \frac{1}{6} \frac{1}{\sqrt{n}}\left(E\left[\left\|\sqrt{n}\left(\boldsymbol{\theta}_n^p-\overline{\boldsymbol{\theta}}^{VB}\left(\mathbf{y}_{\text {rep}}\right)\right)\right\|^4\right]\right)^{1 / 4} \\
		& \times\left(E\left[\left\|\frac{1}{n} \sum_{t=1}^n \nabla^3 l_t\left(\overline{\boldsymbol{\theta}}^{*VB}\left(\mathbf{y}_{\text {rep}}\right)\right) \sqrt{n}\left(\boldsymbol{\theta}_n^p-\overline{\boldsymbol{\theta}}^{VB}\left(\mathbf{y}_{\text {rep}}\right)\right)\right\|^2\right]\right)^{1 / 4} \\
		= & o(1)
	\end{aligned}
\end{equation}

For $RT_{2,n}$, we have 
\begin{equation}
	\label{expectation of tr2n}
	\begin{aligned}
		& E\left\|R T_{2,n}\right\| \\
		\leq & E\left[\left\|\frac{1}{\sqrt{n}}\left(\frac{\partial \ln p\left(\mathbf{y}_{rep}| \overline{\boldsymbol{\theta}}^{VB}\left(\mathbf{y}_{\text {rep}}\right)\right)}{\partial \boldsymbol{\theta}^{\prime}}-\frac{\partial \ln p\left(\mathbf{y}_{\text {rep}}| \widehat{\boldsymbol{\theta}}_n\left(\mathbf{y}_{\text {rep}}\right)\right)}{\partial \boldsymbol{\theta}^{\prime}}\right)\right\|\left\|\sqrt{n}\left(\boldsymbol{\theta}_n^p-\overline{\boldsymbol{\theta}}^{VB}\left(\mathbf{y}_{\text {rep}}\right)\right)\right\|\right] \\
		\leq & \left(E\left[\left\|\frac{1}{\sqrt{n}}\left(\frac{\partial \ln p\left(\mathbf{y}_{\text {rep}}| \overline{\boldsymbol{\theta}}^{VB}\left(\mathbf{y}_{\text {rep}}\right)\right)}{\partial \boldsymbol{\theta}^{\prime}}-\frac{\partial \ln p\left(\mathbf{y}_{\text {rep}}|\widehat{\boldsymbol{\theta}}_n\left(\mathbf{y}_{\text {rep}}\right)\right)}{\partial \boldsymbol{\theta}^{\prime}}\right)\right\|^2\right]\right)^{1 / 2} \\
		& \times\left(E\left[\left\|\sqrt{n}\left(\boldsymbol{\theta}_n^p-\overline{\boldsymbol{\theta}}^{VB}\left(\mathbf{y}_{rep}\right)\right)\right\|^2\right]\right)^{1 / 2},
	\end{aligned}
\end{equation}
where
$$
E\left[\left\|\sqrt{n}\left(\boldsymbol{\theta}_n^p-\overline{\boldsymbol{\theta}}^{VB}\left(\mathbf{y}_{\text {rep}}\right)\right)\right\|^2\right]<\infty
$$
by (\ref{expectation of term1 tr1n is not infinety}). For the first term in the right-hand side of (\ref{expectation of tr2n})
$$
\begin{aligned}
	& \frac{1}{\sqrt{n}}\left(\frac{\partial \ln p\left(\mathbf{y}_{\text {rep}}| \overline{\boldsymbol{\theta}}^{VB}\left(\mathbf{y}_{r e p}\right)\right)}{\partial \boldsymbol{\theta}^{\prime}}-\frac{\partial \ln p\left(\mathbf{y}_{\text {rep}}| \widehat{\boldsymbol{\theta}}_n\left(\mathbf{y}_{r e p}\right)\right)}{\partial \boldsymbol{\theta}^{\prime}}\right) \\
	= & \frac{1}{\sqrt{n}} \frac{\partial^2 \ln p\left(\mathbf{y}_{r e p}| \widehat{\boldsymbol{\theta}}_n^{\#}\left(\mathbf{y}_{r e p}\right)\right)}{\partial \boldsymbol{\theta} \partial \boldsymbol{\theta}^{\prime}}\left(\overline{\boldsymbol{\theta}}^{VB}\left(\mathbf{y}_{\text {rep}}\right)-\widehat{\boldsymbol{\theta}}_n\left(\mathbf{y}_{\text {rep}}\right)\right) \\
	= & \frac{1}{n} \frac{\partial^2 \ln p\left(\mathbf{y}_{\text {rep}}| \widehat{\boldsymbol{\theta}}_n^{\#}\left(\mathbf{y}_{\text {rep}}\right)\right)}{\partial \boldsymbol{\theta} \partial \boldsymbol{\theta}^{\prime}} \sqrt{n}\left(\overline{\boldsymbol{\theta}}^{VB}\left(\mathbf{y}_{rep}\right)-\widehat{\boldsymbol{\theta}}_n\left(\mathbf{y}_{r e p}\right)\right),
\end{aligned}
$$
where $\widehat{\boldsymbol{\theta}}_n^{\#}\left(\mathbf{y}_{\text {rep}}\right)$ lies between $\overline{\boldsymbol{\theta}}^{VB}\left(\mathbf{y}_{\text {rep}}\right)$ and $\widehat{\boldsymbol{\theta}}_n\left(\mathbf{y}_{\text {rep}}\right)$. Thus, we have
$$
\begin{aligned}
	& E\left[\left\|\frac{1}{\sqrt{n}}\left(\frac{\partial \ln p\left(\mathbf{y}_{rep}| \overline{\boldsymbol{\theta}}^{VB}\left(\mathbf{y}_{r e p}\right)\right)}{\partial \boldsymbol{\theta}^{\prime}}-\frac{\partial \ln p\left(\mathbf{y}_{rep}| \widehat{\boldsymbol{\theta}}_n\left(\mathbf{y}_{rep}\right)\right)}{\partial \boldsymbol{\theta}^{\prime}}\right)\right\|^2\right] \\
	& =E\left[\left\|\frac{1}{n} \frac{\partial^2 \ln p\left(\mathbf{y}_{\text {rep}}| \widehat{\boldsymbol{\theta}}_n^{\#}\left(\mathbf{y}_{\text {rep}}\right)\right)}{\partial \boldsymbol{\theta} \partial \boldsymbol{\theta}^{\prime}} \sqrt{n}\left(\overline{\boldsymbol{\theta}}^{VB}\left(\mathbf{y}_{\text {rep}}\right)-\widehat{\boldsymbol{\theta}}_n\left(\mathbf{y}_{\text {rep}}\right)\right)\right\|^2\right] \\
	& \leq E\left[\left\|\frac{1}{n} \frac{\partial^2 \ln p\left(\mathbf{y}_{\text {rep}}| \widehat{\boldsymbol{\theta}}_n^{\#}\left(\mathbf{y}_{\text {rep}}\right)\right)}{\partial \boldsymbol{\theta} \partial \boldsymbol{\theta}^{\prime}}\right\|^2\left\|\sqrt{n}\left(\overline{\boldsymbol{\theta}}^{VB}\left(\mathbf{y}_{\text {rep}}\right)-\widehat{\boldsymbol{\theta}}_n\left(\mathbf{y}_{\text {rep}}\right)\right)\right\|^2\right] \\
	& \leq\left(E\left[\left\|\frac{1}{n} \frac{\partial^2 \ln p\left(\mathbf{y}_{\text {rep}}| \widehat{\boldsymbol{\theta}}_n^{\#}\left(\mathbf{y}_{\text {rep}}\right)\right)}{\partial \boldsymbol{\theta} \partial \boldsymbol{\theta}^{\prime}}\right\|^4\right]\right)^{1 / 2}\left(E\left[\left\|\sqrt{n}\left(\overline{\boldsymbol{\theta}}^{VB}\left(\mathbf{y}_{\text {rep}}\right)-\widehat{\boldsymbol{\theta}}_n\left(\mathbf{y}_{\text {rep}}\right)\right)\right\|^4\right]\right)^{1 / 2}
\end{aligned}
$$
By Assumption 5 and (\ref{2 term in Minkowski inequality of CS term1 is not infinety}), we have
$$
E\left[\left\|\frac{1}{n} \frac{\partial^2 \ln p\left(\mathbf{y}_{\text {rep}}|\widehat{\boldsymbol{\theta}}_n^{\#}\left(\mathbf{y}_{\text {rep}}\right)\right)}{\partial \boldsymbol{\theta} \partial \boldsymbol{\theta}^{\prime}}\right\|^4\right]<\infty,
$$
and
$$
E\left[\left\|\sqrt{n}\left(\overline{\boldsymbol{\theta}}^{VB}\left(\mathbf{y}_{\text {rep}}\right)-\widehat{\boldsymbol{\theta}}_n\left(\mathbf{y}_{\text {rep}}\right)\right)\right\|^4\right]=O\left(n^{-1}\right).
$$
Hence, 
$$
E\left[\left\|\frac{1}{\sqrt{n}}\left(\frac{\partial \ln p\left(\mathbf{y}_{rep}|\overline{\boldsymbol{\theta}}^{VB}\left(\mathbf{y}_{rep}\right)\right)}{\partial \boldsymbol{\theta}^{\prime}}-\frac{\partial \ln p\left(\mathbf{y}_{rep}| \widehat{\boldsymbol{\theta}}_n\left(\mathbf{y}_{rep}\right)\right)}{\partial \boldsymbol{\theta}^{\prime}}\right)\right\|^2\right]=o(1) .
$$
So we get
\begin{equation}
	\label{tr_2n is small op1}
	\begin{aligned}
		& E\left\|R T_{2, n}\right\| \\
		& \leq \left(E\left[\left\|\frac{1}{\sqrt{n}}\left(\frac{\partial \ln p\left(\mathbf{y}_{\text {rep}}| \overline{\boldsymbol{\theta}}^{VB}\left(\mathbf{y}_{\text {rep}}\right)\right)}{\partial \boldsymbol{\theta}^{\prime}}-\frac{\partial \ln p\left(\mathbf{y}_{\text {rep}}| \widehat{\boldsymbol{\theta}}_n\left(\mathbf{y}_{\text {rep}}\right)\right)}{\partial \boldsymbol{\theta}^{\prime}}\right)\right\|^2\right]\right)^{1 / 2} \\
		& \times\left(E\left[\left\|\sqrt{n}\left(\boldsymbol{\theta}_n^p-\overline{\boldsymbol{\theta}}^{VB}\left(\mathbf{y}_{\text {rep}}\right)\right)\right\|^2\right]\right)^{1 / 2} \\
		& = o(1).
	\end{aligned}
\end{equation}

From (\ref{expectation of tr1n is small op1}) and (\ref{tr_2n is small op1}), it can be shown that
$$
E\left\|R T_n\right\| \leq E\left\|R T_{1, n}\right\|+E\left\|R T_{2, n}\right\|=o(1) .
$$
We can further get
$$
\begin{aligned}
	T_2 & =E_{\mathbf{y}} E_{\mathbf{y}_{\text {rep}}}\left(-2 \ln p\left(\mathbf{y}_{\text {rep}} | \boldsymbol{\theta}_n^p\right)\right)-E_{\mathbf{y}} E_{\mathbf{y}_{\text {rep}}}\left(-2 \ln p\left(\mathbf{y}_{\text {rep}}| \overline{\boldsymbol{\theta}}^{VB}\left(\mathbf{y}_{\text {rep}}\right)\right)\right)\\
	& = E_{\mathbf{y}} E_{\mathbf{y}_{\text {rep}}}\left[- \frac{\partial \ln p\left(\mathbf{y}_{\text {rep}} | \overline{\boldsymbol{\theta}}^{VB}\left(\mathbf{y}_{\text {rep}}\right)\right)}{\partial \boldsymbol{\theta}^{\prime}}\left(\overline{\boldsymbol{\theta}}^{VB}\left(\mathbf{y}_{\text {rep}}\right)-\boldsymbol{\theta}_n^p\right)\right]\\
	& + E_{\mathbf{y}} E_{\mathbf{y}_{\text {rep}}}\left[-\left(\overline{\boldsymbol{\theta}}^{V B}\left(\mathbf{y}_{\text {rep}}\right)-\boldsymbol{\theta}_n^p\right)^{\prime} \frac{\partial \ln p\left(\mathbf{y}_{\text {rep}} | \overline{\boldsymbol{\theta}}^{VB}\left(\mathbf{y}_{\text {rep}}\right)\right)}{\partial \boldsymbol{\theta} \partial \boldsymbol{\theta}^{\prime}}\left(\overline{\boldsymbol{\theta}}^{VB}\left(\mathbf{y}_{\text {rep}}\right)-\boldsymbol{\theta}_n^p\right)+R T_n\right] \\
	& =E_{\mathbf{y}_{\text {rep}}}\left[-\left(\overline{\boldsymbol{\theta}}^{VB}\left(\mathbf{y}_{\text {rep}}\right)-\boldsymbol{\theta}_n^p\right)^{\prime} \frac{\partial^2 \ln p\left(\mathbf{y}_{\text {rep}} | \overline{\boldsymbol{\theta}}^{VB}\left(\mathbf{y}_{\text {rep}}\right)\right)}{\partial \boldsymbol{\theta} \partial \boldsymbol{\theta}^{\prime}}\left(\overline{\boldsymbol{\theta}}^{VB}\left(\mathbf{y}_{\text {rep}}\right)-\boldsymbol{\theta}_n^p\right)\right]+o(1) \\
	& =E_{\mathbf{y}}\left[-\left(\overline{\boldsymbol{\theta}}^{VB}(\mathbf{y})-\boldsymbol{\theta}_n^p\right)^{\prime} \frac{\partial^2 \ln p\left(\mathbf{y}| \overline{\boldsymbol{\theta}}^{VB}(\mathbf{y})\right)}{\partial \boldsymbol{\theta} \partial \boldsymbol{\theta}^{\prime}}\left(\overline{\boldsymbol{\theta}}^{VB}(\mathbf{y})-\boldsymbol{\theta}_n^p\right)\right]+o(1) .
\end{aligned}
$$
Next we expand $\ln p\left(\mathbf{y}_{r e p}| \overline{\boldsymbol{\theta}}^{VB}(\mathbf{y})\right)$ at $\boldsymbol{\theta}_n^p$
$$
\begin{aligned}
	\ln p\left(\mathbf{y}_{rep}| \overline{\boldsymbol{\theta}}^{VB}(\mathbf{y})\right)= & \ln p\left(\mathbf{y}_{rep}| \boldsymbol{\theta}_n^p\right)+\frac{\partial \ln p\left(\mathbf{y}_{rep}| \boldsymbol{\theta}_n^p\right)}{\partial \boldsymbol{\theta}^{\prime}}\left(\overline{\boldsymbol{\theta}}^{VB}(\mathbf{y})-\boldsymbol{\theta}_n^p\right) \\
	& +\frac{1}{2}\left(\overline{\boldsymbol{\theta}}^{VB}(\mathbf{y})-\boldsymbol{\theta}_n^p\right)^{\prime} \frac{\partial^2 \ln p\left(\mathbf{y}_{rep}| \boldsymbol{\theta}_n^p\right)}{\partial \boldsymbol{\theta} \partial \boldsymbol{\theta}^{\prime}}\left(\overline{\boldsymbol{\theta}}^{VB}(\mathbf{y})-\boldsymbol{\theta}_n^p\right)+o_p(1).
\end{aligned}
$$
Substituting the above expansion into $T_3$, we have
$$
\begin{aligned}
	& T_3=E_{\mathbf{y}} E_{\mathbf{y}_{\text {rep}}}\left[-2 \ln p\left(\mathbf{y}_{\text {rep}} | \overline{\boldsymbol{\theta}}^{VB}(\mathbf{y})\right)\right]-E_{\mathbf{y}} E_{\mathbf{y}_{\text {rep}}}\left[-2 \ln p\left(\mathbf{y}_{\text {rep}} | \boldsymbol{\theta}_n^p\right)\right] \\
	& =E_{\mathbf{y}} E_{\mathbf{y}_{\text {rep}}}\left[\begin{array}{c}
		-2 \frac{\partial \ln p\left(\mathbf{y}_{\text {rep}}|\boldsymbol{\theta}_n^p\right)}{\partial \boldsymbol{\theta}^{\prime}}\left(\overline{\boldsymbol{\theta}}^{VB}(\mathbf{y})-\boldsymbol{\theta}_n^p\right)- \\
		\left(\overline{\boldsymbol{\theta}}^{VB}(\mathbf{y})-\boldsymbol{\theta}_n^p\right)^{\prime} \frac{\partial^2 \ln p\left(\mathbf{y}_{rep}| \boldsymbol{\theta}_n^p\right)}{\partial \boldsymbol{\theta} \partial \boldsymbol{\theta}^{\prime}}\left(\overline{\boldsymbol{\theta}}^{VB}(\mathbf{y})-\boldsymbol{\theta}_n^p\right)+o_p(1)
	\end{array}\right] \\
	& =E_{\mathbf{y}} E_{\mathbf{y}_{\text {rep}}}\left[-2 \frac{\partial \ln p\left(\mathbf{y}_{\text {rep}}| \boldsymbol{\theta}_n^p\right)}{\partial \boldsymbol{\theta}^{\prime}}\left(\overline{\boldsymbol{\theta}}^{VB}(\mathbf{y})-\boldsymbol{\theta}_n^p\right)\right] \\
	& +E_{\mathbf{y}} E_{\mathbf{y}_{\text {rep}}}\left[-\left(\overline{\boldsymbol{\theta}}^{VB}(\mathbf{y})-\boldsymbol{\theta}_n^p\right)^{\prime} \frac{\partial^2 \ln p\left(\mathbf{y}_{\text {rep}} | \boldsymbol{\theta}_n^p\right)}{\partial \boldsymbol{\theta} \partial \boldsymbol{\theta}^{\prime}}\left(\overline{\boldsymbol{\theta}}^{VB}(\mathbf{y})-\boldsymbol{\theta}_n^p\right)\right]+o(1) \\
	& =-2 E_{\mathbf{y}_{\text {rep}}}\left(\frac{\partial \ln p\left(\mathbf{y}_{\text {rep}} | \boldsymbol{\theta}_n^p\right)}{\partial \boldsymbol{\theta}^{\prime}}\right) E_{\mathbf{y}}\left[\left(\overline{\boldsymbol{\theta}}^{VB}(\mathbf{y})-\boldsymbol{\theta}_n^p\right)\right] \\
	& +E_{\mathbf{y}}\left[-\left(\overline{\boldsymbol{\theta}}^{VB}(\mathbf{y})-\boldsymbol{\theta}_n^p\right)^{\prime} E_{\mathbf{y}_{\text {rep}}}\left(\frac{\partial^2 \ln p\left(\mathbf{y}_{\text {rep}} | \boldsymbol{\theta}_n^p\right)}{\partial \boldsymbol{\theta} \partial \boldsymbol{\theta}^{\prime}}\right)\left(\overline{\boldsymbol{\theta}}^{VB}(\mathbf{y})-\boldsymbol{\theta}_n^p\right)\right]+o(1) \\
	& =E_{\mathbf{y}}\left[-\sqrt{n}\left(\overline{\boldsymbol{\theta}}^{VB}(\mathbf{y})-\boldsymbol{\theta}_n^p\right)^{\prime} E_{\mathbf{y}}\left(\frac{1}{n} \frac{\partial^2 \ln p\left(\mathbf{y} | \boldsymbol{\theta}_n^p\right)}{\partial \boldsymbol{\theta} \partial \boldsymbol{\theta}^{\prime}}\right) \sqrt{n}\left(\overline{\boldsymbol{\theta}}^{VB}(\mathbf{y})-\boldsymbol{\theta}_n^p\right)\right]+o(1),
\end{aligned}
$$
since
$$
\begin{aligned}
	E_{\mathbf{y}} E_{\mathbf{y}_{\text {rep}}}& \left[-2 \frac{\partial \ln p\left(\mathbf{y}_{\text {rep}} | \boldsymbol{\theta}_n^p\right)}{\partial \boldsymbol{\theta}^{\prime}}\left(\overline{\boldsymbol{\theta}}^{VB}(\mathbf{y})-\boldsymbol{\theta}_n^p\right)\right] \\ 
	= E_{\mathbf{y}_{\text {rep}}} & \left[-2 \frac{\partial \ln p\left(\mathbf{y}_{\text {rep}} | \boldsymbol{\theta}_n^p\right)}{\partial \boldsymbol{\theta}^{\prime}}\right] E_{\mathbf{y}}\left[\left(\overline{\boldsymbol{\theta}}^{VB}(\mathbf{y})-\boldsymbol{\theta}_n^p\right)\right]=0
\end{aligned}
$$
by (\ref{var score f}) and the dominated convergence theorem.

We can rewrite $T_2$ as
$$
\begin{aligned}
	T_2 & =E_{\mathbf{y}}\left[-\left(\overline{\boldsymbol{\theta}}^{VB}(\mathbf{y})-\boldsymbol{\theta}_n^p\right)^{\prime} \frac{\partial^2 \ln p\left(\mathbf{y} | \overline{\boldsymbol{\theta}}^{VB}(\mathbf{y})\right)}{\partial \boldsymbol{\theta} \partial \boldsymbol{\theta}^{\prime}}\left(\overline{\boldsymbol{\theta}}^{VB}(\mathbf{y})-\boldsymbol{\theta}_n^p\right)\right]+o(1) \\
	& =E_{\mathbf{y}}\left[-\sqrt{n}\left(\overline{\boldsymbol{\theta}}^{VB}(\mathbf{y})-\boldsymbol{\theta}_n^p\right)^{\prime} \frac{1}{n} E_{\mathbf{y}}\left(\frac{\partial^2 \ln p\left(\mathbf{y} | \boldsymbol{\theta}_n^p\right)}{\partial \boldsymbol{\theta} \partial \boldsymbol{\theta}^{\prime}}\right) \sqrt{n}\left(\overline{\boldsymbol{\theta}}^{VB}(\mathbf{y})-\boldsymbol{\theta}_n^p\right)\right] \\
	& +E_{\mathbf{y}}\left[\begin{array}{c}
		-\sqrt{n}\left(\overline{\boldsymbol{\theta}}^{VB}(\mathbf{y})-\boldsymbol{\theta}_n^p\right)^{\prime}\left(\frac{1}{n} \frac{\partial^2 \ln p\left(\mathbf{y} | \overline{\boldsymbol{\theta}}^{VB}(\mathbf{y})\right)}{\partial \boldsymbol{\theta} \partial \boldsymbol{\theta}^{\prime}}-E_{\mathbf{y}}\left(\frac{1}{n} \frac{\partial^2 \ln p\left(\mathbf{y} | \boldsymbol{\theta}_n^p\right)}{\partial \boldsymbol{\theta} \partial \boldsymbol{\theta}^{\prime}}\right)\right) \\
		\times \sqrt{n}\left(\overline{\boldsymbol{\theta}}^{VB}(\mathbf{y})-\boldsymbol{\theta}_n^p\right)
	\end{array}\right] + o\left(1\right)	
\end{aligned}
$$
where 
\begin{equation}
	\label{t2 scaling inequality}
	\begin{aligned}
		& E_{\mathbf{y}}\left[\begin{array}{c}
			\left.-\sqrt{n}\left(\overline{\boldsymbol{\theta}}^{VB}(\mathbf{y})-\boldsymbol{\theta}_n^p\right)^{\prime}\left(\frac{1}{n} \frac{\partial^2 \ln p\left(\mathbf{y}| \overline{\boldsymbol{\theta}}^{VB}(\mathbf{y})\right)}{\partial \boldsymbol{\theta} \partial \boldsymbol{\theta}^{\prime}}-E_{\mathbf{y}}\left(\frac{1}{n} \frac{\partial^2 \ln p\left(\mathbf{y} | \boldsymbol{\theta}_n^p\right)}{\partial \boldsymbol{\theta} \partial \boldsymbol{\theta}^{\prime}}\right)\right)\right] \\
			\times \sqrt{n}\left(\overline{\boldsymbol{\theta}}^{VB}(\mathbf{y})-\boldsymbol{\theta}_n^p\right)
		\end{array}\right] \\
		\leq & E_{\mathbf{y}}\left[\left\|\sqrt{n}\left(\overline{\boldsymbol{\theta}}^{VB}(\mathbf{y})-\boldsymbol{\theta}_n^p\right)\right\|^2\left\|\frac{1}{n} \frac{\partial^2 \ln p\left(\mathbf{y}| \overline{\boldsymbol{\theta}}^{VB}(\mathbf{y})\right)}{\partial \boldsymbol{\theta} \partial \boldsymbol{\theta}^{\prime}}-E_{\mathbf{y}}\left(\frac{1}{n} \frac{\partial^2 \ln p\left(\mathbf{y}| \boldsymbol{\theta}_n^p\right)}{\partial \boldsymbol{\theta} \partial \boldsymbol{\theta}^{\prime}}\right)\right\|\right] \\
		\leq & \left(E_{\mathbf{y}}\left[\left\|\sqrt{n}\left(\overline{\boldsymbol{\theta}}^{VB}(\mathbf{y})-\boldsymbol{\theta}_n^p\right)\right\|^4\right]\right)^{1 / 2} \\
		& \times\left(E_{\mathbf{y}}\left[\left\|\frac{1}{n} \frac{\partial^2 \ln p\left(\mathbf{y}| \overline{\boldsymbol{\theta}}^{VB}(\mathbf{y})\right)}{\partial \boldsymbol{\theta} \partial \boldsymbol{\theta}^{\prime}}-E_{\mathbf{y}}\left(\frac{1}{n} \frac{\partial^2 \ln p\left(\mathbf{y} | \boldsymbol{\theta}_n^p\right)}{\partial \boldsymbol{\theta} \partial \boldsymbol{\theta}^{\prime}}\right)\right\|^2\right]\right)^{1 / 2} .
	\end{aligned}
\end{equation}
In (\ref{t2 scaling inequality}), we have
\begin{equation}
	\label{t2 scaling inequality term1}
	\begin{aligned}
		& E_{\mathbf{y}}\left[\left\|\frac{1}{n} \frac{\partial^2 \ln p\left(\mathbf{y}| \overline{\boldsymbol{\theta}}^{VB}(\mathbf{y})\right)}{\partial \boldsymbol{\theta} \partial \boldsymbol{\theta}^{\prime}}-E_{\mathbf{y}}\left(\frac{1}{n} \frac{\partial^2 \ln p\left(\mathbf{y} | \boldsymbol{\theta}_n^p\right)}{\partial \boldsymbol{\theta} \partial \boldsymbol{\theta}^{\prime}}\right)\right\|^2\right] \\
		= & E_{\mathbf{y}}\left[\left\|\frac{1}{n} \frac{\partial^2 \ln p\left(\mathbf{y} | \overline{\boldsymbol{\theta}}^{VB}(\mathbf{y})\right)}{\partial \boldsymbol{\theta} \partial \boldsymbol{\theta}^{\prime}}-\overline{\mathbf{H}}_n\left(\boldsymbol{\theta}_n^p\right)+\overline{\mathbf{H}}_n\left(\boldsymbol{\theta}_n^p\right)-\mathbf{H}_n\right\|^2\right] \\
		\leq & {\left[E_{\mathbf{y}}\left[\left\|\frac{1}{n} \frac{\partial^2 \ln p\left(\mathbf{y} | \overline{\boldsymbol{\theta}}^{VB}(\mathbf{y})\right)}{\partial \boldsymbol{\theta} \partial \boldsymbol{\theta}^{\prime}}-\overline{\mathbf{H}}_n\left(\boldsymbol{\theta}_n^p\right)\right\|^2\right]^{1 / 2}+\left[E_{\mathbf{y}}\left[\left\|\overline{\mathbf{H}}_n\left(\boldsymbol{\theta}_n^p\right)-\mathbf{H}_n\right\|^2\right]\right]^{1 / 2}\right]^2 }
	\end{aligned}
\end{equation}

The first term of (\ref{t2 scaling inequality term1}) can be written as
$$
\begin{aligned}
vec & \left(\frac{1}{n} \frac{\partial^2 \ln p\left(\mathbf{y} | \overline{\boldsymbol{\theta}}^{VB}(\mathbf{y})\right)}{\partial \boldsymbol{\theta} \partial \boldsymbol{\theta}^{\prime}}-\overline{\mathbf{H}}_n\left(\boldsymbol{\theta}_n^p\right)\right) \\
= & vec \left(\overline{\mathbf{H}}_n\left(\overline{\boldsymbol{\theta}}^{VB}(\mathbf{y})\right)\right)-\operatorname{vec}\left(\overline{\mathbf{H}}_n\left(\boldsymbol{\theta}_n^p\right)\right)=\frac{1}{n} \sum_{t=1}^n \nabla^3 l_t\left(\widetilde{\boldsymbol{\theta}}_n^{* *}(\mathbf{y})\right)\left(\overline{\boldsymbol{\theta}}^{VB}(\mathbf{y})-\boldsymbol{\theta}_n^p\right) \\
= & \frac{1}{\sqrt{n}} \frac{1}{n} \sum_{t=1}^n \nabla^3 l_t\left(\widetilde{\boldsymbol{\theta}}_n^{* *}(\mathbf{y})\right) \sqrt{n}\left(\overline{\boldsymbol{\theta}}^{VB}(\mathbf{y})-\boldsymbol{\theta}_n^p\right)
\end{aligned}
$$
by vectorization and the Taylor expansion, where $\tilde{\boldsymbol{\theta}}_n^{**}(\mathbf{y})$ lies between $\overline{\boldsymbol{\theta}}^{VB}(\mathbf{y})$ and $\boldsymbol{\theta}_n^p$. Thus,
\begin{equation}
	\label{t2 scaling inequality term2}
	\begin{aligned}
		& E_{\mathbf{y}}\left[\left\|\frac{1}{n} \frac{\partial^2 \ln p\left(\mathbf{y} | \overline{\boldsymbol{\theta}}^{VB}(\mathbf{y})\right)}{\partial \boldsymbol{\theta} \partial \boldsymbol{\theta}^{\prime}}-\overline{\mathbf{H}}_n\left(\boldsymbol{\theta}_n^p\right)\right\|^2\right] \\
		\leq & \frac{1}{n} E_{\mathbf{y}}\left[\left\|\frac{1}{n} \sum_{t=1}^n \nabla^3 l_t\left(\widetilde{\boldsymbol{\theta}}_n^{* *}(\mathbf{y})\right)\right\|^2\left\|\sqrt{n}\left(\overline{\boldsymbol{\theta}}^{VB}(\mathbf{y})-\boldsymbol{\theta}_n^p\right)\right\|^2\right] \\
		\leq & \frac{1}{n}\left(E_{\mathbf{y}}\left[\left\|\frac{1}{n} \sum_{t=1}^n \nabla^3 l_t\left(\widetilde{\boldsymbol{\theta}}_n^{* *}(\mathbf{y})\right)\right\|^4\right]\right)^{1 / 2}\left(E_{\mathbf{y}}\left[\left\|\sqrt{n}\left(\overline{\boldsymbol{\theta}}^{VB}(\mathbf{y})-\boldsymbol{\theta}_n^p\right)\right\|^4\right]\right)^{1 / 2} \\
		= & O\left(n^{-1}\right)
	\end{aligned}
\end{equation}
by Assumption 5 and (\ref{expectation of term1 tr1n is not infinety}). The second term of (\ref{t2 scaling inequality term1}) can be written as
\begin{equation}
E_{\mathbf{y}}\left[\left\|\overline{\mathbf{H}}_n\left(\boldsymbol{\theta}_n^p\right)-\mathbf{H}_n\right\|^2\right] \leq \frac{1}{n} E_{\mathbf{y}}\left[\left\|\sqrt{n}\left(\overline{\mathbf{H}}_n\left(\boldsymbol{\theta}_n^p\right)-\mathbf{H}_n\right)\right\|^2\right]=O\left(n^{-1}\right)
\end{equation}
by Assumption 1-8. From (\ref{t2 scaling inequality term1}) and (\ref{t2 scaling inequality term2}) 
$$
E_{\mathbf{y}}\left[\left\|\frac{1}{n} \frac{\partial^2 \ln p\left(\mathbf{y} | \overline{\boldsymbol{\theta}}^{VB}(\mathbf{y})\right)}{\partial \boldsymbol{\theta} \partial \boldsymbol{\theta}^{\prime}}-E_{\mathbf{y}}\left(\frac{1}{n} \frac{\partial^2 \ln p\left(\mathbf{y} | \boldsymbol{\theta}_n^p\right)}{\partial \boldsymbol{\theta} \partial \boldsymbol{\theta}^{\prime}}\right)\right\|^2\right]=o(1)
$$
Thus, we have
\begin{equation}
	E_{\mathbf{y}}\left[\begin{array}{c}
		-\sqrt{n}\left(\overline{\boldsymbol{\theta}}^{VB}(\mathbf{y})-\boldsymbol{\theta}_n^p\right)^{\prime}\left(\frac{1}{n} \frac{\partial^2 \ln p\left(\mathbf{y} | \overline{\boldsymbol{\theta}}^{VB}(\mathbf{y})\right)}{\partial \boldsymbol{\theta} \partial \boldsymbol{\theta}^{\prime}}-E_{\mathbf{y}}\left(\frac{1}{n} \frac{\partial^2 \ln p\left(\mathbf{y} | \boldsymbol{\theta}_n^p\right)}{\partial \boldsymbol{\theta} \partial \boldsymbol{\theta}^{\prime}}\right)\right) \\
		\times \sqrt{n}\left(\overline{\boldsymbol{\theta}}^{VB}(\mathbf{y})-\boldsymbol{\theta}_n^p\right)
	\end{array}\right] = o\left(1\right)
\end{equation}
We can further rewrite $T_2$ as
$$
\begin{aligned}
	T_2 & =E_{\mathbf{y}}\left[-\left(\overline{\boldsymbol{\theta}}^{VB}(\mathbf{y})-\boldsymbol{\theta}_n^p\right)^{\prime} \frac{\partial^2 \ln p\left(\mathbf{y} | \overline{\boldsymbol{\theta}}^{VB}(\mathbf{y})\right)}{\partial \boldsymbol{\theta} \partial \boldsymbol{\theta}^{\prime}}\left(\overline{\boldsymbol{\theta}}^{VB}(\mathbf{y})-\boldsymbol{\theta}_n^p\right)\right]+o(1) \\
	& =E_{\mathbf{y}}\left[-\sqrt{n}\left(\overline{\boldsymbol{\theta}}^{VB}(\mathbf{y})-\boldsymbol{\theta}_n^p\right)^{\prime} \frac{1}{n} E_{\mathbf{y}}\left(\frac{\partial^2 \ln p\left(\mathbf{y} | \boldsymbol{\theta}_n^p\right)}{\partial \boldsymbol{\theta} \partial \boldsymbol{\theta}^{\prime}}\right) \sqrt{n}\left(\overline{\boldsymbol{\theta}}^{VB}(\mathbf{y})-\boldsymbol{\theta}_n^p\right)\right]+o(1) \\
	& =T_3+o(1) .
\end{aligned}
$$
Hence, we only need to analyze $T_3$. Note that
\begin{equation}
	\label{t3 revised}
	\begin{aligned}
		T_3 & =E_{\mathbf{y}}\left[-\sqrt{n}\left(\overline{\boldsymbol{\theta}}^{VB}(\mathbf{y})-\boldsymbol{\theta}_n^p\right)^{\prime} E_{\mathbf{y}}\left(-\frac{1}{n} \frac{\partial^2 \ln p\left(\mathbf{y} | \boldsymbol{\theta}_n^p\right)}{\partial \boldsymbol{\theta} \partial \boldsymbol{\theta}^{\prime}}\right) \sqrt{n}\left(\overline{\boldsymbol{\theta}}^{VB}(\mathbf{y})-\boldsymbol{\theta}_n^p\right)\right]+o(1) \\
		= & E_{\mathbf{y}}\left[\sqrt{n}\left(\overline{\boldsymbol{\theta}}^{VB}(\mathbf{y})-\boldsymbol{\theta}_n^p\right)^{\prime}\left(-\mathbf{H}_n\right) \sqrt{n}\left(\overline{\boldsymbol{\theta}}^{VB}(\mathbf{y})-\boldsymbol{\theta}_n^p\right)\right]+o(1) \\
		= & E_{\mathbf{y}}\left[\left(\mathbf{C}_n^{-1 / 2} \sqrt{n}\left(\overline{\boldsymbol{\theta}}^{VB}(\mathbf{y})-\boldsymbol{\theta}_n^p\right)\right)^{\prime} \mathbf{C}_n^{1 / 2}\left(-\mathbf{H}_n\right) \mathbf{C}_n^{1 / 2} \mathbf{C}_n^{-1 / 2} \sqrt{n}\left(\overline{\boldsymbol{\theta}}^{VB}(\mathbf{y})-\boldsymbol{\theta}_n^p\right)\right]+o(1) \\
		= & E_{\mathbf{y}}\left\{\boldsymbol{t r}\left[\left(-\mathbf{H}_n\right) \mathbf{C}_n^{1 / 2} \mathbf{C}_n^{-1 / 2} \sqrt{n}\left(\overline{\boldsymbol{\theta}}^{VB}(\mathbf{y})-\boldsymbol{\theta}_n^p\right) \sqrt{n}\left(\overline{\boldsymbol{\theta}}^{VB}(\mathbf{y})-\boldsymbol{\theta}_n^p\right)^{\prime} \mathbf{C}_n^{-1 / 2} \mathbf{C}_n^{1 / 2}\right]\right\}+o(1) \\
		= & \mathbf{tr} \left\{\left(-\mathbf{H}_n\right) \mathbf{C}_n^{1 / 2} E_{\mathbf{y}}\left[\mathbf{C}_n^{-1 / 2} \sqrt{n}\left(\overline{\boldsymbol{\theta}}^{VB}(\mathbf{y})-\boldsymbol{\theta}_n^p\right) \sqrt{n}\left(\overline{\boldsymbol{\theta}}^{VB}(\mathbf{y})-\boldsymbol{\theta}_n^p\right)^{\prime} \mathbf{C}_n^{-1 / 2}\right] \mathbf{C}_n^{1 / 2}\right\}+o(1)
	\end{aligned}
\end{equation}

In (\ref{t3 revised}), we have
$$
\begin{aligned}
	& E_{\mathbf{y}}\left[\mathbf{C}_n^{-1 / 2} \sqrt{n}\left(\overline{\boldsymbol{\theta}}^{VB}(\mathbf{y})-\boldsymbol{\theta}_n^p\right) \sqrt{n}\left(\overline{\boldsymbol{\theta}}^{VB}(\mathbf{y})-\boldsymbol{\theta}_n^p\right)^{\prime} \mathbf{C}_n^{-1 / 2}\right] \\
	= & \mathbf{C}_n^{-1 / 2} E_{\mathbf{y}}\left[ \sqrt{n}\left(\overline{\boldsymbol{\theta}}^{VB}(\mathbf{y})-\boldsymbol{\theta}_n^p\right) \sqrt{n}\left(\overline{\boldsymbol{\theta}}^{VB}(\mathbf{y})-\boldsymbol{\theta}_n^p\right)^{\prime} \right] \mathbf{C}_n^{-1 / 2}
\end{aligned}
$$
where
\begin{equation}
	\label{revised t3 subterm}
	\begin{aligned}
			& E_{\mathbf{y}}\left[\sqrt{n}\left(\overline{\boldsymbol{\theta}}^{VB}(\mathbf{y})-\boldsymbol{\theta}_n^p\right) \sqrt{n}\left(\overline{\boldsymbol{\theta}}^{VB}(\mathbf{y})-\boldsymbol{\theta}_n^p\right)^{\prime}\right] \\
			= &  E_{\mathbf{y}}\left[\sqrt{n}\left(\overline{\boldsymbol{\theta}}^{VB}(\mathbf{y})-\widehat{\boldsymbol{\theta}}_n(\mathbf{y})+\widehat{\boldsymbol{\theta}}_n(\mathbf{y})-\boldsymbol{\theta}_n^p\right) \sqrt{n}\left(\overline{\boldsymbol{\theta}}^{VB}(\mathbf{y})-\widehat{\boldsymbol{\theta}}_n(\mathbf{y})+\widehat{\boldsymbol{\theta}}_n(\mathbf{y})-\boldsymbol{\theta}_n^p\right)^{\prime}\right] \\
			= & E_{\mathbf{y}}\left[\sqrt{n}\left(\widehat{\boldsymbol{\theta}}_n(\mathbf{y})-\boldsymbol{\theta}_n^p\right) \sqrt{n}\left(\widehat{\boldsymbol{\theta}}_n(\mathbf{y})-\boldsymbol{\theta}_n^p\right)^{\prime}\right]+E_{\mathbf{y}}\left[\sqrt{n}\left(\overline{\boldsymbol{\theta}}^{VB}(\mathbf{y})-\widehat{\boldsymbol{\theta}}_n(\mathbf{y})\right) \sqrt{n}\left(\widehat{\boldsymbol{\theta}}_n(\mathbf{y})-\boldsymbol{\theta}_n^p\right)^{\prime}\right] \\
			& +E_{\mathbf{y}}\left[\sqrt{n}\left(\widehat{\boldsymbol{\theta}}_n(\mathbf{y})-\boldsymbol{\theta}_n^p\right) \sqrt{n}\left(\overline{\boldsymbol{\theta}}^{VB}(\mathbf{y})-\widehat{\boldsymbol{\theta}}_n(\mathbf{y})\right)^{\prime}\right] \\
			& +E_{\mathbf{y}}\left[\sqrt{n}\left(\overline{\boldsymbol{\theta}}^{VB}(\mathbf{y})-\widehat{\boldsymbol{\theta}}_n(\mathbf{y})\right) \sqrt{n}\left(\overline{\boldsymbol{\theta}}^{VB}(\mathbf{y})-\widehat{\boldsymbol{\theta}}_n(\mathbf{y})\right)^{\prime}\right] .
	\end{aligned}
\end{equation}
In (\ref{revised t3 subterm}), it can be shown that the last three terms are all $o\left(1\right)$ because of (\ref{1 term in Minkowski inequality of CS term1 is not infinety}) and (\ref{2 term in Minkowski inequality of CS term1 is not infinety}). For the first term, we know that
$$
E_{\mathbf{y}}\left[\sqrt{n}\left(\widehat{\boldsymbol{\theta}}_n(\mathbf{y})-\boldsymbol{\theta}_n^p\right) \sqrt{n}\left(\widehat{\boldsymbol{\theta}}_n(\mathbf{y})-\boldsymbol{\theta}_n^p\right)^{\prime}\right]=\mathbf{H}_n^{-1} \mathbf{B}_n \mathbf{H}_n^{-1}+o(1)=\mathbf{C}_n+o(1)
$$
by \citet{li2024deviance}. Hence, it can be shown that
$$
\begin{aligned}
	T_3 & =\mathbf{tr} \left\{\left(-\mathbf{H}_n\right) \mathbf{C}_n^{1 / 2} \mathbf{C}_n^{-1 / 2} E_{\mathbf{y}}\left[\sqrt{n}\left(\overline{\boldsymbol{\theta}}^{VB}(\mathbf{y})-\boldsymbol{\theta}_n^p\right) \sqrt{n}\left(\overline{\boldsymbol{\theta}}^{VB}(\mathbf{y})-\boldsymbol{\theta}_n^p\right)^{\prime}\right] \mathbf{C}_n^{-1 / 2} \mathbf{C}_n^{1 / 2}\right\}+o(1) \\
	& =\mathbf{tr} \left\{\left(-\mathbf{H}_n\right) \mathbf{C}_n^{1 / 2} \mathbf{C}_n^{-1 / 2} \mathbf{C}_n \mathbf{C}_n^{-1 / 2} \mathbf{C}_n^{1 / 2}\right\}+o(1) \\
	& =\mathbf{tr} \left(\left(-\mathbf{H}_n\right) \mathbf{C}_n\right)+o(1) \\
	& =\mathbf{tr} \left(\left(-\mathbf{H}_n\right)\left(-\mathbf{H}_n\right)^{-1} \mathbf{B}_n\left(-\mathbf{H}_n\right)^{-1}\right)+o(1) \\
	& =\mathbf{tr} \left[\mathbf{B}_n\left(-\mathbf{H}_n\right)^{-1}\right]+o(1).
\end{aligned}
$$
and 
\begin{equation}
	\label{vtic prepared}
	\begin{aligned}
		& E_{\mathbf{y}}\left[E_{\mathbf{y}_{\text {rep}}}\left(-2 \ln p\left(\mathbf{y}_{\text {rep}} | \overline{\boldsymbol{\theta}}^{VB}(\mathbf{y})\right)\right)\right] \\
		= & E_{\mathbf{y}}\left[E_{\mathbf{y}_{\text {rep}}}\left(-2 \ln p\left(\mathbf{y}_{\text {rep}} | \overline{\boldsymbol{\theta}}^{VB}\left(\mathbf{y}_{\text {rep}}\right)\right)+ T_2 + T_3 \right)\right] \\
		= & E_{\mathbf{y}}\left[E_{\mathbf{y}_{\text {rep}}}\left(-2 \ln p\left(\mathbf{y}_{\text {rep}} | \overline{\boldsymbol{\theta}}^{VB}\left(\mathbf{y}_{\text {rep}}\right)\right)\right)\right]+2 \mathbf{t r}\left[\mathbf{B}_n\left(-\mathbf{H}_n\right)^{-1}\right]+o(1) \\
		= & E_{\mathbf{y}}\left[E_{\mathbf{y}}\left(-2 \ln p\left(\mathbf{y} | \overline{\boldsymbol{\theta}}^{VB}(\mathbf{y})\right)\right)\right]+2 \mathbf{t r}\left[\mathbf{B}_n\left(-\mathbf{H}_n\right)^{-1}\right]+o(1) \\
		= & E_{\mathbf{y}}\left[-2 \ln p\left(\mathbf{y} | \overline{\boldsymbol{\theta}}^{VB}(\mathbf{y})\right)\right] - 2 \mathbf{t r}\left[\mathbf{B}_n \mathbf{H}_n^{-1}\right]+o(1).
	\end{aligned}
\end{equation}

Note that in (\ref{vtic prepared}), we have tranformed $T_1$ as
$$
\begin{aligned}
	T_1 & = E_{\mathbf{y}}\left[E_{\mathbf{y}_{\text {rep}}}\left(-2 \ln p\left(\mathbf{y}_{\text {rep}} | \overline{\boldsymbol{\theta}}^{VB}\left(\mathbf{y}_{\text {rep}}\right)\right)\right)\right] \\
	& = E_{\mathbf{y}}\left[E_{\mathbf{y}}\left(-2 \ln p\left(\mathbf{y} | \overline{\boldsymbol{\theta}}^{VB}\left(\mathbf{y}\right)\right)\right)\right]\\
	& = E_{\mathbf{y}}\left[-2 \ln p\left(\mathbf{y} | \overline{\boldsymbol{\theta}}^{VB}\left(\mathbf{y}\right)\right)\right],
\end{aligned}
$$
The last step to prove Theroem 3.1
%\protect\ref{riskvtic}
is  to make a slight chage on $T_1$ 
$$
\begin{aligned}
	T_1 & = E_{\mathbf{y}}\left[-2 \ln p\left(\mathbf{y} | \overline{\boldsymbol{\theta}}^{VB}\left(\mathbf{y}\right)\right)\right] \\
	& = T_{11} + T_{12},
\end{aligned}
$$
where
$$
\begin{aligned}
T_{11} &= E_{\mathbf{y}}  \left[-2 \ln p\left(\mathbf{y} | \widehat{\boldsymbol{\theta}}_n\left(\mathbf{y}\right)\right)\right] \\
T_{22} &= E_{\mathbf{y}} \left[ \left(-2 \ln p\left(\mathbf{y} | \overline{\boldsymbol{\theta}}^{VB}\left(\mathbf{y}\right)\right)\right) -\left(-2 \ln p\left(\mathbf{y} | \widehat{\boldsymbol{\theta}}_n\left(\mathbf{y}\right)\right)\right) \right],
\end{aligned}
$$
where we expand the term in $T_{22}$ at $\widehat{\boldsymbol{\theta}}_n$ 
$$
\begin{aligned}
	& \ln p\left(\mathbf{y} | \overline{\boldsymbol{\theta}}^{VB}\left(\mathbf{y}\right)\right) - \ln p\left(\mathbf{y} | \widehat{\boldsymbol{\theta}}_n\left(\mathbf{y}\right)\right) \\
	= & \frac{\partial \ln p\left(\mathbf{y}|\boldsymbol{\theta}^{\#\#}_n\left(\mathbf{y}\right)\right)}{\partial \boldsymbol{\theta}^{\prime}}\left(\overline{\boldsymbol{\theta}}^{VB}\left(\mathbf{y}\right)-\widehat{\boldsymbol{\theta}}_n\left(\mathbf{y}\right)\right),
\end{aligned}
$$
where $\boldsymbol{\theta}^{\#\#}_n$ lies between $\overline{\boldsymbol{\theta}}^{VB}\left(\mathbf{y}\right)$ and $\widehat{\boldsymbol{\theta}}_n\left(\mathbf{y}\right)$.
From (\ref{vb consistency mle}), and Assumption 5, we have 
$$
\begin{aligned}
	& \left(-2 \ln p\left(\mathbf{y} | \overline{\boldsymbol{\theta}}^{VB}\left(\mathbf{y}\right)\right)\right) -\left(-2 \ln p\left(\mathbf{y} | \widehat{\boldsymbol{\theta}}_n\left(\mathbf{y}\right)\right)\right) \\
	= & O_p\left(1\right) \times O_p\left(n^{-3/4}\right) = O_p\left(n^{-3/4}\right)\\
	= & o_p\left(1\right),
\end{aligned}
$$
thus we have
\begin{equation}
	\label{t1 term finally}
	\begin{aligned}
		T_1 & = E_{\mathbf{y}}\left[-2 \ln p\left(\mathbf{y} | \overline{\boldsymbol{\theta}}^{VB}\left(\mathbf{y}\right)\right)\right] = T_{11} + T_{12} \\
		& = E_{\mathbf{y}} \left[-2 \ln p\left(\mathbf{y} |\widehat{\boldsymbol{\theta}}_n\left(\mathbf{y}\right)\right) + o_p\left(1\right) \right] \\
		& = E_{\mathbf{y}} \left[-2 \ln p\left(\mathbf{y} |\widehat{\boldsymbol{\theta}}_n\left(\mathbf{y}\right)\right) \right] + o\left(1\right).
	\end{aligned}
\end{equation}

With (\ref{vtic prepared}) and (\ref{t1 term finally}), we have 
\begin{equation}
	\begin{aligned}
		& E_{\mathbf{y}}\left[E_{\mathbf{y}_{\text {rep}}}\left(-2 \ln p\left(\mathbf{y}_{\text {rep}} | \overline{\boldsymbol{\theta}}^{VB}(\mathbf{y})\right)\right)\right] \\
		= & E_{\mathbf{y}}\left[-2 \ln p\left(\mathbf{y} | \overline{\boldsymbol{\theta}}^{VB}(\mathbf{y})\right)\right] - 2 \mathbf{t r}\left[\mathbf{B}_n \mathbf{H}_n^{-1}\right]+o(1) \\
        = & E_{\mathbf{y}} \left[-2 \ln p\left(\mathbf{y} |\widehat{\boldsymbol{\theta}}_n\left(\mathbf{y}\right)\right) \right] - 2 \mathbf{t r}\left[\mathbf{B}_n \mathbf{H}_n^{-1}\right]+o(1)
	\end{aligned}
\end{equation}
Therefore $-2 \ln p\left(\mathbf{y} |\widehat{\boldsymbol{\theta}}_n\left(\mathbf{y}\right)\right) - 2 \mathbf{t r}\left[\mathbf{B}_n \mathbf{H}_n^{-1}\right]$ is an unbiased estimator of $$E_{\mathbf{y}}\left[E_{\mathbf{y}_{\text {rep}}}\left(-2 \ln p\left(\mathbf{y}_{\text {rep}} | \overline{\boldsymbol{\theta}}^{VB}(\mathbf{y})\right)\right)\right]$$ asymptotically.

\subsubsection{Proof of Theorem 3.2}
%\protect\ref{riskvpic}}

We are now in the position to prove Theorem 3.2.
%\ref{riskvpic}.
Under Assumptions
1-8, it can be shown that, 
\begin{equation*}
E_{\mathbf{y}}E_{\mathbf{y}_{rep}}\left( -2\ln p\left( \mathbf{y}_{rep}|%
\mathbf{y}\right) \right) =E_{\mathbf{y}}\left[ -2\ln p\left( \mathbf{y}|%
\overleftrightarrow{{\mbox{\boldmath${\theta}$}}}_{n}\right) +\left( 1+\ln
2\right) P\right] +o\left( 1\right). 
\end{equation*}

%\begin{proof}

By the Laplace approximation (Tierney et al., 1989 and Kass et al., 1990)
and Lemma \ref{lemmaclts}, we have%
\begin{eqnarray*}
&&p^{VB}\left( \mathbf{y}_{rep}|\mathbf{y}\right) \\
&=&\int p\left( \mathbf{y}_{rep}|\mbox{\boldmath${\theta}$}\right)
p^{VB}\left( \mbox{\boldmath${\theta}$}|\mathbf{y}\right) d%
\mbox{\boldmath${\theta}$} \\
&=&\int p\left( \mathbf{y}_{rep}|\mbox{\boldmath${\theta}$}\right)
p^{VBN}\left( \mbox{\boldmath${\theta}$}|\mathbf{y}\right) d%
\mbox{\boldmath${\theta}$}+\int p\left( \mathbf{y}_{rep}|\mbox{\boldmath${%
\theta}$}\right) \left( p^{VB}\left( \mbox{\boldmath${\theta}$}|\mathbf{y}%
\right) -p^{VBN}\left( \mbox{\boldmath${\theta}$}|\mathbf{y}\right) \right) d%
\mbox{\boldmath${\theta}$} \\
&=&\int p\left( \mathbf{y}_{rep}|\mbox{\boldmath${\theta}$}\right)
p^{VBN}\left( \mbox{\boldmath${\theta}$}|\mathbf{y}\right) d%
\mbox{\boldmath${\theta}$}\left( 1+\frac{\int p\left( \mathbf{y}_{rep}|%
\mbox{\boldmath${\theta}$}\right) \left( p^{VB}\left( \mbox{\boldmath${%
\theta}$}|\mathbf{y}\right) -p^{VBN}\left( \mbox{\boldmath${\theta}$}|%
\mathbf{y}\right) \right) d\mbox{\boldmath${\theta}$}}{\int p\left( \mathbf{y%
}_{rep}|\mbox{\boldmath${\theta}$}\right) p^{VBN}\left( \mbox{\boldmath${%
\theta}$}|\mathbf{y}\right) d\mbox{\boldmath${\theta}$}}\right)
\end{eqnarray*}%
Note that%
\begin{eqnarray*}
&&\frac{\int p\left( \mathbf{y}_{rep}|\mbox{\boldmath${\theta}$}\right)
\left( p^{VB}\left( \mbox{\boldmath${\theta}$}|\mathbf{y}\right)
-p^{VBN}\left( \mbox{\boldmath${\theta}$}|\mathbf{y}\right) \right) d%
\mbox{\boldmath${\theta}$}}{\int p\left( \mathbf{y}_{rep}|%
\mbox{\boldmath${\theta}$}\right) p^{VBN}\left( \mbox{\boldmath${\theta}$}|%
\mathbf{y}\right) d\mbox{\boldmath${\theta}$}} \\
&=&\frac{\int \frac{p\left( \mathbf{y}_{rep}|\mbox{\boldmath${\theta}$}%
\right) }{p\left( \mathbf{y}_{rep}|\widehat{{\mbox{\boldmath${\theta}$}}}%
_{n}\left( \mathbf{y}_{rep}\right) \right) }\left( p^{VB}\left( %
\mbox{\boldmath${\theta}$}|\mathbf{y}\right) -p^{VBN}\left( %
\mbox{\boldmath${\theta}$}|\mathbf{y}\right) \right) d\mbox{\boldmath${%
\theta}$}}{\int \frac{p\left( \mathbf{y}_{rep}|\mbox{\boldmath${\theta}$}%
\right) }{p\left( \mathbf{y}_{rep}|\widehat{{\mbox{\boldmath${\theta}$}}}%
_{n}\left( \mathbf{y}_{rep}\right) \right) }p^{VBN}\left( %
\mbox{\boldmath${\theta}$}|\mathbf{y}\right) d\mbox{\boldmath${\theta}$}}
\end{eqnarray*}%
where%
\begin{eqnarray*}
&&\left\vert \int \frac{p\left( \mathbf{y}_{rep}|\mbox{\boldmath${\theta}$}%
\right) }{p\left( \mathbf{y}_{rep}|\widehat{{\mbox{\boldmath${\theta}$}}}%
_{n}\left( \mathbf{y}_{rep}\right) \right) }\left( p^{VB}\left( %
\mbox{\boldmath${\theta}$}|\mathbf{y}\right) -p^{VBN}\left( %
\mbox{\boldmath${\theta}$}|\mathbf{y}\right) \right) d\mbox{\boldmath${%
\theta}$}\right\vert \\
&\leq &\int \left\vert \frac{p\left( \mathbf{y}_{rep}|\mbox{\boldmath${%
\theta}$}\right) }{p\left( \mathbf{y}_{rep}|\widehat{{\mbox{\boldmath${%
\theta}$}}}_{n}\left( \mathbf{y}_{rep}\right) \right) }\right\vert
\left\vert p^{VB}\left( \mbox{\boldmath${\theta}$}|\mathbf{y}\right)
-p^{VBN}\left( \mbox{\boldmath${\theta}$}|\mathbf{y}\right) \right\vert d%
\mbox{\boldmath${\theta}$} \\
&\leq &\int \left\vert p^{VB}\left( \mbox{\boldmath${\theta}$}|\mathbf{y}%
\right) -p^{VBN}\left( \mbox{\boldmath${\theta}$}|\mathbf{y}\right)
\right\vert d\mbox{\boldmath${\theta}$}=o_{p}\left( 1\right)
\end{eqnarray*}%
by \citep{Wang_2018,wang2019variational}. Then we have%
\begin{eqnarray*}
&&\int p\left( \mathbf{y}_{rep}|\mbox{\boldmath${\theta}$}\right)
p^{VB}\left( \mbox{\boldmath${\theta}$}|\mathbf{y}\right) d%
\mbox{\boldmath${\theta}$} \\
&=&\int p\left( \mathbf{y}_{rep}|\mbox{\boldmath${\theta}$}\right)
p^{VBN}\left( \mbox{\boldmath${\theta}$}|\mathbf{y}\right) d%
\mbox{\boldmath${\theta}$}\left( 1+o_{p}\left( 1\right) \right)
\end{eqnarray*}%
and%
\begin{eqnarray*}
&&\ln \int p\left( \mathbf{y}_{rep}|\mbox{\boldmath${\theta}$}\right)
p^{VB}\left( \mbox{\boldmath${\theta}$}|\mathbf{y}\right) d%
\mbox{\boldmath${\theta}$} \\
&=&\ln \int p\left( \mathbf{y}_{rep}|\mbox{\boldmath${\theta}$}\right)
p^{VBN}\left( \mbox{\boldmath${\theta}$}|\mathbf{y}\right) d%
\mbox{\boldmath${\theta}$}+o_{p}\left( 1\right) .
\end{eqnarray*}%
Then we can further rewrite $\int p\left( \mathbf{y}_{rep}|%
\mbox{\boldmath${\theta}$}\right) p^{VBN}\left( \mbox{\boldmath${\theta}$}|%
\mathbf{y}\right) d\mbox{\boldmath${\theta}$}$ as 
\begin{eqnarray*}
&&\int p\left( \mathbf{y}_{rep}|\mbox{\boldmath${\theta}$}\right)
p^{VBN}\left( \mbox{\boldmath${\theta}$}|\mathbf{y}\right) d%
\mbox{\boldmath${\theta}$} \\
&\mathbf{=}&\left( \frac{1}{2\pi }\right) ^{\frac{P}{2}}\left\vert \left( -n%
\mathbf{H}_{n}^{d}\right) ^{-1}\right\vert ^{-\frac{1}{2}}\int p\left( 
\mathbf{y}_{rep}|\mbox{\boldmath${\theta}$}\right) \exp \left[ -\frac{n}{2}%
\left( \widehat{{\mbox{\boldmath${\theta}$}}}_{n}\left( \mathbf{y}\right) -%
\mbox{\boldmath${\theta}$}\right) ^{\prime }\left( -\mathbf{H}_{n}^{d}
\left( \widehat{{\mbox{\boldmath${\theta}$}}}_{n}\left( \mathbf{y}\right)
\right)\right) \left( \widehat{{\mbox{\boldmath${\theta}$}}}_{n}\left( 
\mathbf{y}\right) -\mbox{\boldmath${\theta}$}\right) \right] d%
\mbox{\boldmath${\theta}$} \\
&=&\left( \frac{1}{2\pi }\right) ^{\frac{P}{2}}\left\vert \left( -n\mathbf{H}%
_{n}^{d}\right) ^{-1}\right\vert ^{-\frac{1}{2}} \\
&&\times \int \exp \left[ \ln p\left( \mathbf{y}_{rep}|\mbox{\boldmath${%
\theta}$}\right) -\frac{n}{2}\left( \widehat{{\mbox{\boldmath${\theta}$}}}%
_{n}\left( \mathbf{y}\right) -\mbox{\boldmath${\theta}$}\right) ^{\prime
}\left( -\mathbf{H}_{n}^{d} \left( \widehat{{\mbox{\boldmath${\theta}$}}}%
_{n}\left( \mathbf{y}\right) \right)\right) \left( \widehat{{%
\mbox{\boldmath${\theta}$}}}_{n}\left( \mathbf{y}\right) -%
\mbox{\boldmath${\theta}$}\right) \right] d\mbox{\boldmath${\theta}$} \\
&=&\left( \frac{1}{2\pi }\right) ^{\frac{P}{2}}\left\vert \frac{1}{n}\left( -%
\mathbf{H}_{n}^{d}\right) ^{-1}\right\vert ^{-\frac{1}{2}}\left( \frac{1}{%
2\pi }\right) ^{-\frac{P}{2}}\left\vert n\nabla ^{2}h_{N}^{s}\left( 
\widetilde{\mbox{\boldmath${\theta}$}}_{n}^{s}\right) \right\vert
^{-1/2}\exp \left( -nh_{N}^{s}\left( \widetilde{\mbox{\boldmath${\theta}$}}%
_{n}^{s}\right) \right) \left( 1+O_{p}\left( \frac{1}{n}\right) \right)
\end{eqnarray*}%
where 
\begin{equation*}
\begin{aligned} h_{N}^{s}\left( \mbox{\boldmath${\theta}$}\right)
=&-\frac{1}{n}\left( \ln p\left(
\mathbf{y}_{rep}|\mbox{\boldmath${\theta}$}\right) -\frac{n}{2}\left(
\widehat{{\mbox{\boldmath${\theta}$}}}_{n}\left( \mathbf{y}\right)
-\mbox{\boldmath${\theta}$}\right) ^{\prime }\left(
-\mathbf{H}_{n}^{d}\right) \left(
\widehat{{\mbox{\boldmath${\theta}$}}}_{n}\left( \mathbf{y}\right)
-\mbox{\boldmath${\theta}$}\right) \right), \\ \mathbf{H}_{n}^{d} =&
\mathbf{H}_{n}^{d} \left( \widehat{{\mbox{\boldmath${\theta}$}}}_{n}\left(
\mathbf{y}\right) \right). \end{aligned}
\end{equation*}
Note that%
\begin{eqnarray*}
&&\left( \frac{1}{2\pi }\right) ^{\frac{P}{2}}\left\vert \frac{1}{n}\left( -%
\mathbf{H}_{n}^{d}\right) ^{-1}\right\vert ^{-\frac{1}{2}}\left( \frac{1}{%
2\pi }\right) ^{-\frac{P}{2}}\left\vert n\nabla ^{2}h_{N}^{s}\left( 
\widetilde{\mbox{\boldmath${\theta}$}}_{n}^{s}\right) \right\vert ^{-1/2} \\
&=&\left\vert \left( -\mathbf{H}_{n}^{d}\right)^{-1} \right\vert ^{-\frac{1%
}{2}}\left\vert \nabla ^{2}h_{N}^{s}\left( \widetilde{\mbox{\boldmath${%
\theta}$}}_{n}^{s}\right) \right\vert ^{-1/2}=\left\vert \left( -\mathbf{H}%
_{n}^{d}\right) ^{-1}\left( -\frac{1}{n}\frac{\partial \ln p\left( \mathbf{%
y}_{rep}|\widetilde{\mbox{\boldmath${\theta}$}}_{n}\right) }{\partial %
\mbox{\boldmath${\theta}$}\partial \mbox{\boldmath${\theta}$}^{\prime }}%
+\left( -\mathbf{H}_{n}^{d}\right) \right) \right\vert ^{-\frac{1}{2}} \\
&=&\left\vert \left( -\mathbf{H}_{n}^{d}\right) ^{-1}\left( -\mathbf{H}%
_{n}+\left( -\mathbf{H}_{n}^{d}\right) \right) \right\vert ^{-\frac{1}{2}%
}+o_{p}\left( 1\right) =\left\vert \left( -\mathbf{H}_{n}+\left( -\mathbf{H}%
_{n}^{d}\right) \right) \left( -\mathbf{H}_{n}^{d}\right)
^{-1}\right\vert ^{-\frac{1}{2}}+o_{p}\left( 1\right) \\
&=&\left\vert -\mathbf{H}_{n}\left( -\mathbf{H}_{n}^{d}\right) ^{-1}+%
\mathbf{I}_{n}\right\vert ^{-\frac{1}{2}}+o_{p}\left( 1\right) .
\end{eqnarray*}%
Then take logrithm, we have%
\begin{eqnarray}
\ln p^{VB}\left( \mathbf{y}_{rep}|\mathbf{y}\right) &=&\ln \int p\left( 
\mathbf{y}_{rep}|\mbox{\boldmath${\theta}$}\right) p^{VB}\left( %
\mbox{\boldmath${\theta}$}|\mathbf{y}\right) d\mbox{\boldmath${\theta}$}
\label{finalequation1_s} \\
&=&\ln \int p\left( \mathbf{y}_{rep}|\mbox{\boldmath${\theta}$}\right)
p^{VBN}\left( \mbox{\boldmath${\theta}$}|\mathbf{y}\right) d%
\mbox{\boldmath${\theta}$}+o_{p}\left( 1\right)  \notag \\
&=&-\frac{1}{2}\ln \left( \left\vert -\mathbf{H}_{n}\left( -\mathbf{H}%
_{n}^{d}\right) ^{-1}+\mathbf{I}_{n}\right\vert \right) -nh_{N}^{s}\left( 
\widetilde{\mbox{\boldmath${\theta}$}}_{n}^{s}\right) +o_{p}\left( 1\right) 
\notag
\end{eqnarray}%
where second term is%
\begin{eqnarray}
&&-nh_{N}^{s}\left( \widetilde{\mbox{\boldmath${\theta}$}}_{n}^{s}\right) 
\notag \\
&=&\ln p\left( \mathbf{y}_{rep}|\widetilde{\mbox{\boldmath${\theta}$}}%
_{n}^{s}\right) -\frac{n}{2}\left( \widehat{{\mbox{\boldmath${\theta}$}}}%
_{n}\left( \mathbf{y}\right) -\widetilde{\mbox{\boldmath${\theta}$}}%
_{n}^{s}\right) ^{\prime }\left( -\mathbf{H}_{n}^{d}\right) \left( 
\widehat{{\mbox{\boldmath${\theta}$}}}_{n}\left( \mathbf{y}\right) -%
\widetilde{\mbox{\boldmath${\theta}$}}_{n}^{s}\right)  \notag \\
&=&\ln p\left( \mathbf{y}_{rep}|\widehat{{\mbox{\boldmath${\theta}$}}}%
_{n}\left( \mathbf{y}\right) \right) +\ln p\left( \mathbf{y}_{rep}|%
\widetilde{\mbox{\boldmath${\theta}$}}_{n}^{s}\right) -\ln p\left( \mathbf{y}%
_{rep}|\widehat{{\mbox{\boldmath${\theta}$}}}_{n}\left( \mathbf{y}\right)
\right)  \notag \\
&&-\frac{n}{2}\left( \widehat{{\mbox{\boldmath${\theta}$}}}_{n}\left( 
\mathbf{y}\right) -\widetilde{\mbox{\boldmath${\theta}$}}_{n}^{s}\right)
^{\prime }\left( -\mathbf{H}_{n}^{d}\right) \left( \widehat{{%
\mbox{\boldmath${\theta}$}}}_{n}\left( \mathbf{y}\right) -\widetilde{%
\mbox{\boldmath${\theta}$}}_{n}^{s}\right)  \notag \\
&=&\ln p\left( \mathbf{y}_{rep}|\widehat{{\mbox{\boldmath${\theta}$}}}%
_{n}\left( \mathbf{y}\right) \right) +L_{1}+L_{2},  \label{laplaceap1_s}
\end{eqnarray}%
where%
\begin{equation*}
L_{1}=\ln p\left( \mathbf{y}_{rep}|\widetilde{\mbox{\boldmath${\theta}$}}%
_{n}^{s}\right) -\ln p\left( \mathbf{y}_{rep}|\widehat{{\mbox{\boldmath${%
\theta}$}}}_{n}\left( \mathbf{y}\right) \right) \text{, }L_{2}=-\frac{n}{2}%
\left( \widehat{{\mbox{\boldmath${\theta}$}}}_{n}\left( \mathbf{y}\right) -%
\widetilde{\mbox{\boldmath${\theta}$}}_{n}^{s}\right) ^{\prime }\left( -%
\mathbf{H}_{n}^{d}\right) \left( \widehat{{\mbox{\boldmath${\theta}$}}}%
_{n}\left( \mathbf{y}\right) -\widetilde{\mbox{\boldmath${\theta}$}}%
_{n}^{s}\right) .
\end{equation*}%
We can further decompose $L_{1}$ as%
\begin{equation*}
L_{1}=L_{11}+L_{12},
\end{equation*}%
where%
\begin{equation*}
L_{11}=\ln p\left( \mathbf{y}_{rep}|\widetilde{\mbox{\boldmath${\theta}$}}%
_{n}^{s}\right) -\ln p\left( \mathbf{y}_{rep}|\mbox{\boldmath${\theta}$}%
_{n}^{p}\right) \text{, }L_{12}=\ln p\left( \mathbf{y}_{rep}|%
\mbox{\boldmath${\theta}$}_{n}^{p}\right) -\ln p\left( \mathbf{y}_{rep}|%
\widehat{{\mbox{\boldmath${\theta}$}}}_{n}\left( \mathbf{y}\right) \right) .
\end{equation*}

For $L_{11}$, we have%
\begin{eqnarray*}
L_{11} &=&\ln p\left( \mathbf{y}_{rep}|\widetilde{\mbox{\boldmath${\theta}$}}%
_{n}^{s}\right) -\ln p\left( \mathbf{y}_{rep}|\mbox{\boldmath${\theta}$}%
_{n}^{p}\right) \\
&=&\frac{1}{\sqrt{n}}\frac{\partial \ln p\left( \mathbf{y}_{rep}|%
\mbox{\boldmath${\theta}$}_{n}^{p}\right) }{\partial \mbox{\boldmath${%
\theta}$}^{\prime }}\sqrt{n}\left( \widetilde{\mbox{\boldmath${\theta}$}}%
_{n}^{s}-\mbox{\boldmath${\theta}$}_{n}^{p}\right) +\frac{1}{2}\sqrt{n}%
\left( \widetilde{\mbox{\boldmath${\theta}$}}_{n}^{s}-\mbox{\boldmath${%
\theta}$}_{n}^{p}\right) ^{\prime }\frac{1}{n}\frac{\partial ^{2}\ln p\left( 
\mathbf{y}_{rep}|\mbox{\boldmath${\theta}$}_{n}^{p}\right) }{\partial %
\mbox{\boldmath${\theta}$}\partial \mbox{\boldmath${\theta}$}^{\prime }}%
\sqrt{n}\left( \widetilde{\mbox{\boldmath${\theta}$}}_{n}^{s}-%
\mbox{\boldmath${\theta}$}_{n}^{p}\right) +o_{p}\left( 1\right) .
\end{eqnarray*}%
Following Assumption 1-8 and Lemma \ref{lemmajoints}, we can similarly prove
that 
\begin{eqnarray*}
&&\frac{1}{\sqrt{n}}\frac{\partial \ln p\left( \mathbf{y}_{rep}|%
\mbox{\boldmath${\theta}$}_{n}^{p}\right) }{\partial \mbox{\boldmath${%
\theta}$}^{\prime }}\sqrt{n}\left( \widetilde{\mbox{\boldmath${\theta}$}}%
_{n}^{s}-\mbox{\boldmath${\theta}$}_{n}^{p}\right) \\
&=&\sqrt{n}\left( \widehat{{\mbox{\boldmath${\theta}$}}}_{n}\left( \mathbf{y}%
_{rep}\right) -\mbox{\boldmath${\theta}$}_{n}^{p}\right) ^{\prime }\left(
-n^{-1}\sum_{t=1}^{n}\bigtriangledown ^{2}l_{t}\left( \mathbf{y}_{rep}^{t},%
\mbox{\boldmath${\theta}$}_{n}^{p}\right) \right) \sqrt{n}\left( \widetilde{%
\mbox{\boldmath${\theta}$}}_{n}^{s}-\mbox{\boldmath${\theta}$}%
_{n}^{p}\right) +o_{p}\left( 1\right) \\
&=&\sqrt{n}\left( \widehat{{\mbox{\boldmath${\theta}$}}}_{n}\left( \mathbf{y}%
_{rep}\right) -\mbox{\boldmath${\theta}$}_{n}^{p}\right) ^{\prime }\left( -%
\mathbf{H}_{n}\right) \sqrt{n}\left( \widetilde{\mbox{\boldmath${\theta}$}}%
_{n}-\mbox{\boldmath${\theta}$}_{n}^{p}\right) +o_{p}\left( 1\right) \\
&=&\mathbf{tr}\left[ \left( -\mathbf{H}_{n}\right) \sqrt{n}\left( \widetilde{%
\mbox{\boldmath${\theta}$}}_{n}-\mbox{\boldmath${\theta}$}_{n}^{p}\right) 
\sqrt{n}\left( \widehat{{\mbox{\boldmath${\theta}$}}}_{n}\left( \mathbf{y}%
_{rep}\right) -\mbox{\boldmath${\theta}$}_{n}^{p}\right) ^{\prime }\right]
+o_{p}\left( 1\right) .
\end{eqnarray*}%
Hence, we have%
\begin{eqnarray}
&&E_{\mathbf{y}}E_{\mathbf{y}_{rep}}\left[ \frac{1}{\sqrt{n}}\frac{\partial
\ln p\left( \mathbf{y}_{rep}|\mbox{\boldmath${\theta}$}_{n}^{p}\right) }{%
\partial \mbox{\boldmath${\theta}$}^{\prime }}\sqrt{n}\left( \widetilde{%
\mbox{\boldmath${\theta}$}}_{n}^{s}-\mbox{\boldmath${\theta}$}%
_{n}^{p}\right) \right]  \notag \\
&=&E_{\mathbf{y}}E_{\mathbf{y}_{rep}}\left[ \mathbf{tr}\left[ \left( -%
\mathbf{H}_{n}\right) \sqrt{n}\left( \widetilde{\mbox{\boldmath${\theta}$}}%
_{n}^{s}-\mbox{\boldmath${\theta}$}_{n}^{p}\right) \sqrt{n}\left( \widehat{{%
\mbox{\boldmath${\theta}$}}}_{n}\left( \mathbf{y}_{rep}\right) -%
\mbox{\boldmath${\theta}$}_{n}^{p}\right) ^{\prime }\right] +o\left(
1\right) \right]  \notag \\
&=&\mathbf{tr}\left[ \left( -\mathbf{H}_{n}\right) E_{\mathbf{y}}E_{\mathbf{y%
}_{rep}}\left[ \sqrt{n}\left( \widetilde{\mbox{\boldmath${\theta}$}}_{n}^{s}-%
\mbox{\boldmath${\theta}$}_{n}^{p}\right) \sqrt{n}\left( \widehat{{%
\mbox{\boldmath${\theta}$}}}_{n}\left( \mathbf{y}_{rep}\right) -%
\mbox{\boldmath${\theta}$}_{n}^{p}\right) ^{\prime }\right] +o\left(
1\right) \right]  \notag \\
&=&\mathbf{tr}\left[ \left( -\mathbf{H}_{n}\right) \mathbf{G}_{n}\right]
+o\left( 1\right)=\mathbf{tr}\left[ \left( -\mathbf{H}_{n}\right) \left( -%
\mathbf{H}_{n}+\left( -\mathbf{H}_{n}^{d}\right) \right) ^{-1}\mathbf{B}%
_{n}\left( -\mathbf{H}_{n}\right) ^{-1}\right]+o\left( 1\right) \\
&=&\mathbf{tr}\left[ \left( -\mathbf{H}_{n}+\left( -\mathbf{H}%
_{n}^{d}\right) \right) ^{-1}\mathbf{B}_{n}\right] +o\left( 1\right)
\label{eqlap8_s}
\end{eqnarray}%
following Lemma \ref{lemmajoints}. Moreover,%
\begin{eqnarray}
&&\frac{1}{2}\sqrt{n}\left( \widetilde{\mbox{\boldmath${\theta}$}}_{n}^{s}-%
\mbox{\boldmath${\theta}$}_{n}^{p}\right) ^{\prime }\frac{1}{n}\frac{%
\partial ^{2}\ln p\left( \mathbf{y}_{rep}|\mbox{\boldmath${\theta}$}%
_{n}^{p}\right) }{\partial \mbox{\boldmath${\theta}$}\partial %
\mbox{\boldmath${\theta}$}^{\prime }}\sqrt{n}\left( \widetilde{%
\mbox{\boldmath${\theta}$}}_{n}^{s}-\mbox{\boldmath${\theta}$}_{n}^{p}\right)
\notag \\
&=&\frac{1}{2}\sqrt{n}\left( \widetilde{\mbox{\boldmath${\theta}$}}_{n}^{s}-%
\mbox{\boldmath${\theta}$}_{n}^{p}\right) ^{\prime }\mathbf{H}_{n}\sqrt{n}%
\left( \widetilde{\mbox{\boldmath${\theta}$}}_{n}^{s}-\mbox{\boldmath${%
\theta}$}_{n}^{p}\right) +o_{p}\left( 1\right)  \notag \\
&=&\frac{1}{2}\mathbf{tr}\left[\mathbf{H}_{n}\sqrt{n}\left( \widetilde{%
\mbox{\boldmath${\theta}$}}_{n}^{s}-\mbox{\boldmath${\theta}$}%
_{n}^{p}\right) \sqrt{n}\left( \widetilde{\mbox{\boldmath${\theta}$}}%
_{n}^{s}-\mbox{\boldmath${\theta}$}_{n}^{p}\right) ^{\prime }\right]%
+o_{p}\left( 1\right) ,  \label{eqlap2_s}
\end{eqnarray}%
then%
\begin{eqnarray*}
&&E_{\mathbf{y}}E_{\mathbf{y}_{rep}}\left[ \frac{1}{2}\mathbf{tr}\left[%
\mathbf{H}_{n}\sqrt{n}\left( \widetilde{\mbox{\boldmath${\theta}$}}_{n}^{s}-%
\mbox{\boldmath${\theta}$}_{n}^{p}\right) \sqrt{n}\left( \widetilde{%
\mbox{\boldmath${\theta}$}}_{n}^{s}-\mbox{\boldmath${\theta}$}%
_{n}^{p}\right) ^{\prime }\right]\right] \\
&=&\frac{1}{2}\mathbf{tr}\left[ \mathbf{H}_{n}\mathbf{D}_{n}\right] +o\left(
1\right)
\end{eqnarray*}%
%
%
%
%
%
%
%
%
%
%
%
%
%
%
%
%
%
%
%
%
%
%
%
%
%
%where%
%\begin{equation}
%\left[ \left( -2\mathbf{H}_{n}\right) ^{1/2}\sqrt{n}\left( \widetilde{%
%\mbox{\boldmath${\theta}$}}_{n}-\mbox{\boldmath${\theta}$}_{n}^{p}\right) %
%\right] ^{\prime }\left( -2\mathbf{H}_{n}\right) ^{1/2}\sqrt{n}\left(
%\widetilde{\mbox{\boldmath${\theta}$}}_{n}-\mbox{\boldmath${\theta}$}%
%_{n}^{p}\right) \overset{d}{\rightarrow }\chi ^{2}\left( P\right) .
%\label{eqlap2}
%\end{equation}%
From (\ref{eqlap8_s}) and (\ref{eqlap2_s}) we have 
\begin{equation*}
E_{\mathbf{y}}E_{\mathbf{y}_{rep}}\left( L_{11}\right) =\mathbf{tr}\left[
\left( -\mathbf{H}_{n}+\left( -\mathbf{H}_{n}^{d}\right) \right) ^{-1}%
\mathbf{B}_{n}\right] -\frac{1}{2}\mathbf{tr}\left[ \left( -\mathbf{H}%
_{n}\right) \mathbf{D}_{n}\right] +o\left( 1\right)
\end{equation*}%
by Lemma \ref{lemmajoints}.

For $L_{12}$, we have%
\begin{eqnarray*}
&&L_{12}=\ln p\left( \mathbf{y}_{rep}|\mbox{\boldmath${\theta}$}%
_{n}^{p}\right) -\ln p\left( \mathbf{y}_{rep}|\widehat{{\mbox{\boldmath${%
\theta}$}}}_{n}\left( \mathbf{y}\right)\right) = -\frac{1}{\sqrt{n}}\frac{%
\partial \ln p\left( \mathbf{y}_{rep}|\mbox{\boldmath${\theta}$}%
_{n}^{p}\right) }{\partial \mbox{\boldmath${\theta}$}^{\prime }}\sqrt{n}%
\left( \widehat{{\mbox{\boldmath${\theta}$}}}_{n}\left( \mathbf{y}\right) -%
\mbox{\boldmath${\theta}$}_{n}^{p}\right) \\
&&-\frac{1}{2}\left( \widehat{{\mbox{\boldmath${\theta}$}}}_{n}\left( 
\mathbf{y}\right) -\mbox{\boldmath${\theta}$}_{n}^{p}\right) ^{\prime }\frac{%
\partial ^{2}\ln p\left( \mathbf{y}_{rep}|\mbox{\boldmath${\theta}$}%
_{n}^{p}\right) }{\partial \mbox{\boldmath${\theta}$}\partial %
\mbox{\boldmath${\theta}$}^{\prime }}\left( \widehat{{\mbox{\boldmath${%
\theta}$}}}_{n}\left( \mathbf{y}\right) -\mbox{\boldmath${\theta}$}%
_{n}^{p}\right) +o_{p}\left( 1\right) .
\end{eqnarray*}%
Since%
\begin{align}
& E_{\mathbf{y}}E_{\mathbf{y}_{rep}}\left( \left( \widehat{{%
\mbox{\boldmath${\theta}$}}}_{n}\left( \mathbf{y}\right) -%
\mbox{\boldmath${\theta}$}_{n}^{p}\right) ^{\prime }\frac{\partial ^{2}\ln
p\left( \mathbf{y}_{rep}|\mbox{\boldmath${\theta}$}_{n}^{p}\right) }{%
\partial \mbox{\boldmath${\theta}$}\partial \mbox{\boldmath${\theta}$}%
^{\prime }}\left( \widehat{{\mbox{\boldmath${\theta}$}}}_{n}\left( \mathbf{y}%
\right) -\mbox{\boldmath${\theta}$}_{n}^{p}\right) \right)  \notag \\
& =E_{\mathbf{y}}E_{\mathbf{y}_{rep}}\left( \mathbf{tr}\left[ \frac{\partial
^{2}\ln p\left( \mathbf{y}_{rep}|\mbox{\boldmath${\theta}$}_{n}^{p}\right) }{%
\partial \mbox{\boldmath${\theta}$}\partial \mbox{\boldmath${\theta}$}%
^{\prime }}\left( \widehat{{\mbox{\boldmath${\theta}$}}}_{n}\left( \mathbf{y}%
\right) -\mbox{\boldmath${\theta}$}_{n}^{p}\right) \left( \widehat{{%
\mbox{\boldmath${\theta}$}}}_{n}\left( \mathbf{y}\right) -%
\mbox{\boldmath${\theta}$}_{n}^{p}\right) ^{\prime }\right] \right)  \notag
\\
& =\mathbf{tr}\left[ E_{\mathbf{y}_{rep}}\left( \frac{1}{n}\frac{\partial
^{2}\ln p\left( \mathbf{y}_{rep}|\mbox{\boldmath${\theta}$}_{n}^{p}\right) }{%
\partial \mbox{\boldmath${\theta}$}\partial \mbox{\boldmath${\theta}$}%
^{\prime }}\right) E_{\mathbf{y}}\left( n\left( \widehat{{%
\mbox{\boldmath${\theta}$}}}_{n}\left( \mathbf{y}\right) -%
\mbox{\boldmath${\theta}$}_{n}^{p}\right) \left( \widehat{{%
\mbox{\boldmath${\theta}$}}}_{n}\left( \mathbf{y}\right) -%
\mbox{\boldmath${\theta}$}_{n}^{p}\right) ^{\prime }\right) \right]  \notag
\\
& =-\mathbf{tr}\left[ \mathbf{B}_{n}\left( -\mathbf{H}_{n}\right) ^{-1}%
\right] +o\left( 1\right)  \label{eqlap3_s}
\end{align}

\begin{equation}
E_{\mathbf{y}_{rep}}\left( \frac{1}{\sqrt{n}}\frac{\partial \ln p\left( 
\mathbf{y}_{rep}|\mbox{\boldmath${\theta}$}_{n}^{p}\right) }{\partial %
\mbox{\boldmath${\theta}$}}\right) =0,E_{\mathbf{y}_{rep}}\left( \sqrt{n}%
\left( \widehat{{\mbox{\boldmath${\theta}$}}}_{n}\left( \mathbf{y}\right) -%
\mbox{\boldmath${\theta}$}_{n}^{p}\right) \right) =o\left( 1\right)
\label{eqlap4_s}
\end{equation}%
from (\ref{eqlap3_s}), and (\ref{eqlap4_s}), we have%
\begin{equation*}
E_{\mathbf{y}}E_{\mathbf{y}_{rep}}\left( L_{12}\right) =\frac{1}{2}\mathbf{tr%
}\left[ \mathbf{B}_{n}\left( -\mathbf{H}_{n}\right) ^{-1}\right] +o\left(
1\right) .
\end{equation*}%
Then%
\begin{eqnarray}
E_{\mathbf{y}}E_{\mathbf{y}_{rep}}\left( L_{1}\right) &=&E_{\mathbf{y}}E_{%
\mathbf{y}_{rep}}\left( \ln p\left( \mathbf{y}_{rep}|\widetilde{%
\mbox{\boldmath${\theta}$}}_{n}\right) -\ln p\left( \mathbf{y}_{rep}|%
\overleftrightarrow{{\mbox{\boldmath${\theta}$}}}_{n}\right) \right) =E_{%
\mathbf{y}}E_{\mathbf{y}_{rep}}\left( L_{11}+L_{12}\right)
\label{finalequation2_s} \\
&=&\mathbf{tr}\left[ \left( -\mathbf{H}_{n}+\left( -\mathbf{H}%
_{n}^{d}\right) \right) ^{-1}\mathbf{B}_{n}\right] -\frac{1}{2}\mathbf{tr}%
\left[ \left( -\mathbf{H}_{n}\right) \mathbf{D}_{n}\right] +\frac{1}{2}%
\mathbf{tr}\left[ \mathbf{B}_{n}\left( -\mathbf{H}_{n}\right) ^{-1}\right]
+o\left( 1\right) .  \notag
\end{eqnarray}%
Similarly, we can decompose $L_{2}$ $=-\frac{n}{2}\left( \widehat{{%
\mbox{\boldmath${\theta}$}}}_{n}\left( \mathbf{y}\right) -\widetilde{%
\mbox{\boldmath${\theta}$}}_{n}^{s}\right) ^{\prime }\left( -\mathbf{H}%
_{n}^{d}\right) \left( \widehat{{\mbox{\boldmath${\theta}$}}}_{n}\left( 
\mathbf{y}\right) -\widetilde{\mbox{\boldmath${\theta}$}}_{n}^{s}\right) $ as%
\begin{equation*}
L_{2}=L_{21}+L_{22}+L_{23}+L_{24},
\end{equation*}%
where%
\begin{equation*}
L_{21}=-\frac{n}{2}\left( \widehat{{\mbox{\boldmath${\theta}$}}}_{n}\left( 
\mathbf{y}\right) -\mbox{\boldmath${\theta}$}_{n}^{p}\right) ^{\prime
}\left( -\mathbf{H}_{n}^{d}\right) \left( \widehat{{\mbox{\boldmath${%
\theta}$}}}_{n}\left( \mathbf{y}\right) -\mbox{\boldmath${\theta}$}%
_{n}^{p}\right) ,L_{22}=-\frac{n}{2}\left( \widehat{{\mbox{\boldmath${%
\theta}$}}}_{n}\left( \mathbf{y}\right) -\mbox{\boldmath${\theta}$}%
_{n}^{p}\right) ^{\prime }\left( -\mathbf{H}_{n}^{d}\right) \left( %
\mbox{\boldmath${\theta}$}_{n}^{p}-\widetilde{\mbox{\boldmath${\theta}$}}%
_{n}^{s}\right) ,
\end{equation*}%
\begin{equation*}
L_{23}=-\frac{n}{2}\left( \mbox{\boldmath${\theta}$}_{n}^{p}-\widetilde{%
\mbox{\boldmath${\theta}$}}_{n}^{s}\right) ^{\prime }\left( -\mathbf{H}%
_{n}^{d}\right) \left( \widehat{{\mbox{\boldmath${\theta}$}}}_{n}\left( 
\mathbf{y}\right) -\mbox{\boldmath${\theta}$}_{n}^{p}\right) ,L_{24}=-\frac{n%
}{2}\left( \mbox{\boldmath${\theta}$}_{n}^{p}-\widetilde{\mbox{\boldmath${%
\theta}$}}_{n}^{s}\right) ^{\prime }\left( -\mathbf{H}_{n}^{d}\right)
\left( \mbox{\boldmath${\theta}$}_{n}^{p}-\widetilde{\mbox{\boldmath${%
\theta}$}}_{n}^{s}\right) .
\end{equation*}%
For $L_{21}$, we have%
\begin{eqnarray*}
L_{21} &=&-\frac{n}{2}\left( \widehat{{\mbox{\boldmath${\theta}$}}}%
_{n}\left( \mathbf{y}\right) -\mbox{\boldmath${\theta}$}_{n}^{p}\right)
^{\prime }\left( -\mathbf{H}_{n}^{d}\right) \left( \widehat{{%
\mbox{\boldmath${\theta}$}}}_{n}\left( \mathbf{y}\right) -%
\mbox{\boldmath${\theta}$}_{n}^{p}\right) \\
&=&-\frac{1}{2}\sqrt{n}\left( \widehat{{\mbox{\boldmath${\theta}$}}}%
_{n}\left( \mathbf{y}\right) -\mbox{\boldmath${\theta}$}_{n}^{p}\right)
^{\prime }\left( -\mathbf{H}_{n}^{d}\right) \sqrt{n}\left( \widehat{{%
\mbox{\boldmath${\theta}$}}}_{n}\left( \mathbf{y}\right) -%
\mbox{\boldmath${\theta}$}_{n}^{p}\right) \\
&=&-\frac{1}{2}\mathbf{tr}\left[ \left( -\mathbf{H}_{n}^{d}\right) \sqrt{n}%
\left( \widehat{{\mbox{\boldmath${\theta}$}}}_{n}\left( \mathbf{y}\right) -%
\mbox{\boldmath${\theta}$}_{n}^{p}\right) \sqrt{n}\left( \widehat{{%
\mbox{\boldmath${\theta}$}}}_{n}\left( \mathbf{y}\right) -%
\mbox{\boldmath${\theta}$}_{n}^{p}\right) ^{\prime }\right] ,
\end{eqnarray*}%
then%
\begin{equation*}
E_{\mathbf{y}}E_{\mathbf{y}_{rep}}\left( L_{21}\right) =-\frac{1}{2}\mathbf{%
tr}\left[ \left( -\mathbf{H}_{n}^{d}\right) \mathbf{C}_{n}\right] +o\left(
1\right) .
\end{equation*}%
For $L_{22}$ and $L_{23}$, we have%
\begin{eqnarray*}
L_{22} &=&L_{23}=-\frac{n}{2}\left( \widehat{{\mbox{\boldmath${\theta}$}}}%
_{n}\left( \mathbf{y}\right) -\mbox{\boldmath${\theta}$}_{n}^{p}\right)
^{\prime }\left( -\mathbf{H}_{n}^{d}\right) \left( \mbox{\boldmath${%
\theta}$}_{n}^{p}-\widetilde{\mbox{\boldmath${\theta}$}}_{n}^{s}\right) \\
&=&-\frac{1}{2}\sqrt{n}\left( \widehat{{\mbox{\boldmath${\theta}$}}}%
_{n}\left( \mathbf{y}\right) -\mbox{\boldmath${\theta}$}_{n}^{p}\right)
^{\prime }\left( -\mathbf{H}_{n}^{d}\right) \sqrt{n}\left( %
\mbox{\boldmath${\theta}$}_{n}^{p}-\widetilde{\mbox{\boldmath${\theta}$}}%
_{n}^{s}\right) \\
&=&-\frac{1}{2}\mathbf{tr}\left[ \left( -\mathbf{H}_{n}^{d}\right) \sqrt{n}%
\left( \mbox{\boldmath${\theta}$}_{n}^{p}-\widetilde{\mbox{\boldmath${%
\theta}$}}_{n}^{s}\right) \sqrt{n}\left( \widehat{{\mbox{\boldmath${\theta}$}%
}}_{n}\left( \mathbf{y}\right) -\mbox{\boldmath${\theta}$}_{n}^{p}\right)
^{\prime }\right] \\
&=&\frac{1}{2}\mathbf{tr}\left[ \left( -\mathbf{H}_{n}^{d}\right) \sqrt{n}%
\left( \widetilde{\mbox{\boldmath${\theta}$}}_{n}^{s}-\mbox{\boldmath${%
\theta}$}_{n}^{p}\right) \sqrt{n}\left( \widehat{{\mbox{\boldmath${\theta}$}}%
}_{n}\left( \mathbf{y}\right) -\mbox{\boldmath${\theta}$}_{n}^{p}\right)
^{\prime }\right]
\end{eqnarray*}%
then%
\begin{equation*}
E_{\mathbf{y}}E_{\mathbf{y}_{rep}}\left( L_{22}\right) =E_{\mathbf{y}}E_{%
\mathbf{y}_{rep}}\left( L_{23}\right) =\frac{1}{2}\mathbf{tr}\left[ \left( -%
\mathbf{H}_{n}^{d}\right) \mathbf{F}_{n}\right] +o\left( 1\right) .
\end{equation*}%
For $L_{24}$, we have%
\begin{eqnarray*}
L_{24} &=&-\frac{n}{2}\left( \mbox{\boldmath${\theta}$}_{n}^{p}-\widetilde{%
\mbox{\boldmath${\theta}$}}_{n}^{s}\right) ^{\prime }\left( -\mathbf{H}%
_{n}^{d}\right) \left( \mbox{\boldmath${\theta}$}_{n}^{p}-\widetilde{%
\mbox{\boldmath${\theta}$}}_{n}^{s}\right) \\
&=&-\frac{1}{2}\sqrt{n}\left( \mbox{\boldmath${\theta}$}_{n}^{p}-\widetilde{%
\mbox{\boldmath${\theta}$}}_{n}^{s}\right) ^{\prime }\left( -\mathbf{H}%
_{n}^{d}\right) \sqrt{n}\left( \mbox{\boldmath${\theta}$}_{n}^{p}-%
\widetilde{\mbox{\boldmath${\theta}$}}_{n}^{s}\right) \\
&=&-\frac{1}{2}\mathbf{tr}\left[ \left( -\mathbf{H}_{n}^{d} \left( 
\widehat{{\mbox{\boldmath${\theta}$}}}_{n}\left( \mathbf{y}\right)
\right)\right) \sqrt{n}\left( \mbox{\boldmath${\theta}$}_{n}^{p}-\widetilde{%
\mbox{\boldmath${\theta}$}}_{n}^{s}\right) \sqrt{n}\left( %
\mbox{\boldmath${\theta}$}_{n}^{p}-\widetilde{\mbox{\boldmath${\theta}$}}%
_{n}^{s}\right) ^{\prime }\right] ,
\end{eqnarray*}%
then%
\begin{equation*}
E_{\mathbf{y}}E_{\mathbf{y}_{rep}}\left( L_{24}\right) =-\frac{1}{2}\mathbf{%
tr}\left[ \left( -\mathbf{H}_{n}^{d}\right) \mathbf{D}_{n}\right] +o\left(
1\right) .
\end{equation*}%
Hence we have%
\begin{eqnarray}
E_{\mathbf{y}}E_{\mathbf{y}_{rep}}\left( L_{2}\right) &=&E_{\mathbf{y}}E_{%
\mathbf{y}_{rep}}\left( L_{21}+L_{22}+L_{23}+L_{24}\right)
\label{finalequation3_s} \\
&=&-\frac{1}{2}\mathbf{tr}\left[ \left( -\mathbf{H}_{n}^{d}\right) \mathbf{%
C}_{n}\right] +\mathbf{tr}\left[ \left( -\mathbf{H}_{n}^{d}\right) \mathbf{%
F}_{n}\right] -\frac{1}{2}\mathbf{tr}\left[ \left( -\mathbf{H}%
_{n}^{d}\right) \mathbf{D}_{n}\right] +o\left( 1\right) .  \notag
\end{eqnarray}

Note that%
\begin{equation*}
{\mbox{\boldmath${\bar\theta}$}}_{n}^{VB}=\widehat{{\mbox{\boldmath${%
\theta}$}}}_{n}\left( \mathbf{y}\right) +o_{p}(n^{-1/2}),
\end{equation*}%
by \citet{Wang_2018} and \citet{Zhang_2024}. Mimicking the proof of \citet{li2024deviance}, we get%
\begin{equation}
E_{\mathbf{y}}E_{\mathbf{y}_{rep}}\ln p\left( \mathbf{y}_{rep}|\widehat{{%
\mbox{\boldmath${\theta}$}}}_{n}\left( \mathbf{y}\right) \right) =E_{\mathbf{%
y}}\left[ \ln p\left( \mathbf{y}|\widehat{{\mbox{\boldmath${\theta}$}}}%
_{n}\left( \mathbf{y}\right) \right) \right] -\mathbf{tr}\left[ \mathbf{B}%
_{n}\left( -\mathbf{H}_{n}\right) ^{-1}\right] .  \label{finalequation4_s}
\end{equation}%
With (\ref{finalequation1_s}), (\ref{finalequation2_s}), (\ref%
{finalequation3_s}) and (\ref{finalequation4_s}), we have%
\begin{eqnarray*}
&&E_{\mathbf{y}}\left[ E_{\mathbf{y}_{rep}}\ln p^{VB}\left( \mathbf{y}_{rep}|%
\mathbf{y}\right) \right] \\
&=&-\frac{1}{2}\ln \left( \left\vert -\mathbf{H}_{n}\left( -\mathbf{H}%
_{n}^{d}\right) ^{-1}+\mathbf{I}_{n}\right\vert \right) +E_{\mathbf{y}}E_{%
\mathbf{y}_{rep}}\ln p\left( \mathbf{y}_{rep}|\widehat{{\mbox{\boldmath${%
\theta}$}}}_{n}\left( \mathbf{y}\right) \right) +E_{\mathbf{y}}E_{\mathbf{y}%
_{rep}}\left( L_{1}+L_{2}\right) \\
&=&-\frac{1}{2}\ln \left( \left\vert -\mathbf{H}_{n}\left( -\mathbf{H}%
_{n}^{d}\right) ^{-1}+\mathbf{I}_{n}\right\vert \right) +E_{\mathbf{y}}%
\left[ \ln p\left( \mathbf{y}|\widehat{{\mbox{\boldmath${\theta}$}}}%
_{n}\left( \mathbf{y}\right) \right) \right] -\mathbf{tr}\left[ \mathbf{B}%
_{n}\left( -\mathbf{H}_{n}\right) ^{-1}\right] \\
&&+\mathbf{tr}\left[ \left( -\mathbf{H}_{n}+\left( -\mathbf{H}%
_{n}^{d}\right) \right) ^{-1}\mathbf{B}_{n}\right] -\frac{1}{2}\mathbf{tr}%
\left[ \left( -\mathbf{H}_{n}\right) \mathbf{D}_{n}\right] +\frac{1}{2}%
\mathbf{tr}\left[ \mathbf{B}_{n}\left( -\mathbf{H}_{n}\right) ^{-1}\right] \\
&&-\frac{1}{2}\mathbf{tr}\left[ \left( -\mathbf{H}_{n}^{d}\right) \mathbf{C%
}_{n}\right] +\mathbf{tr}\left[ \left( -\mathbf{H}_{n}^{d}\right) \mathbf{F%
}_{n}\right] -\frac{1}{2}\mathbf{tr}\left[ \left( -\mathbf{H}%
_{n}^{d}\right) \mathbf{D}_{n}\right] +o\left( 1\right) \\
&=&-\frac{1}{2}\ln \left( \left\vert -\mathbf{H}_{n}\left( -\mathbf{H}%
_{n}^{d}\right) ^{-1}+\mathbf{I}_{n}\right\vert \right) +E_{\mathbf{y}}%
\left[ \ln p\left( \mathbf{y}|\widehat{{\mbox{\boldmath${\theta}$}}}%
_{n}\left( \mathbf{y}\right) \right) \right] -\frac{1}{2}\mathbf{tr}\left[ 
\mathbf{B}_{n}\left( -\mathbf{H}_{n}\right) ^{-1}\right] \\
&&+\mathbf{tr}\left[ \left( -\mathbf{H}_{n}+\left( -\mathbf{H}%
_{n}^{d}\right) \right) ^{-1}\mathbf{B}_{n}\right] -\frac{1}{2}\mathbf{tr}%
\left[ \left( -\mathbf{H}_{n}\right) \mathbf{D}_{n}\right] \\
&&-\frac{1}{2}\mathbf{tr}\left[ \left( -\mathbf{H}_{n}^{d}\right) \mathbf{C%
}_{n}\right] +\mathbf{tr}\left[ \left( -\mathbf{H}_{n}^{d}\right) \mathbf{F%
}_{n}\right] -\frac{1}{2}\mathbf{tr}\left[ \left( -\mathbf{H}%
_{n}^{d}\right) \mathbf{D}_{n}\right] +o\left( 1\right) \\
&=&-\frac{1}{2}\ln \left( \left\vert -\mathbf{H}_{n}\left( -\mathbf{H}%
_{n}^{d}\right) ^{-1}+\mathbf{I}_{n}\right\vert \right) +E_{\mathbf{y}}%
\left[ \ln p\left( \mathbf{y}|\widehat{{\mbox{\boldmath${\theta}$}}}%
_{n}\left( \mathbf{y}\right) \right) \right] -\frac{1}{2}\mathbf{tr}\left[ 
\mathbf{B}_{n}\left( -\mathbf{H}_{n}\right) ^{-1}\right] \\
&&+\mathbf{tr}\left[ \left( -\mathbf{H}_{n}+\left( -\mathbf{H}%
_{n}^{d}\right) \right) ^{-1}\mathbf{B}_{n}\right] -\frac{1}{2}\mathbf{tr}%
\left[ \left( -\mathbf{H}_{n}-\mathbf{H}_{n}^{d}\right) \mathbf{D}_{n}%
\right] \\
&&-\frac{1}{2}\mathbf{tr}\left[ \left( -\mathbf{H}_{n}^{d}\right) \mathbf{C%
}_{n}\right] +\mathbf{tr}\left[ \left( -\mathbf{H}_{n}^{d}\right) \mathbf{F%
}_{n}\right] +o\left( 1\right).
\end{eqnarray*}%
Then we have%
\begin{eqnarray*}
&&E_{\mathbf{y}}\left[ E_{\mathbf{y}_{rep}}\ln p^{VB}\left( \mathbf{y}_{rep}|%
\mathbf{y}\right) \right] \\
&=&E_{\mathbf{y}}\left[ \ln p\left( \mathbf{y}|\widehat{{\mbox{\boldmath${%
\theta}$}}}_{n}\left( \mathbf{y}\right) \right) \right] -\frac{1}{2}\ln
\left( \left\vert -\mathbf{H}_{n}\left( -\mathbf{H}_{n}^{d}\right)^{-1}+%
\mathbf{I}_{n}\right\vert \right) -\frac{1}{2}\mathbf{tr}\left[ \mathbf{B}%
_{n}\left( -\mathbf{H}_{n}\right) ^{-1}\right] \\
&&+\mathbf{tr}\left[ \left( -\mathbf{H}_{n}+\left( -\mathbf{H}%
_{n}^{d}\right) \right) ^{-1}\mathbf{B}_{n}\right] -\frac{1}{2}\mathbf{tr}%
\left[ \left( -\mathbf{H}_{n}-\mathbf{H}_{n}^{d}\right) \mathbf{D}_{n}%
\right] -\frac{1}{2}\mathbf{tr}\left[ \left( -\mathbf{H}_{n}^{d}\right) 
\mathbf{C}_{n}\right] \\
&&+\mathbf{tr}\left[ \left( -\mathbf{H}_{n}^{d}\right) \mathbf{F}_{n}%
\right] +o\left( 1\right) \\
&=&E_{\mathbf{y}}\left[ \ln p\left( \mathbf{y}|\widehat{{\mbox{\boldmath${%
\theta}$}}}_{n}\left( \mathbf{y}\right) \right) \right] -\frac{1}{2}\ln
\left( \left\vert -\mathbf{H}_{n}\left( -\mathbf{H}_{n}^{d}\right) ^{-1}+%
\mathbf{I}_{n}\right\vert \right) -\frac{1}{2}\mathbf{tr}\left[ \mathbf{B}%
_{n}\left( -\mathbf{H}_{n}\right) ^{-1}\right] \\
&&+\frac{1}{2}\mathbf{tr}\left[ \left( -\mathbf{H}_{n}+\left( -\mathbf{H}%
_{n}^{d}\right) \right) ^{-1}\mathbf{B}_{n}\right] -\frac{1}{2}\mathbf{tr}%
\left[ \left( -\mathbf{H}_{n}^{d}\right) \mathbf{C}_{n}\left( -\mathbf{H}%
_{n}^{d}\right) \left( -\mathbf{H}_{n}+\left( -\mathbf{H}_{n}^{d}\right)
\right) ^{-1}\right] \\
&&-\frac{1}{2}\mathbf{tr}\left[ \left( -\mathbf{H}_{n}^{d}\right) \mathbf{C%
}_{n}\right] +\mathbf{tr}\left[ \left( -\mathbf{H}_{n}^{d}\right) \left( -%
\mathbf{H}_{n}+\left( -\mathbf{H}_{n}^{d}\right) \right) ^{-1}\left( -%
\mathbf{H}_{n}^{d}\right) \mathbf{C}_{n}\right] +o\left( 1\right) \\
&=&E_{\mathbf{y}}\left[ \ln p\left( \mathbf{y}|\widehat{{\mbox{\boldmath${%
\theta}$}}}_{n}\left( \mathbf{y}\right) \right) \right] -\frac{1}{2}\ln
\left( \left\vert -\mathbf{H}_{n}\left( -\mathbf{H}_{n}^{d}\right) ^{-1}+%
\mathbf{I}_{n}\right\vert \right) -\frac{1}{2}\mathbf{tr}\left[ \mathbf{B}%
_{n}\left( -\mathbf{H}_{n}\right) ^{-1}\right] \\
&&+\frac{1}{2}\mathbf{tr}\left[ \left( -\mathbf{H}_{n}+\left( -\mathbf{H}%
_{n}^{d}\right) \right) ^{-1}\mathbf{B}_{n}\right] +\frac{1}{2}\mathbf{tr}%
\left[ \left( -\mathbf{H}_{n}^{d}\right) \mathbf{C}_{n}\left( -\mathbf{H}%
_{n}^{d}\right) \left( -\mathbf{H}_{n}+\left( -\mathbf{H}_{n}^{d}\right)
\right) ^{-1}\right] \\
&&-\frac{1}{2}\mathbf{tr}\left[ \left( -\mathbf{H}_{n}^{d}\right) \mathbf{C%
}_{n}\right] +o\left( 1\right) \\
&=&E_{\mathbf{y}}\left[ \ln p\left( \mathbf{y}|\widehat{{\mbox{\boldmath${%
\theta}$}}}_{n}\left( \mathbf{y}\right) \right) \right] -\frac{1}{2}\ln
\left( \left\vert -\mathbf{H}_{n}\left( -\mathbf{H}_{n}^{d}\right) ^{-1}+%
\mathbf{I}_{n}\right\vert \right) -\frac{1}{2}\mathbf{tr}\left[ \mathbf{B}%
_{n}\left( -\mathbf{H}_{n}\right) ^{-1}\right] \\
&&+\frac{1}{2}\mathbf{tr}\left[ \left( -\mathbf{H}_{n}+\left( -\mathbf{H}%
_{n}^{d}\right) \right) ^{-1}\left( \mathbf{B}_{n}+\left( -\mathbf{H}%
_{n}^{d}\right) \mathbf{C}_{n}\left( -\mathbf{H}_{n}^{d}\right) \right) %
\right] -\frac{1}{2}\mathbf{tr}\left[ \left( -\mathbf{H}_{n}^{d}\right) 
\mathbf{C}_{n}\right] +o\left( 1\right)
\end{eqnarray*}%
Therefore, 
\begin{equation*}
\begin{aligned} -\ln p\left(
\mathbf{y}|\widehat{{\mbox{\boldmath${\theta}$}}}_{n}\left(
\mathbf{y}\right) \right) &+\frac{1}{2}\ln \left( \left\vert
-\mathbf{H}_{n}\left( -\mathbf{H}_{n}^{d}\right)
^{-1}+\mathbf{I}_{n}\right\vert \right) +\frac{1}{2}\mathbf{tr}\left[
\mathbf{B}_{n}\left( -\mathbf{H}_{n}\right) ^{-1}\right] \\ +
\frac{1}{2}\mathbf{tr}&\left[ \left( -\mathbf{H}_{n}^{d}\right)
\mathbf{C}_{n}\right] -\frac{1}{2}\mathbf{tr}\left[ \left(
-\mathbf{H}_{n}+\left( -\mathbf{H}_{n}^{d}\right) \right) ^{-1}\left(
\mathbf{B}_{n}+\left( -\mathbf{H}_{n}^{d}\right) \mathbf{C}_{n}\left(
-\mathbf{H}_{n}^{d}\right) \right) \right] \end{aligned}
\end{equation*}
is an unbiased estimator of $E_{\mathbf{y}_{rep}}\left( -\ln p^{VB}\left( 
\mathbf{y}_{rep}|\mathbf{y}\right) \right) $ asymptotically.

\subsubsection{Proof of Theorem 4.1}
%\protect \protect\ref{theom vtic}}
We are now in the position to prove Theorem 4.1.
%\ref{theom vtic}. 
The key step is to prove that both $$\mathbf{\bar{\Omega}}_{n}\left( {%
\mbox{\boldmath${\bar{\theta}}$}}^{VB}\right) \quad \text{and} \quad \mathbf{\bar{H}}_{n}\left( {\ \mbox{\boldmath${\bar{\theta}}$}}^{VB}\right) $$
are the consistent estimator of both 
$\mathbf{B}_{n}({\mbox{\boldmath${\theta}$}}_{n}^{p})$ and $\mathbf{H}_{n}({
\mbox{\boldmath${\theta}$}}_{n}^{p})$, where $\boldsymbol{\bar{\theta}}^{VB}$ is the VB posterior mean.

From (\ref{var score f}), we have
$$
\frac{1}{\sqrt{n}} \mathbf{B}_n^{-1/2} \frac{\partial \ln p\left(\mathbf{y}| \boldsymbol{\theta}_n^p\right)}{\partial \boldsymbol{\theta}} \xrightarrow{d} N\left(0, \mathbf{I}_P\right).
$$
It should be noted that, 
$$
\mathbf{s}(\mathbf{y},{\mbox{\boldmath${\theta}$}})=\frac{\partial \ln p(%
\mathbf{y}|{\mbox{\boldmath${\theta}$}})}{\partial {\mbox{\boldmath${%
\theta}$}}}=\sum_{t=1}^{n}\bigtriangledown l_{t}\left( {\mbox{\boldmath${%
\theta}$}}\right),
$$
the left side of (\ref{var score f}) is equivalent to
\begin{equation}
    \label{expasion of var score f}
    \begin{aligned}
        &\frac{1}{\sqrt{n}} \mathbf{B}_n^{-1/2} \frac{\partial \ln p\left(\mathbf{y}| \boldsymbol{\theta}_n^p\right)}{\partial \boldsymbol{\theta}} = \frac{1}{\sqrt{n}} \mathbf{B}_n^{-1/2} \sum_{t=1}^{n}\bigtriangledown l_{t}\left( {\mbox{\boldmath${\theta}$}}\right)\\
        = \frac{1}{\sqrt{n}} \mathbf{B}_n^{-1/2}&\sum_{t=1}^{n}\left(\bigtriangledown l_{t}\left( {\mbox{\boldmath${\theta}$}_n^p}\right)-\bigtriangledown l_{t}\left( {\mbox{\boldmath${\bar\theta}$}}^{VB}\right)\right)+\frac{1}{\sqrt{n}} \mathbf{B}_n^{-1/2}\sum_{t=1}^{n}\bigtriangledown l_{t}\left( {\mbox{\boldmath${\bar\theta}$}}^{VB}\right),
    \end{aligned}
\end{equation}
for the first term we have 
$$
\bigtriangledown l_{t}\left( {\mbox{\boldmath${\theta}$}_n^p}\right)-\bigtriangledown l_{t}\left( {\mbox{\boldmath${\bar\theta}$}}^{VB}\right)=\bigtriangledown ^2 l_{t}\left( {\mbox{\boldmath${\bar\theta}$}}^{\#*}\right)\left({\mbox{\boldmath${\theta}$}_n^p}-{\mbox{\boldmath${\bar\theta}$}}^{VB}\right),
$$
where ${\mbox{\boldmath${\bar\theta}$}}^{\#*}$ lies in ${\mbox{\boldmath${\theta}$}_n^p}$ and ${\mbox{\boldmath${\bar\theta}$}}^{VB}$. From Assumption 5 and (\ref{vb consistency psedo true}) , we have $\left\|\bigtriangledown ^2 l_{t}\left( {\mbox{\boldmath${\bar\theta}$}}^{\#*}\right)\right\|$ is bounded, and $\overline{\boldsymbol{\theta}}^{V B}(\mathbf{y})=\boldsymbol{\theta}_n^p+O_p\left(n^{-1 / 2}\right)$, so we derive that
$$
\bigtriangledown l_{t}\left( {\mbox{\boldmath${\theta}$}_n^p}\right)-\bigtriangledown l_{t}\left( {\mbox{\boldmath${\bar\theta}$}}^{VB}\right) = O_p\left(n^{-1 / 2}\right). 
$$
Because $\mathbf{B}_{n}\left( {\mbox{\boldmath${\theta}$}}\right)=Var\left[ \frac{%
1}{\sqrt{n}}\sum_{t=1}^{n}\bigtriangledown l_{t}\left( {\mbox{\boldmath${%
\theta}$}}\right) \right]$, under Assumption 5, $\mathbf{B}_{n}$ is also bounded, so we finally get
\begin{equation}
    \frac{1}{\sqrt{n}} \mathbf{B}_n^{-1/2}\sum_{t=1}^{n}\left(\bigtriangledown l_{t}\left( {\mbox{\boldmath${\theta}$}_n^p}\right)-\bigtriangledown l_{t}\left( {\mbox{\boldmath${\bar\theta}$}}^{VB}\right)\right)=O_p\left(n^{-1}\right).
\end{equation}
Combined with (\ref{expasion of var score f}), we have 
\begin{equation}
    \label{expasion2 of var score f}
    \begin{aligned}
        &\frac{1}{\sqrt{n}} \mathbf{B}_n^{-1/2} \frac{\partial \ln p\left(\mathbf{y}| \boldsymbol{\theta}_n^p\right)}{\partial \boldsymbol{\theta}} = \frac{1}{\sqrt{n}} \mathbf{B}_n^{-1/2} \sum_{t=1}^{n}\bigtriangledown l_{t}\left( {\mbox{\boldmath${\theta}$}}\right)\\
        = \frac{1}{\sqrt{n}} \mathbf{B}_n^{-1/2}&\sum_{t=1}^{n}\left(\bigtriangledown l_{t}\left( {\mbox{\boldmath${\theta}$}_n^p}\right)-\bigtriangledown l_{t}\left( {\mbox{\boldmath${\bar\theta}$}}^{VB}\right)\right)+\frac{1}{\sqrt{n}} \mathbf{B}_n^{-1/2}\sum_{t=1}^{n}\bigtriangledown l_{t}\left( {\mbox{\boldmath${\bar\theta}$}}^{VB}\right)\\
        =O_p\left(n^{-1}\right)&+\frac{1}{\sqrt{n}} \mathbf{B}_n^{-1/2}\sum_{t=1}^{n}\bigtriangledown l_{t}\left( {\mbox{\boldmath${\bar\theta}$}}^{VB}\right) \xrightarrow{d} N\left(0, \mathbf{I}_P\right).
    \end{aligned}
\end{equation}
Note that $\mathbf{\bar{\Omega}}_{n}\left( {%
\mbox{\boldmath${\bar{\theta}}$}}^{VB}\right)=\sum_{t=1}^{n}\bigtriangledown l_{t}\left( {\mbox{\boldmath${\bar\theta}$}}^{VB}\right)\bigtriangledown l_{t}\left( {\mbox{\boldmath${\bar\theta}$}}^{VB}\right)^{\prime}$ we finally have 
\begin{equation}
\label{vb consistent omega bar}
\mathbf{\bar{\Omega}}_{n}\left( {%
\mbox{\boldmath${\bar{\theta}}$}}^{VB}\right) = \mathbf{B}_n\left(\boldsymbol{\theta}_n^p\right)+O_p\left(n^{-1}\right).
\end{equation}
With (\ref{t2 scaling inequality term1}) in proof of Theorem 3.1.
%\ref{riskvtic},
\begin{equation}
\label{vb consistent hn bar}
    \mathbf{\bar{H}}_{n}\left( {\ \mbox{\boldmath${\bar{\theta}}$}}^{VB}\right) = \mathbf{H}_{n}({
\mbox{\boldmath${\theta}$}}_{n}^{p}) + O_p\left(n^{-1/2}\right).
\end{equation}
Combined (\ref{vb consistent omega bar}) and (\ref{vb consistent hn bar}), we have
\begin{equation}
\label{vb consistent trace hn and bn bar}
    \mathbf{tr}\left[
\mathbf{\bar{\Omega}}_{n}\left(
{\mbox{\boldmath${\bar{					\theta}}$}}^{VB}\right) \left(
\mathbf{\bar{H}}_{n}\left( {
\mbox{\boldmath${\bar{\theta}}$}}^{VB}\right) \right) ^{-1}\right] = \mathbf{tr}\left[\mathbf{B}_{n}\mathbf{H}_{n}^{-1}\right]+O_p\left(n^{-1/2}\right). 
\end{equation}

Thus, with (\ref{vtic prepared}) and (\ref{vb consistent trace hn and bn bar}),
\begin{equation}
\label{final proof of vtic asymptotic estimator}
	\begin{aligned}
		& E_{\mathbf{y}}\left[E_{\mathbf{y}_{\text {rep}}}\left(-2 \ln p\left(\mathbf{y}_{\text {rep}} | \overline{\boldsymbol{\theta}}^{VB}(\mathbf{y})\right)\right)\right] \\
		= & E_{\mathbf{y}}\left[-2 \ln p\left(\mathbf{y} | \overline{\boldsymbol{\theta}}^{VB}(\mathbf{y})\right) - 2 \mathbf{t r}\left[\mathbf{B}_n \mathbf{H}_n^{-1}\right]\right]+o(1)\\
        = & E_{\mathbf{y}}\left[-2 \ln p\left(\mathbf{y} | \overline{\boldsymbol{\theta}}^{VB}(\mathbf{y})\right) - 2 \mathbf{tr}\left[
        \mathbf{\bar{\Omega}}_{n}\left(
        {\mbox{\boldmath${\bar{					\theta}}$}}^{VB}\right) \left(
        \mathbf{\bar{H}}_{n}\left( {
        \mbox{\boldmath${\bar{\theta}}$}}^{VB}\right) \right) ^{-1}\right]+o_p(1)\right]+o(1)\\
        = & E_{\mathbf{y}}\left[-2 \ln p\left(\mathbf{y} | \overline{\boldsymbol{\theta}}^{VB}(\mathbf{y})\right) - 2 \mathbf{tr}\left[
        \mathbf{\bar{\Omega}}_{n}\left(
        {\mbox{\boldmath${\bar{					\theta}}$}}^{VB}\right) \left(
        \mathbf{\bar{H}}_{n}\left( {
        \mbox{\boldmath${\bar{\theta}}$}}^{VB}\right) \right) ^{-1}\right]\right]+o(1),
	\end{aligned}
\end{equation}
which means $\text{VDIC}_M^{k}$ is an asymptotically unbaised estimator of $Risk(d_{k^1})$ up to a constant.

\subsubsection{Proof of Theorem 4.2}
%\protect \protect\ref{theom2}}
Proof of Theorem 4.2
%\ref{theom2}
is similar like proof of Theorem 4.1.
%(\ref{theom vtic}). 
In Theorem 3.2
%(\ref{riskvpic})
\begin{eqnarray*}
&&E_{\mathbf{y}}E_{\mathbf{y}_{rep}}\left( -2\ln p^{VB}\left( \mathbf{y}%
_{rep}|\mathbf{y}\right) \right) \\
&=&E_{\mathbf{y}}\left( -2\ln p\left( \mathbf{y}|\widehat{{%
\mbox{\boldmath${\theta}$}}}_{n}\left( \mathbf{y}\right) \right) \right)
+\ln \left( \left\vert -\mathbf{H}_{n}\left( -\mathbf{H}_{n}^{d}\right)
^{-1}+\mathbf{I}_{n}\right\vert \right) +\mathbf{tr}\left[ \mathbf{B}%
_{n}\left( -\mathbf{H}_{n}\right) ^{-1}\right] \\
&&-\mathbf{tr}\left[ \left( -\mathbf{H}_{n}+\left( -\mathbf{H}%
_{n}^{d}\right) \right) ^{-1}\left( \mathbf{B}_{n}+\left( -\mathbf{H}%
_{n}^{d}\right) \mathbf{C}_{n}\left( -\mathbf{H}_{n}^{d}\right) \right) %
\right] +\mathbf{tr}\left[ \left( -\mathbf{H}_{n}^{d}\right) \mathbf{C}_{n}%
\right] +o\left( 1\right),
\end{eqnarray*}
where $\mathbf{C}_{n}=\mathbf{H}_{n}^{-1}\mathbf{B}_{n}\mathbf{H}_{n}^{-1}$, 
$\mathbf{H}_{n}^{d}$ is a diagonal matrix with the same diagonal elements as
in $\mathbf{H}_{n}$.

Because both $$\mathbf{\bar{\Omega}}_{n}\left( {%
\mbox{\boldmath${\bar{\theta}}$}}^{VB}\right) \quad \text{and} \quad \mathbf{\bar{H}}_{n}\left( {\ \mbox{\boldmath${\bar{\theta}}$}}^{VB}\right) $$
are the consistent estimator of both 
$\mathbf{B}_{n}$ and $\mathbf{H}_{n}$, proved in Theorem 4.1, 
%\ref{theom vtic}, 
so we derive a consistent estimator that
\begin{equation}
\label{vb consistent cn bar}
   \mathbf{\hat{C}}_{n}\left( {%
\mbox{\boldmath${\bar{\theta}}$}}^{VB}\right) =\left( \mathbf{\bar{H}}%
_{n}\left( {\mbox{\boldmath${\bar{\theta}}$}}^{VB}\right) \right) ^{-1}%
\mathbf{\bar{\Omega}}_{n}\left( {\mbox{\boldmath${\bar{\theta}}$}}%
^{VB}\right) \left( \mathbf{\bar{H}}_{n}\left( {\mbox{\boldmath${\bar{%
\theta}}$}}^{VB}\right) \right) ^{-1}=\mathbf{C}_{n}+ O_p(n^{-1}) 
\end{equation}
Mimicking the proof of equation (\ref{t2 scaling inequality term1}), % in proof of Theorem \ref{riskvtic}, 
\begin{equation}
\label{vb consistent hn bar diag}
    \mathbf{\bar{H}}_{n}^d\left( {\ \mbox{\boldmath${\bar{\theta}}$}}^{VB}\right) = \mathbf{H}_{n}^d({
\mbox{\boldmath${\theta}$}}_{n}^{p}) + O_p\left(n^{-1/2}\right).
\end{equation}
is derived.

Combined with (\ref{t1 term finally}), (\ref{vb consistent omega bar}), (\ref{vb consistent hn bar}), (\ref{vb consistent trace hn and bn bar}), (\ref{vb consistent cn bar}) and (\ref{vb consistent hn bar diag}),
\begin{equation}
    \begin{aligned}
        &E_{\mathbf{y}}E_{\mathbf{y}_{rep}}\left( -2\ln p^{VB}\left( \mathbf{y}%
        _{rep}|\mathbf{y}\right) \right) \\
        &=E_{\mathbf{y}}\left( -2\ln p\left( \mathbf{y}|\widehat{{%
        \mbox{\boldmath${\theta}$}}}_{n}\left( \mathbf{y}\right) \right) \right)
        +\ln \left( \left\vert -\mathbf{H}_{n}\left( -\mathbf{H}_{n}^{d}\right)
        ^{-1}+\mathbf{I}_{n}\right\vert \right) +\mathbf{tr}\left[ \mathbf{B}%
        _{n}\left( -\mathbf{H}_{n}\right) ^{-1}\right] \\
        &-\mathbf{tr}\left[ \left( -\mathbf{H}_{n}+\left( -\mathbf{H}%
        _{n}^{d}\right) \right) ^{-1}\left( \mathbf{B}_{n}+\left( -\mathbf{H}%
        _{n}^{d}\right) \mathbf{C}_{n}\left( -\mathbf{H}_{n}^{d}\right) \right) %
        \right] +\mathbf{tr}\left[ \left( -\mathbf{H}_{n}^{d}\right) \mathbf{C}_{n}%
        \right] +o\left( 1\right)\\
        & = E_{\mathbf{y}}\left(\text{VPIC}\right) +o(1),
    \end{aligned}
\end{equation}
where $\text{VPIC}=-2\ln p(\mathbf{y}|{\mbox{\boldmath${\bar{\theta}}$}}%
^{VB})+2 P_{VPIC}$, with 
\begin{equation*}
\begin{aligned} P_{VPIC} &= \frac{1}{2}\mathbf{tr}\left[
\mathbf{\bar{\Omega}}_{n}\left(
{\mbox{\boldmath${\bar{					\theta}}$}}^{VB}\right) \left(
-\mathbf{\bar{H}}_{n}\left( {
\mbox{\boldmath${\bar{\theta}}$}}^{VB}\right) \right) ^{-1}\right] +
\frac{1}{2}\ln \left( \left\vert \left( -\mathbf{\bar{H}}_{n}\left(
{\mbox{\boldmath${\bar{					\theta}}$}}^{VB}\right) \right) \left( -\mathbf{\bar{H}}_{n}^d\left(
{\mbox{\boldmath${\bar{					\theta}}$}}^{VB}\right) \right)^{-1} +\mathbf{I}_{n}\right\vert \right) \\
-&\frac{1}{2}\mathbf{tr}\left[\begin{array}{c}
	\left(
	-\mathbf{\bar{H}}_{n}\left(
	{\mbox{\boldmath${				\bar{\theta}}$}}^{VB}\right) +\left( -\mathbf{\bar{H}}_{n}^d\left(
	{\mbox{\boldmath${\bar{					\theta}}$}}^{VB}\right) \right) \right) ^{-1} \\
	 \times \left( \mathbf{\bar{\Omega}}_{n}\left(
	{\mbox{\boldmath${\bar{\theta}}$} }^{VB}\right)+\left( -\mathbf{\bar{H}}_{n}^d\left(
	{\mbox{\boldmath${\bar{					\theta}}$}}^{VB}\right) \right) \mathbf{\hat{C}}_{n}\left( {
		\mbox{\boldmath${\bar{\theta}}$}}^{VB}\right) \left( -\mathbf{\bar{H}}_{n}^d\left(
		{\mbox{\boldmath${\bar{					\theta}}$}}^{VB}\right) \right) \right)
\end{array}\right]\\ 
+&\frac{1}{2}\mathbf{tr}\left[ \left( -\mathbf{\bar{H}}_{n}^d\left(
{\mbox{\boldmath${\bar{					\theta}}$}}^{VB}\right) \right) \mathbf{\hat{C}}_{n}\left( {
\mbox{\boldmath${\bar{\theta}}$}}^{VB}\right) \right] . \end{aligned}
\end{equation*}	
\newpage
\section{Analytical expression of VB for used parametric model}

\subsection{Mean-Field VB for linear regression with normal error}
\label{mean feild linear}
As literatures shows, for parameter $\theta_i \subset \boldsymbol{\theta}$,
one can derive the mean-field VB posterior 
\begin{equation*}
\log q\left(\theta_i\right) \propto E_{q_{-i}}[\log p(\theta_{i}\mid%
\boldsymbol{\theta}_{-i}, \mathbf{y})],
\end{equation*}
which can be transformed by Gibbs sampling using full conditional
distributions. By setting priors in main paper, we write the full conditional
density of $\beta$ 
\begin{equation*}
\begin{aligned} & \log P(\beta \mid Y, X, h) \propto \log P(Y \mid X, \beta,
h)+\log P(\beta \mid h) \\ & \propto-\frac{h}{2}\left(Y^{\prime}
Y-2\left(X^{\prime} Y+\tilde{V}^{-1} \tilde{\mu}\right)^{\prime}
\beta+\beta^{\prime}\left(X^{\prime} X+\tilde{V}^{-1}\right)
\beta+\tilde{\mu}^{\prime} \tilde{V}^{-1} \tilde{\mu}\right) \\ & =
N\left(\mu_\beta, V_\beta\right) \\ & \mu_\beta =\left(X^{\prime}
X+\tilde{V}^{-1}\right)^{-1}\left(X^{\prime} Y+\tilde{V}^{-1}
\tilde{\mu}\right) \\ & V_\beta = h^{-1}\left(X^{\prime}
X+\tilde{V}^{-1}\right)^{-1}, \end{aligned}
\end{equation*}
and $h$ 
\begin{equation*}
\begin{aligned} & \log P(h \mid Y, X, \beta) \propto \log P(Y \mid X, \beta,
h)+\log P(h) \\ & \propto -\left(b+\frac{1}{2} Y^{\prime} Y-Y^{\prime} X
\beta+\frac{1}{2} \beta^{\prime} X^{\prime} X \beta\right) h
+\left(a-1+\frac{N}{2}\right) \log h \\ & = \operatorname{Gamma}\left(a_h,
b_h\right) \\ & a_h = a+\frac{N}{2} \\ & b_h = b+\frac{1}{2}(Y-X
\beta)^{\prime}(Y-X \beta), \end{aligned}
\end{equation*}
the optimal VB posterior of $\beta$ and $h$ that approximate the true
posterior $p\left(\beta,h\mid y\right)$ of linear regression model by
coordinate ascent variational bayes, having the same form as prior that 
\begin{equation*}
q(\beta, h)=q(\beta)q(h)
\end{equation*}
with 
\begin{equation}
\begin{aligned} q(\beta) & \sim N\left(\mu_{\beta}^*, V_{\beta}^*\right) \\
\mu_{\beta}^* & \leftarrow\left(X^{\prime}
X+\tilde{V}^{-1}\right)^{-1}\left[\tilde{V}^{-1} \tilde{\mu}+X^{\prime}
Y\right] \\ V_{\beta}^* & \leftarrow\left(X^{\prime}
X+\tilde{V}^{-1}\right)^{-1} \frac{b_h^*}{a_h^*} \\ q(h) & \sim
Gamma\left(a_h^*, b_h^*\right) \\ a_h^* & \leftarrow \frac{N}{2}+a \\ b_h^*
& \leftarrow b+\frac{1}{2} Y^{\prime} Y-Y^{\prime} X
\mu_{\beta}^*+\frac{1}{2} \operatorname{trace}\left(X^{\prime}
X\left(V_{\beta}^*+\mu_{\beta}^* \mu_{\beta}^{* \prime}\right)\right).
\end{aligned}  \label{VB for linear normal terms}
\end{equation}

For linear regression model, the parameters we are interested about are $%
\boldsymbol{\theta}=\left(\beta^{\prime}, h\right)^{\prime}$ and denote $%
L\left(\boldsymbol{y}\mid \boldsymbol{\theta}\right)$ as logrithm likelihood
function. To derive $\mathrm{IC}_k^{VB}$ of candidate model $k=1,\dots,K$,
we need consistent estimator of $\mathbf{B}_n\left(\boldsymbol{\theta}%
_n^p\right)$ 
\begin{equation*}
\overline{\boldsymbol{\Omega}}_n\left(\bar{\boldsymbol{\theta}}%
_k^{VB}\right)=\frac{1}{N} \sum_{t=1}^N \mathbf{s}_t\left(\bar{\boldsymbol{%
\theta}}_k^{VB}\right) \mathbf{s}_t\left(\bar{\boldsymbol{\theta}}%
_k^{VB}\right)^{\prime}
\end{equation*}
where 
\begin{equation*}
\mathbf{s}_i\left(\boldsymbol{\theta}\right) = \left(\frac{\partial
L\left(Y_i\mid \boldsymbol{\theta}\right)^{\prime}}{\partial \beta}, \frac{%
\partial L\left(Y_i\mid \boldsymbol{\theta}\right)}{\partial h}%
\right)^{\prime}
\end{equation*}
with 
\begin{equation*}
\begin{aligned} \frac{\partial L\left(Y_i\mid
\boldsymbol{\theta}\right)}{\partial \beta} &= h \left(Y_i X_i - X _i
X_i^{\prime} \beta \right)\\ \frac{\partial L\left(Y_i\mid
\boldsymbol{\theta}\right)}{\partial h} &= \frac{1}{2h} -
\frac{1}{2}\left(Y_i-X_i^{\prime}\beta\right)^2 \end{aligned}
\end{equation*}
and consistent estimator of $\mathbf{H}_n\left(\boldsymbol{\theta}%
_n^p\right) $ 
\begin{equation*}
\begin{aligned}
\overline{\mathbf{H}}_n\left(\overline{\boldsymbol{\theta}}_k^{V B}\right)
&= \frac{1}{N} \sum_{t=1}^N \mathbf{h}_t(\overline{\boldsymbol{\theta}}_k^{V
B}),\\ \sum_{t=1}^N \mathbf{h}_t(\boldsymbol{\theta}) & =
\left(\begin{array}{cc} \frac{\partial^2 L\left(Y\mid
\boldsymbol{\theta}\right)}{\partial \beta \partial \beta^{\prime}} &
\frac{\partial^2 L\left(Y\mid \boldsymbol{\theta}\right)}{\partial \beta
\partial h} \\ \frac{\partial^2 L\left(Y\mid
\boldsymbol{\theta}\right)}{\partial h \partial \beta^{\prime}} &
\frac{\partial^2 L\left(Y\mid \boldsymbol{\theta}\right)}{\partial h^2}
\end{array}\right) \end{aligned}
\end{equation*}
where 
\begin{equation*}
\begin{aligned} \frac{\partial^2 L\left(Y\mid
\boldsymbol{\theta}\right)}{\partial \beta \partial \beta^{\prime}} & =
\sum_{i=1}^{N} \left(-h X_i X_i^{\prime} \right)= -h X^{\prime} X\\
\frac{\partial^2 L\left(Y\mid \boldsymbol{\theta}\right)}{\partial \beta
\partial h} & = \sum_{i=1}^{N} \left(Y_i X_i - X_i X_i^{\prime}\beta
\right)= X^{\prime}Y-X^{\prime}X\beta \\ \frac{\partial^2 L\left(Y\mid
\boldsymbol{\theta}\right)}{\partial h \partial \beta^{\prime}} & =
\sum_{i=1}^{N} \left(Y_i X_i - X_i X_i^{\prime}\beta \right)^{\prime} =
Y^{\prime}X - \beta^{\prime}X^{\prime}X \\ \frac{\partial^2 L\left(Y\mid
\boldsymbol{\theta}\right)}{\partial h^2} & = \sum_{i=1}^{N}
\left(-\frac{1}{2h^2}\right)= -\frac{N}{2}\frac{1}{h^2} \end{aligned}
\end{equation*}
then we have the consistent estimator of $\mathbf{C}_n$ where 
\begin{equation*}
\hat{\mathbf{C}}_n\left(\overline{\boldsymbol{\theta}}_k^{V B}\right)=\left(%
\overline{\mathbf{H}}_n\left(\overline{\boldsymbol{\theta}}_k^{V
B}\right)\right)^{-1} \overline{\boldsymbol{\Omega}}_n\left(\overline{%
\boldsymbol{\theta}}_k^{V B}\right)\left(\overline{\mathbf{H}}_n\left(%
\overline{\boldsymbol{\theta}}_k^{V B}\right)\right)^{-1}.
\end{equation*}

\subsection{Mean-Field VB for probit regression}
\label{mean feild probit}
We use mean-field VB algorithm for the probit model, for all observed i.i.d.
data 
\begin{equation*}
Y=\left[%
\begin{array}{c}
Y_1 \\ 
Y_2 \\ 
\vdots \\ 
Y_N%
\end{array}%
\right] \quad X=\left[%
\begin{array}{cccc}
1 & x_{12} & \ldots & x_{1 p} \\ 
1 & x_{22} & \ldots & x_{2 p} \\ 
\vdots & \vdots & \vdots & \vdots \\ 
1 & x_{N 2} & \ldots & x_{N p}%
\end{array}%
\right],
\end{equation*}
we have linear predictor based on vector $X_i$ 
\begin{equation*}
Z_i=X_i^{\prime} \beta,
\end{equation*}
and we choose the probit link as link function 
\begin{equation*}
\Phi^{-1}\left(p_i\right)=Z_i,
\end{equation*}
the inverse of the link function $\Phi\left(\cdot\right)$ is the cdf of
standard normal distribution, let $g^{-1}\left(\cdot\right)=\Phi\left(\cdot%
\right)$, we will have $\mathbb{E}[Y \mid X]=g^{-1}\left(X^{\prime} \beta
\right)$, the likelihood function of probit model is 
\begin{equation}
Y_i \mid X_i \overset{i . i . d .}{\sim} Bernoulli\left(\Phi%
\left(X_i^{\prime} \beta \right)\right),  \label{llk of probit model}
\end{equation}
the likelihood function of all the observed data is 
\begin{equation*}
f(Y \mid
\beta)=\prod_{i=1}^N\left(\Phi\left(Z_i\right)\right)^{Y_i}\left(1-\Phi%
\left(Z_i\right)\right)^{1-Y_i}
\end{equation*}
which $Y_i$ equals 0 or 1. In Bayesian framework, we will posit a normal
prior $\beta \sim N (0,\tilde{V})$, To facilitate computation, it is common
to augment the model by introducing $N$ latent variables $\boldsymbol{z}%
=\left(z_1, \ldots, z_N\right)$ with latent distribution 
\begin{equation}
z_i \mid \beta \overset{i . i . d .}{\sim} N \left(X_i^{\prime} \beta,1
\right),  \label{probit latent distribution}
\end{equation}
so that $p\left(Y_i \mid z_i\right)=I\left(z_i \geq 0\right)^{Y_i}
I\left(z_i<0\right)^{1-Y_i}$. Under the model augmentation, we can write the
logarithm of the joint posterior distribution over the parameter-latent pair 
$(\beta,\boldsymbol{z})$ as 
\begin{equation*}
\begin{aligned} \log p\left(\boldsymbol{z}, \beta \mid Y\right) & =
\sum_{i=1}^N\left[Y_i \log I\left(z_i \geq 0\right)+\left(1-Y_i\right) \log
I\left(z_i<0\right)\right] \\ - & \frac{1}{2}\beta^{\prime} \tilde{V}^{-1}
\beta -\frac{1}{2}\left(\boldsymbol{z}-X
\beta\right)^{\prime}\left(\boldsymbol{z}-X \beta\right)+\text {const}.
\end{aligned}
\end{equation*}
With mean-field VB updating formula, we have 
\begin{equation}
q\left(z_i\right) \sim 
\begin{cases}
N_{+}\left(X_i^{\prime} E_q[\beta], 1\right) & Y_i=1 \\ 
N_{-}\left(X_i^{\prime} E_q[\beta], 1\right), & Y_i=0%
\end{cases}%
.  \label{cavb of per-data par in vb probit}
\end{equation}
where $N_{+}\left(\cdot\right)$ and $N_{-}\left(\cdot\right)$ denote the
normal distributions truncated to positive and negative part, respectively.
For $\beta$, we have 
\begin{equation}
q\left(\beta\right) \sim N\left(\left(X^{\prime} X+\tilde{V}%
^{-1}\right)^{-1} X^{\prime} E_q[\boldsymbol{z}],\left(X^{\prime} X+\tilde{V}%
^{-1}\right)^{-1}\right)  \label{cavb of global par in vb probit}
\end{equation}
Both VB optimal distributions of $\beta$ and $\boldsymbol{z}$ are normal or
truncated normal, with fixed variance. Let $\mu_{\beta}^*=E_q[\beta]$ and $%
\mu_{\boldsymbol{z}}^*=E_q[\boldsymbol{z}]$ as follows. 
\begin{equation}
\begin{aligned} \mu_{\beta}^* & =\left(X^{\prime}
X+\tilde{V}^{-1}\right)^{-1} X^{\prime} \mu_{\boldsymbol{z}}^* \\
\mu_{z_i}^* & =X_i^{\prime} \mu_{\beta}^*+\frac{\phi\left(X_i^{\prime}
\mu_{\beta}^* \right)}{\Phi\left(X_i^{\prime}
\mu_{\beta}^*\right)^{Y_i}\left[\Phi\left(X_i^{\prime}
\mu_{\beta}^*\right)-1\right]^{1-Y_i}} \end{aligned}
\label{cavb para of logit}
\end{equation}
where $\phi$ is the pdf of standard normal distribution. The optimal ELBO
has an analytical form as 
\begin{equation}
\begin{aligned} ELBO & = \sum_{i=1}^{N}\left[Y_i \log \Phi
\left(X_i^{\prime} \mu_{\beta}^*\right)+\left(1-Y_i\right) \log
\left(1-\Phi\left(X_i^{\prime} \mu_{\beta}^*\right)\right)\right] \\ - &
\frac{1}{2} \mu_{\beta}^{*\prime} \tilde{V}^{-1} \mu_{\beta}^*-\frac{1}{2}
\log \operatorname{det}\left(\tilde{V} X^{\prime} X +I_d\right) \end{aligned}
\label{vb elbo of probit}
\end{equation}
As discussed in the literature, one can use this ELBO value as the criterion
to conduct variable selection by selecting a subset of variables that
maximizes it.

The interested parameters $\boldsymbol{\theta}$ in this model is $\beta$, to
derive $\mathrm{IC}_k^{VB}$, we need consistent estimator of $\mathbf{B}%
_n\left(\boldsymbol{\theta}_n^p\right)$ 
\begin{equation*}
\overline{\boldsymbol{\Omega}}_n\left(\bar{\boldsymbol{\theta}}%
_k^{VB}\right)=\frac{1}{N} \sum_{t=1}^N \mathbf{s}_t\left(\bar{\boldsymbol{%
\theta}}_k^{VB}\right) \mathbf{s}_t\left(\bar{\boldsymbol{\theta}}%
_k^{VB}\right)^{\prime}
\end{equation*}
where 
\begin{equation*}
\mathbf{s}_i\left(\boldsymbol{\theta}\right) = \frac{\phi\left(X_i^{\prime}
\beta\right)}{\Phi\left(X_i^{\prime} \beta\right)\left[1-\Phi\left(X_i^{%
\prime} \beta\right)\right]}\left[Y_i-\Phi\left(X_i^{\prime} \beta\right)%
\right] X_i
\end{equation*}
and consistent estimator of $\mathbf{H}_n\left(\boldsymbol{\theta}%
_n^p\right) $ 
\begin{equation*}
\overline{\mathbf{H}}_n\left(\overline{\boldsymbol{\theta}}_k^{V B}\right) = 
\frac{1}{N} \sum_{t=1}^N \mathbf{h}_t(\overline{\boldsymbol{\theta}}_k^{V B})
\end{equation*}
where 
\begin{equation*}
\begin{aligned} \sum_{t=1}^{N}\mathbf{h}_t(\boldsymbol{\theta}) & =
-\sum_{i=1}^N \phi\left(X_i^{\prime} \beta\right)\left[Y_i
\frac{\phi\left(X_i^{\prime} \beta\right)+X_i^{\prime} \beta
\Phi\left(X_i^{\prime} \beta\right)}{\Phi\left(X_i^{\prime}
\beta\right)^2}\right] X_i X_i^{\prime}\\ & -\sum_{i=1}^N
\phi\left(X_i^{\prime} \beta\right)\left[\left(1-Y_i\right)
\frac{\phi\left(X_i^{\prime} \beta\right)-X_i^{\prime}
\beta\left(1-\Phi\left(X_i^{\prime}
\beta\right)\right)}{\left[1-\Phi\left(X_i^{\prime}
\beta\right)\right]^2}\right] X_i X_i^{\prime} \end{aligned}
\end{equation*}
then we have the consistent estimator of $\mathbf{C}_n$ where 
\begin{equation*}
\hat{\mathbf{C}}_n\left(\overline{\boldsymbol{\theta}}_k^{V B}\right)=\left(%
\overline{\mathbf{H}}_n\left(\overline{\boldsymbol{\theta}}_k^{V
B}\right)\right)^{-1} \overline{\boldsymbol{\Omega}}_n\left(\overline{%
\boldsymbol{\theta}}_k^{V B}\right)\left(\overline{\mathbf{H}}_n\left(%
\overline{\boldsymbol{\theta}}_k^{V B}\right)\right)^{-1}
\end{equation*}

\end{document}